\documentclass[journal]{IEEEtran} 
\usepackage{amsmath, amsfonts, amssymb}
\usepackage{amsfonts}
\usepackage{amssymb}
\usepackage{amsthm}
\usepackage[dvips]{graphicx}
\usepackage{verbatim}
\usepackage{setspace}
\usepackage{bm}
\usepackage{algorithmic} 
\usepackage[ruled,vlined]{algorithm2e}
\usepackage{cite}
\usepackage{url}

\usepackage{breakurl}
\usepackage{changepage}
\usepackage{pdfpages}
\usepackage{color}
\usepackage{blkarray, bigstrut}
\usepackage{caption}
\usepackage{subcaption}

\usepackage[labelformat=simple]{subcaption}

\newtheorem{lemma}{Lemma}
\newtheorem{proposition}{Proposition}

\newcommand{\highlightcolor}{black}

\IEEEoverridecommandlockouts

\begin{document}

\title{
    {Power Measurement Enabled Channel Autocorrelation Matrix Estimation for IRS-Assisted Wireless Communication}
    \vspace{-5pt}
}

\author{Ge Yan, 
        Lipeng Zhu,~\IEEEmembership{Member,~IEEE,}
        Rui Zhang,~\IEEEmembership{Fellow,~IEEE}
        \vspace{-25pt}

\thanks{Part of this work was presented at the IEEE Global Communications
Conference Workshops 2023, Kuala Lumpur, Malaysia~\cite{ref:my-conf-ver}.}
\thanks{G. Yan is with the NUS Graduate School, National University of Singapore, Singapore 119077, and also with the Department of Electrical and Computer Engineering, National University of Singapore, Singapore 117583 (e-mail: geyan@u.nus.edu). }
\thanks{L. Zhu is with the Department of Electrical and Computer Engineering, National University of Singapore, Singapore 117583 (zhulp@nus.edu.sg). }
\thanks{R. Zhang is with The Chinese University of Hong Kong, Shenzhen, and Shenzhen Research Institute of Big Data, Shenzhen, China 518172 (e-mail: rzhang@cuhk.edu.cn). 
He is also with the Department of Electrical and Computer Engineering, National University of Singapore, Singapore 117583 (e-mail: elezhang@nus.edu.sg). }
}


\maketitle

\IEEEpeerreviewmaketitle

\begin{abstract}
    By reconfiguring wireless channels via passive signal reflection, intelligent reflecting surface (IRS) can bring significant performance enhancement for wireless communication systems. 
    However, such performance improvement generally relies on the knowledge of channel state information (CSI) for IRS-involved links. 
    Prior works on IRS CSI acquisition mainly estimate IRS-cascaded channels based on the extra pilot signals received at the users/base station (BS) with time-varying IRS reflections, which, however, needs to modify the existing channel training/estimation protocols of wireless systems. 
    To address this issue, we propose in this paper a new channel estimation scheme for IRS-assisted communication systems based on the received signal power measured at the user terminal, which is practically attainable without the need of changing the current protocol. 
    Due to the lack of signal phase information in measured power, the autocorrelation matrix of the BS-IRS-user cascaded channel is estimated by solving an equivalent rank-minimization problem. 
    To this end, a low-rank-approaching (LRA) algorithm is proposed by employing the fractional programming and alternating optimization techniques. 
    To reduce computational complexity, an approximate LRA (ALRA) algorithm is also developed. 
    Furthermore, these two algorithms are extended to be robust against the receiver noise and quantization error in power measurement. 
    Simulation results are provided to verify the effectiveness of the proposed channel estimation algorithms as well as the IRS passive reflection design based on the estimated channel autocorrelation matrix. 
\end{abstract}
\vspace{-4pt}

\vspace{-16pt}
\section{INTRODUCTION}\label{sec:introduction}
    \IEEEPARstart{I}{n} recent years, intelligent reflecting surface (IRS) has received great attention due to its appealing capability of reconfiguring wireless channels. 
    By applying tunable phase shifts to incident wireless signals, IRS can effectively control their propagation channels and thereby significantly enhance the wireless communication performance, such as spectral/energy efficiency and transmission reliability~\cite{ref:PIEEE-IRS6G, ref:IRSTutorial, ref:my-discrete-bf, ref:ma-bf-3d-cover}. 
    Given such benefits as well as its high deployment flexibility, low hardware cost, and low power consumption, IRS has been identified as a key enabling technology for future wireless networks such as 6G~\cite{ref:PIEEE-IRS6G, ref:IRSTutorial}. 
    However, to reap the high performance gain by IRS, it is essential to acquire the channel state information (CSI) for the IRS channels with its assisting base station (BS) and users, which is practically difficult due to the following reasons. 
    On one hand, the passive IRS is not equipped with wireless transceivers at its reflecting elements, making it impossible to estimate the BS-IRS and IRS-user channels separately. 
    Instead, only the cascaded BS-IRS-user/user-IRS-BS channel can be estimated at the user/BS~\cite{ref:IRS-Survey}. 
    On the other hand, to compensate for the significant product-distance path loss of the IRS-cascaded link, the number of IRS reflecting elements needs to be sufficiently large in practice, e.g., several tens or even hundreds~\cite{ref:Physics-based-modeling-IRS,ref:power-scaling-law-IRS}. 
    This results in high-dimensional IRS channel vectors/matrices that incur prohibitive overhead and computational complexity to estimate. 

    To tackle the above challenges, extensive studies have been devoted to the cascaded channel estimation for IRS-assisted wireless communication systems, while these works mainly adopted conventional pilot-based channel estimation approaches by sending additional pilot signals with concurrent time-varying IRS reflections~\cite{ref:low-comp-ON-OFF, ref:cascaded-ON-OFF, ref:optimal-CE-MinVar, ref:IRS-OFDM-CE-DFT, ref:ce-bf-discrete, ref:IRS-assisted-OFDMA, ref:IRS-meet-OFDM, ref:CS-IRS-mmW, ref:low-rank-CE-IRS-mmW-OFDM, ref:CS-IRS-AtomicNormMin, ref:CE-IRS-double-sparsity, ref:deep-residual-ce}. 
    For example, by switching on only one reflecting element at one time, the IRS-cascaded channel for each element was estimated based on the received pilot signals at the user in~\cite{ref:low-comp-ON-OFF, ref:cascaded-ON-OFF}. 
    To exploit the array gain of IRS for channel estimation, reflection codebooks for channel estimation were designed in~\cite{ref:IRS-OFDM-CE-DFT, ref:ce-bf-discrete, ref:IRS-assisted-OFDMA} with all reflecting elements switched on. 
    In particular, the discrete Fourier transform (DFT)-based IRS reflection codebook was shown to yield the minimum-mean-square-error (MMSE) estimation of the cascaded channel in IRS-assisted multiple-input single-output (MISO) systems~\cite{ref:optimal-CE-MinVar} and IRS-enhanced orthogonal frequency division multiplexing (OFDM) systems~\cite{ref:IRS-OFDM-CE-DFT}. 
    In~\cite{ref:IRS-assisted-OFDMA}, the DFT-based codebook was employed for reflection training in an IRS-assisted multi-user orthogonal frequency division multiple access (OFDMA) system. 
    Furthermore, under the practical discrete-phase-shift constraint on IRS reflection coefficients, a Hadamard matrix-based IRS reflection pattern was proposed in~\cite{ref:ce-bf-discrete}, while more sophisticated codebooks were designed in~\cite{ref:Physics-based-modeling-IRS, ref:quadratic-phase-codebook, ref:optimized-codebook-IRS} to achieve higher training efficiency. 
    Besides, the authors in~\cite{ref:IRS-OFDM-CE-DFT, ref:ce-bf-discrete, ref:IRS-assisted-OFDMA, ref:IRS-meet-OFDM} proposed to group adjacent IRS reflecting elements into sub-surfaces so that only the effective cascaded channel for each sub-surface needs to be estimated, thus reducing the number of training pilots. 
    In addition, various compressed sensing algorithms were developed for IRS channel estimation by utilizing the sparsity of the channel paths in the angular domain~\cite{ref:CS-IRS-mmW, ref:low-rank-CE-IRS-mmW-OFDM, ref:CS-IRS-AtomicNormMin, ref:CE-IRS-double-sparsity}. 
    Moreover, the deep residual network was employed in~\cite{ref:deep-residual-ce} to refine the least-sqaure (LS) estimation of IRS channels. 
    
    In the aforementioned works, IRS-involved CSI is estimated based on the received complex-valued pilot signals at the users/BS. 
    However, in the protocol of existing wireless communication systems, such as 4G/5G~\cite{ref:3gpp:38.211}, the pilot signals are dedicated to estimating the BS-user direct channels only. 
    As such, substantially additional pilots are required to estimate the new IRS-cascaded CSI, which thus requires significant modifications of the existing channel estimation/training protocols. 
    To address this issue, IRS reflection designs based on the received signal power measurement at the users have been proposed~\cite{ref:CSM, ref:RFocus}, which do not require additional pilot signals for explicit IRS CSI estimation. 
    As user power measurement is commonly adopted and easy to obtain in existing wireless systems, such as reference signal received power (RSRP)~\cite{ref:3gpp:36.214}, this approach can be practically implemented without any change of the current protocols. 
    For example, for the conditional sample mean (CSM) method proposed in~\cite{ref:CSM}, the received signal power at the user was modelled as a random variable and a large number of IRS reflection patterns are randomly generated for user power measurement. 
    After the user measures the received signal power values for all IRS reflection patterns, each IRS reflecting element sets the reflection coefficient as the one that achieves the maximum received power expectation conditioned on its value. 
    The majority voting algorithm proposed in~\cite{ref:RFocus} employed a similar idea but the received signal strength indicator (RSSI) was used instead of RSRP. 
    Both methods in~\cite{ref:CSM, ref:RFocus} directly design the IRS reflection coefficients based on the power measurement without estimating the channel. 
    However, to obtain the optimal IRS reflection performance, an excessively large number of IRS training reflections/power measurement are generally required (in the quadratic order of the number of IRS reflecting elements). 
    Therefore, these methods may still be time-consuming for practical implementation. 

    It is worth noting that the methods in~\cite{ref:RFocus, ref:CSM} did not fully exploit the power measurement to obtain explicit CSI of IRS-cascaded channels, thus resulting in their high overhead and low efficiency. 
    To improve the existing IRS channel estimation/reflection designs based on user power measurement, this paper proposes a new channel autocorrelation matrix estimation scheme. 
    Specifically, since there is no signal phase information in received power measurement, the IRS-cascaded channel vector cannot be completely recovered, while its autocorrelation matrix can be (uniquely) estimated. 
    By equivalently transforming the channel autocorrelation estimation problem into a rank-minimization problem, a low-rank-approaching (LRA) algorithm is proposed to recover the channel autocorrelation matrix given the received signal power measurement with time-varying random IRS reflections. 
    In particular, the LRA algorithm converts the rank-minimization problem into a fractional programming problem and alternating optimization is applied to obtain its solution. 
    To reduce the computational complexity of the LRA algorithm, an approximate LRA (ALRA) algorithm is also proposed, where the fractional programming is approximated by low-complexity quadratic programming. 
    Furthermore, the robust extensions of both LRA and ALRA algorithms are designed by considering the effects of receiver noise and quantization error in power measurement. 
    The convergence and estimation accuracy of the proposed algorithms are verified via simulations. 
    Besides, it is validated that with the IRS reflection designs based on the estimated channel autocorrelation matrix, significantly higher passive reflection gains can be achieved compared to other benchmark schemes, which demonstrates the effectiveness of the proposed channel estimation methods for IRS-assisted wireless communication systems based on user power measurement. 

    \textit{Notations:} 
    Boldface letters refer to vectors (lower case) or matrices (upper case). For square matrix $\boldsymbol{A}$, $\text{tr}(\boldsymbol{A})$ denotes its trace and $\boldsymbol{A}^{-1}$ denotes its inverse matrix. 
    For matrix $\boldsymbol{B}$, let $\boldsymbol{B}^{T}$, $\boldsymbol{B}^{H}$, $\text{rank}(\boldsymbol{B})$, $\boldsymbol{B}^{\dagger}$, $\|\boldsymbol{B}\|_F$, and $\text{vec}(\boldsymbol{B})$ denote the transpose, conjugate transpose, rank, pseudo inverse, Frobenius norm, and vectorization of $\boldsymbol{B}$, respectively. 
    $\boldsymbol{I}_N$ denotes an $N\times N$ identity matrix. 
    $\boldsymbol{0}_{N\times M}$ denotes an $N\times M$ all-zero matrix. 
    $\boldsymbol{x}^T$, $\boldsymbol{x}^H$, and $\|\boldsymbol{x}\|_2$ denote the transpose, conjugate transpose, and Euclidean norm of vector $\boldsymbol{x}$, respectively. 
    Vector $\boldsymbol{1}_{K}$ denotes an all-one vector of size $K\times 1$. 
    For vector $\boldsymbol{x}$, $\text{diag}(\boldsymbol{x})$ denotes the diagonal matrix whose main diagonal elements are extracted from $\boldsymbol{x}$. 
    For matrix $\boldsymbol{A}$, $\text{diag}(\boldsymbol{A})$ denotes the vector whose elements are extracted from the main diagonal elements of $\boldsymbol{A}$. 
    For complex number $c$, let $\text{Re}(c)$, $\text{Im}(c)$, and $|c|$ denote its real part, imaginary part, and magnitude, respectively. 
    $\mathbb{C}^{a\times b}$ and $\mathbb{R}^{a\times b}$ denote the space of $a\times b$-dimensional complex and real matrices, respectively. 
    $\mathbb{E}[\cdot]$ denotes the statistical expectation. 
    Symbol $j$ denotes the imaginary unit $\sqrt{-1}$. 

\vspace{-6pt}
\section{System Model}\label{sec:system-model-and-protocol}
        \begin{figure}[t]
            \begin{center}
                \includegraphics[scale = 0.18]{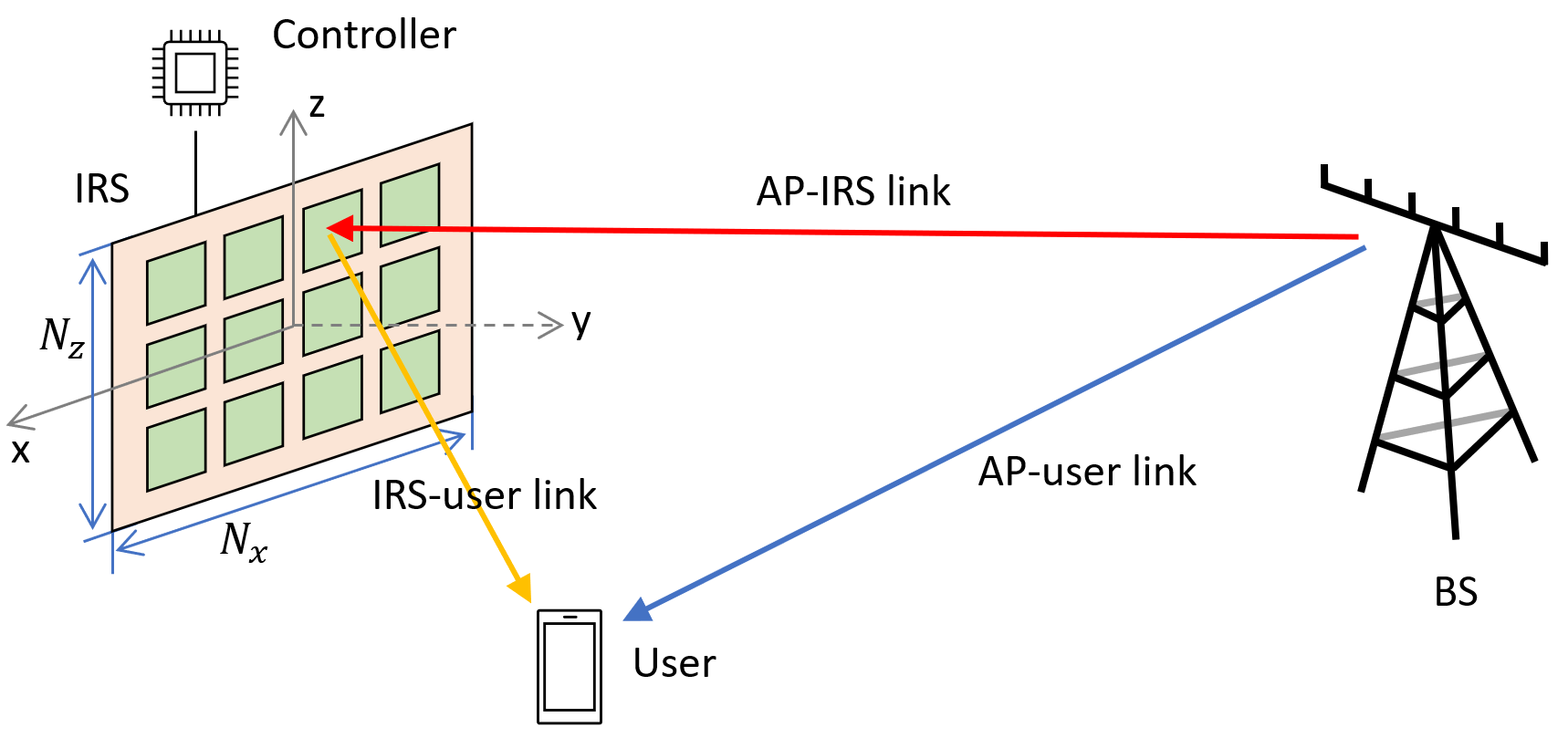}
                \caption{An IRS-aided wireless communication system. }
                \label{Fig:system}
            \end{center}
            \vspace{-6pt}
        \end{figure}
        
        As shown in Fig.~\ref{Fig:system}, we consider an IRS-aided downlink communication system with a multi-antenna BS serving a single-antenna user\footnote{The solutions proposed in this paper are also applicable to multiuser systems, by performing received signal power measurement at different users simultaneously and employing the proposed algorithms for each user's channel estimation. }, where an IRS is deployed near the user to establish a reflected link to assist in their communication. 
        To focus on the IRS passive reflection design, we assume that the BS fixes its precoding vector to guarantee the coverage of the user's residing area, and thus it can be considered having a single directional antenna equivalently. 
        The IRS is composed of $N_{irs} = N_{x}N_{z}$ reflecting elements, where $N_{x}$ and $N_{z}$ are the number of reflecting elements in the horizontal and vertical dimensions, respectively. 
        Each element of the IRS introduces a phase shift to the reflected signal. 
        Let $u_n$ denote the reflection coefficient of the $n$-th element, $n = 1, 2, \ldots, N_{irs}$, while $\boldsymbol{u} = [u_1, \ldots, u_{N_{irs}}]^T\in\mathbb{C}^{N_{irs}\times 1}$ and $\boldsymbol{\Theta} = \text{diag}(\boldsymbol{u}) \in\mathbb{C}^{N_{irs}\times N_{irs}}$ denote the IRS reflection coefficient vector and matrix, respectively. 
        Due to the unit amplitude constraint on the reflecting coefficients, we set $|u_n| = 1$, $\forall n$. 
        Furthermore, denote the number of bits for controlling the discrete phase shift of each element as $b$. 
        Then, the reflection coefficient $u_n$ should be selected from a discrete set $\Phi_b = \{e^{j\Delta\theta}, \ldots, e^{j2^b\Delta\theta}\}$, with $\Delta\theta = 2\pi/2^b$. 
        Denoting $\Phi_b^M = \{\boldsymbol{x}\in\mathbb{C}^{M\times 1} | x_n\in\Phi_b, n = 1, \ldots, M\}$ as the set of $M$-dimensional vectors whose elements are selected from $\Phi_b$, we thus have $\boldsymbol{u}\in\Phi_b^{N_{irs}}$. 
        
        The baseband equivalent channels of the BS-IRS link, BS-user link, and IRS-user link are denoted as $\boldsymbol{g}\in\mathbb{C}^{N_{irs}\times 1}$, $h_d\in\mathbb{C}$, and $\boldsymbol{h}_r\in\mathbb{C}^{N_{irs}\times 1}$, respectively. 
        The received signal corrupted by noise at the user is given by
        \begin{equation}\label{def:signal-model}
            y = \left(\boldsymbol{h}_{r}^{H}\boldsymbol{\Theta}\boldsymbol{g} + h_{d}^*\right)x + z, 
        \end{equation}
        where $x$ is the transmitted signal with zero mean and power $p_0$ and $z\sim\mathcal{CN}(0, \sigma^2)$ is the independent and identically distributed (i.i.d.) additive complex Guassian noise at the user receiver with mean zero and power $\sigma^2$. 
        Due to $\boldsymbol{h}_{r}^{H}\boldsymbol{\Theta} = \boldsymbol{h}_{r}^{H}\text{diag}(\boldsymbol{u}) = \boldsymbol{u}^{H}\text{diag}(\boldsymbol{h}_{r}^{H})$, we can rewrite~\eqref{def:signal-model} as 
        \begin{equation}
            y = \left(\boldsymbol{u}^{H}\text{diag}(\boldsymbol{h}_{r}^{H})\boldsymbol{g} + h_{d}^*\right)x + z. 
        \end{equation}
        Let $N = N_{irs} + 1$ and define $\bar{\boldsymbol{h}} = [\boldsymbol{g}^H\text{diag}(\boldsymbol{h}_r), h_d]^H\in\mathbb{C}^{N\times 1}$ as the equivalent channel and $\boldsymbol{v} = [\boldsymbol{u}^T, 1]^T\in\Phi_b^{N}$ as the equivalent IRS reflection vector. 
        Then, the received noisy signal can be simplified as $y = (\boldsymbol{v}^H\bar{\boldsymbol{h}})x + z$ and the power of the signal (ignoring the noise for the time being) is given by
        \begin{equation}\label{def:received-power}
            p = \mathbb{E}\left[\left|\left(\boldsymbol{h}_{r}^{H}\boldsymbol{\Theta}\boldsymbol{g} + h_{d}^*\right)x\right|^2\right] = p_0\left|
                    \boldsymbol{v}^H\bar{\boldsymbol{h}}
                \right|^2
            = p_0\text{tr}(\bar{\boldsymbol{H}}\boldsymbol{V}), 
        \end{equation}
        where $\bar{\boldsymbol{H}} = \bar{\boldsymbol{h}}\bar{\boldsymbol{h}}^H$ and $\boldsymbol{V} = \boldsymbol{v}\boldsymbol{v}^H$ are the autocorrelation matrices of the equivalent channel $\bar{\boldsymbol{h}}$ and the equivalent IRS reflection vector $\boldsymbol{v}$, respectively. 

\vspace{-6pt}
\section{Channel Autocorrelation Matrix Estimation}\label{subsec:estimation-framework}
    We assume that the user is quasi-static and its channels (or predominant deterministic channel components) with the BS and IRS do not change within a sufficiently long channel coherence time, during which the IRS first changes its reflection coefficients $T_{p}$ times in total. 
    In the meanwhile, the user measures the corresponding received signal power values and feed them back to a central processing unit (e.g., BS or IRS controller) that can design accordingly and then set the IRS reflection coefficients for data transmission. 
    Specifically, for the $t$-th power measurement, $t = 1, \ldots, T_{p}$, a random IRS reflection vector $\boldsymbol{u}_t$ is applied and $p_t$ denotes the received signal power at the user. 
    Define $\boldsymbol{v}_t = [\boldsymbol{u}_t, 1]^T$ and $\boldsymbol{V}_t = \boldsymbol{v}_t\boldsymbol{v}_t^H$, and then we have $p_t = p_0\text{tr}(\bar{\boldsymbol{H}}\boldsymbol{V}_t)$. 
    With all the received signals' power obtained, the channel autocorrelation matrix is estimated based on $p_t$ and $\boldsymbol{V}_t$, $t = 1, \ldots, T_{p}$. 

    Given power measurement $\boldsymbol{p} = [p_1, \ldots, p_{T_{p}}]^T\in\mathbb{R}^{T_{p}\times 1}$, the estimation problem can be formulated as finding a semidefinite rank-one matrix $\boldsymbol{H}$ that satisfies $p_0\text{tr}(\boldsymbol{H}\boldsymbol{V}_t) = p_t$, $t = 1, \ldots, T_{p}$, i.e., 
    \begin{subequations}\label{prob:cov-est-find-origin}
        \begin{align}
            & \text{find} \ {\boldsymbol{H}}\in\mathbb{S}_{+}^{N} \tag{\ref{prob:cov-est-find-origin}} \\
            & ~ \mathrm{s.t.} \ p_0\text{tr}(\boldsymbol{H}\boldsymbol{V}_t) = p_t, t = 1, \ldots, T_{p}, \label{prob:cov-est-find-power} \\
            & ~~~~~~ \text{rank}(\boldsymbol{H}) = 1, \label{prob:cov-est-find-rank-one}
        \end{align}
    \end{subequations}
    where $\mathbb{S}_{+}^{N}$ denotes the set of all positive semidefinite hermitian matrices of dimension $N\times N$. 
    The form of problem~\eqref{prob:cov-est-find-origin} is the same as the PhaseLift problem studied in~\cite{ref:PhaseLift, ref:stable-phaseless-recovery}, where the trace-minimization relaxation was applied to find an approximate solution. 
    However, the performance of the approximate solution relies on the assumption that vectors $\{\boldsymbol{v}_t, t = 1, \ldots, T_{p}\}$ are i.i.d. Guassian random vectors, which is not applicable to the considered IRS reflection vector due to its unit-amplitude elements with discrete phase shifts. 
    {\color{\highlightcolor}Thus, we propose to design more efficient methods customized to the practical constraint on $\{\boldsymbol{v}_t, t = 1, \ldots, T_{p}\}$ to achieve better performance but with lower computational complexity. }

    Due to the discrete phase shift constraint on IRS reflecting elements, the solution for problem~\eqref{prob:cov-est-find-origin} may not be unique. 
    Specifically, constraint~\eqref{prob:cov-est-find-power} forms a system of linear equations {\color{\highlightcolor}and the $N\times N$ hermitian matrices form a linear space of dimension $N^2$~\cite{ref:OLS}. 
    Thus, the hermitian autocorrelation matrix $\boldsymbol{H}$ can be uniquely determined if the maximum number of linearly independent matrices in $\{\boldsymbol{V}_t, t = 1, \ldots, T_{p}\}$, denoted as $D_V$ ($D_V\le T_{p}$), is no less than $N^2$. }
    However, the unit-amplitude and discrete phase of the entries in $\boldsymbol{V}_t$ result in $\{\boldsymbol{V}_t, t = 1, \ldots, T_{p}\}$ always being confined in a subspace of $\mathbb{C}^{N\times N}$ with its dimension smaller than $N^2$, as shown in the following lemma. 
    \begin{lemma}\label{lemma:dim-defficiency}
        {\color{\highlightcolor}For any $T_{p}$, $D_V\le \mathcal{D}_{N}^{(b)} < N^2$ always holds, where $\mathcal{D}_{N}^{(b)} = (N^2 - N)/2 + 1$ for $b = 1$ and $\mathcal{D}_{N}^{(b)} = N^2 - N + 1$ for $b \ge 2$. }
    \end{lemma}
    \begin{proof}[Proof\textup{:}\nopunct]\label{proof:lemma-dim-deff}
        {\color{\highlightcolor}See Appendix~\ref{appendix:lemma-proof-dim-deff}. }
    \end{proof}
    \vspace{-3pt}
    {\color{\highlightcolor}According to Lemma~\ref{lemma:dim-defficiency}, constraint~\eqref{prob:cov-est-find-power} forms an underdetermined system of equations due to the discrete phase shifts of IRS reflection vectors. 
    As a result, the uniqueness of the solution for problem~\eqref{prob:cov-est-find-origin} depends on the value of $b$, as illustrated in the following proposition. }
    \begin{proposition}\label{prop:cov-est-existence-uniqueness}
        {\color{\highlightcolor}For $b = 1$, problem~\eqref{prob:cov-est-find-origin} has only two solutions, i.e., $\bar{\boldsymbol{H}}$ and its conjugate matrix $\bar{\boldsymbol{H}}^*$, if $N$ and $D_V$ are sufficiently large (i.e., $N\ge 6$, $D_V = \mathcal{D}_{N}^{(1)}$). 
        For $b\ge 2$, problem~\eqref{prob:cov-est-find-origin} has one unique solution $\bar{\boldsymbol{H}}$ if $N$ and $D_V$ are sufficiently large (i.e., $N\ge 3$, $D_V = \mathcal{D}_{N}^{(b)}$). }
    \end{proposition}
    \begin{proof}[Proof\textup{:}\nopunct]\label{proof:solution-set}
        {\color{\highlightcolor}See Appendix~\ref{appendix:prop-proof-uniqueness}. }
    \end{proof}

    \vspace{-3pt}
    As the solutions for $b = 1$ and $b\ge 2$ are different, we solve problem~\eqref{prob:cov-est-find-origin} for these two cases separately. 
    For the case of $b\ge 2$, we derive the autocorrelation matrix $\bar{\boldsymbol{H}}$ by directly solving problem~\eqref{prob:cov-est-find-origin}.
    For the case of $b = 1$, the autocorrelation matrix cannot be uniquely determined, while the real part of the matrix, denoted as $\bar{\boldsymbol{H}}_r = \text{Re}(\bar{\boldsymbol{H}}) = \text{Re}(\bar{\boldsymbol{H}}^*)$, is unique according to Proposition~\ref{prop:cov-est-existence-uniqueness}. 
    Since the equivalent IRS reflection vector $\boldsymbol{v}$ is always a real vector for $b = 1$, the received signal power at the user always satisfies $p = p_0\text{tr}(\bar{\boldsymbol{H}}\boldsymbol{V}) = p_0\text{tr}(\bar{\boldsymbol{H}}_r\boldsymbol{V})$. 
    Therefore, we only need to estimate $\bar{\boldsymbol{H}}_r$ for optimizing the IRS reflection vector for data transmission. 
    Note that $\bar{\boldsymbol{H}}_r = (\bar{\boldsymbol{H}} + \bar{\boldsymbol{H}}^*) / 2$ is semidefinite with its rank no more than two because $\bar{\boldsymbol{H}}$ and $\bar{\boldsymbol{H}}^*$ are semidefinite rank-one matrices. 
    As such, for $b = 1$, we consider to estimate $\bar{\boldsymbol{H}}_r$ instead of $\bar{\boldsymbol{H}}$ by solving the following problem: 
    \begin{subequations}\label{prob:cov-est-find-rank-two}
        \allowdisplaybreaks
        \begin{align}
            & \text{find} \ {\boldsymbol{H}_r}\in\mathbb{M}_{+}^{N} \tag{\ref{prob:cov-est-find-rank-two}} \\
            & ~ \mathrm{s.t.} \ p_0\text{tr}(\boldsymbol{H}_r\boldsymbol{V}_t) = p_t, \ t = 1, \ldots, T_{p}, \label{prob:cov-est-find-rank-two-power} \\
            & ~~~~~~ \text{rank}(\boldsymbol{H}_r) \le 2, \label{prob:cov-est-find-rank-two-rank}
        \end{align}
    \end{subequations}
    where $\mathbb{M}_{+}^{N}$ denotes the set of all positive semidefinite real symmetric matrices of dimension $N\times N$. 
    The following proposition ensures the existence and uniqueness of the solution for problem~\eqref{prob:cov-est-find-rank-two}\footnote{{\color{\highlightcolor}Note that Propositions~\ref{prop:cov-est-existence-uniqueness} and~\ref{prop:solution-set-cov-real} only provide sufficient conditions on the uniqueness of the solution for problem~\eqref{prob:cov-est-find-origin} for the two cases of $b = 1$ and $b\ge 2$, respectively. 
    However, as revealed by simulations in Section~\ref{sec:performance-evaluation}, the unique solution $\bar{\boldsymbol{H}}$ for $b\ge 2$ and $\bar{\boldsymbol{H}}_r$ for $b = 1$ may also be found for smaller values of $D_V$ (or $T_{p}$) than that given in the two propositions. }}. 
    \begin{proposition}\label{prop:solution-set-cov-real}
        {\color{\highlightcolor}For $b = 1$, problem~\eqref{prob:cov-est-find-rank-two} has one unique solution $\bar{\boldsymbol{H}}_r$ if $N$ and $D_V$ are sufficiently large (i.e., $N \ge 6$, $D_V = \mathcal{D}_{N}^{(1)}$). }
    \end{proposition}
    \begin{proof}[Proof\textup{:}\nopunct]
        {\color{\highlightcolor}See Appendix~\ref{appendix:prop-proof-bit1-uniqueness}. }
    \end{proof}

    \vspace{-3pt}
    {\color{\highlightcolor}In the following sections, the LRA and ALRA algorithms are proposed for $b = 1$ and $b \ge 2$ by solving problems~\eqref{prob:cov-est-find-origin} and~\eqref{prob:cov-est-find-rank-two}, respectively, based on received signal power. }

\vspace{-6pt}
\section{LRA Algorithm}\label{sec:proposed-ratio-max}
    \vspace{-2pt}
    In this section, the channel autocorrelation matrix estimation problems~\eqref{prob:cov-est-find-origin} and~\eqref{prob:cov-est-find-rank-two} are solved by transforming them into equivalent rank-minimization problems. 
    Instead of directly finding a rank-one/two matrix solution, the proposed LRA algorithm iteratively approaches a low-rank matrix solution via alternating optimization. 

    \subsection{Solution for Problem~\eqref{prob:cov-est-find-origin}}\label{subsec:ratio-max-est-multiary}
        For $b\ge 2$, problem~\eqref{prob:cov-est-find-origin} can be equivalently transformed into the following rank-minimization problem: 
        \begin{subequations}\label{prob:cov-est-rank-min}
            \allowdisplaybreaks
            \begin{align}
                & \mathop{\min_{\boldsymbol{H}}} \ \text{rank}(\boldsymbol{H}) \tag{\ref{prob:cov-est-rank-min}} \\
                & ~ \mathrm{s.t.} \ p_0\text{tr}(\boldsymbol{H}\boldsymbol{V}_t) = p_t, \ t = 1, \ldots, T_{p}, \label{prob:cov-est-rank-min-power} \\
                & ~~~~~~ {\boldsymbol{H}}\in\mathbb{S}_{+}^{N}. \label{prob:cov-est-rank-min-semidefinite}
            \end{align}
        \end{subequations}
        The equivalence between problems~\eqref{prob:cov-est-find-origin} and~\eqref{prob:cov-est-rank-min} is analyzed as follows. 
        As the channel autocorrelation matrix $\bar{\boldsymbol{H}}$ is feasible to problem~\eqref{prob:cov-est-rank-min}, any optimal solution for this problem, denoted by $\hat{\boldsymbol{H}}$, should satisfy $\text{rank}(\hat{\boldsymbol{H}}) \le \text{rank}(\bar{\boldsymbol{H}}) = 1$, leading to $\text{rank}(\hat{\boldsymbol{H}}) = 1$ and thus $\hat{\boldsymbol{H}}$ is also a solution for problem~\eqref{prob:cov-est-find-origin}. 
        Reversely, any solution $\hat{\boldsymbol{H}}'$ for problem~\eqref{prob:cov-est-find-origin} is feasible to the rank-minimization problem~\eqref{prob:cov-est-rank-min} and satisfies $\text{rank}(\hat{\boldsymbol{H}}') = 1$, which indicates that $\hat{\boldsymbol{H}}'$ is an optimal solution for problem~\eqref{prob:cov-est-rank-min}. 

        Since $\boldsymbol{H}\in\mathbb{S}_{+}^{N}$ is nonzero, all the eigenvalues of $\boldsymbol{H}$ are real and non-negative and $\text{tr}(\boldsymbol{H}) > 0$. 
        Define the eigenvalue-ratio function for matrix $\boldsymbol{H}$ as
        \begin{equation}\label{def:lambda-ratio-func}
            g(\boldsymbol{H}) = \frac{\lambda_{1}(\boldsymbol{H})}{\text{tr}(\boldsymbol{H})}, \ \boldsymbol{H}\in\mathbb{S}_{+}^{N}, 
        \end{equation}
        where $\lambda_{1}(\boldsymbol{H})$ is the largest eigenvalue of $\boldsymbol{H}$. 
        Obviously, $0 < g(\boldsymbol{H})\le 1$ holds for any nonzero $\boldsymbol{H}\in\mathbb{S}_{+}^{N}$, and it is easy to verify that $\text{rank}(\boldsymbol{H}) = 1$ if and only if $g(\boldsymbol{H}) = 1$. 
        As we have mentioned above, any solution $\hat{\boldsymbol{H}}$ for problem~\eqref{prob:cov-est-rank-min} satisfies $\text{rank}(\hat{\boldsymbol{H}}) = 1$. 
        Thus, we have $g(\hat{\boldsymbol{H}}) = 1$, which means that $\hat{\boldsymbol{H}}$ maximizes the eigenvalue-ratio function $g(\boldsymbol{H})$. 
        On the other hand, any matrix $\hat{\boldsymbol{H}}'$ that maximizes $g(\boldsymbol{H})$ subject to constraints~\eqref{prob:cov-est-rank-min-power} and~\eqref{prob:cov-est-rank-min-semidefinite} also minimizes $\text{rank}(\boldsymbol{H})$. 
        Therefore, the solutions for problem~\eqref{prob:cov-est-rank-min} are the same as the solutions that maximize the eigenvalue-ratio function $g(\boldsymbol{H})$ subject to constraints~\eqref{prob:cov-est-rank-min-power} and~\eqref{prob:cov-est-rank-min-semidefinite}. 
        Note that $\lambda_{1}(\boldsymbol{H}) = \mathop{\max_{\|\boldsymbol{x}\|_2\le 1}}{\boldsymbol{x}^H\boldsymbol{H}\boldsymbol{x}}$. 
        Thus, problem~\eqref{prob:cov-est-rank-min} can be written as 
        \begin{equation}\label{prob:cov-est-maxmax}
            \mathop{\max_{\boldsymbol{H}}\max_{\|\boldsymbol{x}\|_2 \le 1}} \ f(\boldsymbol{H}, \boldsymbol{x}) = \frac{\boldsymbol{x}^H\boldsymbol{H}\boldsymbol{x}}{\text{tr}(\boldsymbol{H})}, 
            ~~ \mathrm{s.t.} ~\eqref{prob:cov-est-rank-min-power},\eqref{prob:cov-est-rank-min-semidefinite}. 
        \end{equation}
        This optimization problem is non-convex, while alternating optimization can be employed to obtain a suboptimal solution for it. 
        Given $\boldsymbol{H}$, an optimal $\boldsymbol{x}$ can be obtained as the normalized eigenvector of $\boldsymbol{H}$ corresponding to the largest eigenvalue. 
        Given $\boldsymbol{x}$, the optimization of $\boldsymbol{H}$ is simplified as 
        \begin{equation}\label{prob:cov-est-maxmax-cov}
            \mathop{\max_{\boldsymbol{H}}} \ \frac{\text{tr}(\boldsymbol{H}\boldsymbol{X})}{\text{tr}(\boldsymbol{H})}, 
            ~~ \mathrm{s.t.} ~\eqref{prob:cov-est-rank-min-power},\eqref{prob:cov-est-rank-min-semidefinite}, 
        \end{equation}
        with $\boldsymbol{X} = \boldsymbol{x}\boldsymbol{x}^H$. 
        This is a fractional programming problem and can be transformed into a convex optimization problem. 
        Specifically, define $\boldsymbol{G} = \boldsymbol{H}/\text{tr}(\boldsymbol{H})$ and $\gamma = 1/\text{tr}(\boldsymbol{H})$. 
        Then, problem~\eqref{prob:cov-est-maxmax-cov} can be transformed into
        \begin{subequations}\label{prob:cov-est-maxmax-frac-progm}
            \allowdisplaybreaks
            \begin{align}
                & \mathop{\max_{\boldsymbol{G}, \gamma > 0}} \ \text{tr}(\boldsymbol{G}\boldsymbol{X}) \tag{\ref{prob:cov-est-maxmax-frac-progm}} \\
                & ~ \mathrm{s.t.} \ p_0\text{tr}(\boldsymbol{G}\boldsymbol{V}_t) = p_t\gamma, \ t = 1, \ldots, T_{p}, \label{prob:cov-est-maxmax-frac-progm-power} \\
                & ~~~~~~ {\boldsymbol{G}}\in\mathbb{S}_{+}^{N}, \ \text{tr}(\boldsymbol{G}) = 1, \label{prob:cov-est-maxmax-frac-progm-semidefinite} 
            \end{align}
        \end{subequations}
        which is a convex semidefinite programming (SDP) problem and thus can be solved by CVX~\cite{ref:Boyd-cvx}. 
        Denoting the optimal solution for problem~\eqref{prob:cov-est-maxmax-frac-progm} as $\hat{\boldsymbol{G}}$ and $\hat{\gamma}$, then the optimal solution for problem~\eqref{prob:cov-est-maxmax-cov} is obtained as $\hat{\boldsymbol{H}} = \hat{\boldsymbol{G}}/\hat{\gamma}$. 

        The proposed LRA algorithm to solve problem~\eqref{prob:cov-est-find-origin} for $b\ge 2$ is summarized in Algorithm~\ref{alg:RX-est-multiary}. 
        The computational complexity of the algorithm is dominated by the SDP problem~\eqref{prob:cov-est-maxmax-frac-progm}, which is given by $\mathcal{O}(N^{4.5})$~\cite{ref:SDR-in-Quadratic}. 
        Note that in line $3$, only the eigenvector corresponding to the largest eigenvalue is needed. 
        The well-known power method~\cite{ref:power-method} with the complexity of $\mathcal{O}(N^2)$ can be applied. 
        Thus, the computational complexity of Algorithm~\ref{alg:RX-est-multiary} is given by $\mathcal{O}(N^{4.5}I_1)$, {\color{\highlightcolor}which is in the same order as the trace-minimization method in~\cite{ref:PhaseLift}}, where $I_1$ is the total number of iterations. 

        The convergence of the LRA algorithm is analyzed as follows. 
        In Algorithm~\ref{alg:RX-est-multiary}, variables $\boldsymbol{H}$ and $\boldsymbol{x}$ are updated as $\boldsymbol{H}^{(i)}$ and $\boldsymbol{x}^{(i)}$ in the $i$-th iteration. 
        It is worth noting that $\boldsymbol{x}^{(i)}$ maximizes $f(\boldsymbol{H}, \boldsymbol{x})$ given $\boldsymbol{H} = \boldsymbol{H}^{(i - 1)}$ and thus
        \begin{equation}
            g(\boldsymbol{H}^{(i - 1)}) = \max_{\|\boldsymbol{x}\|_2 \le 1} f(\boldsymbol{H}^{(i - 1)}, \boldsymbol{x}) = f(\boldsymbol{H}^{(i - 1)}, \boldsymbol{x}^{(i)})
        \end{equation}
        holds. 
        Moreover, $\boldsymbol{H}^{(i)}$ is the optimal solution for problem~\eqref{prob:cov-est-maxmax-cov} with $\boldsymbol{X} = \boldsymbol{X}^{(i)}$, which means 
        \begin{equation}
            \begin{aligned}
                f(\boldsymbol{H}^{(i)}, \boldsymbol{x}^{(i)}) & = \frac{\text{tr}(\boldsymbol{H}^{(i)}\boldsymbol{X}^{(i)})}{\text{tr}(\boldsymbol{H}^{(i)})} \ge \frac{\text{tr}(\boldsymbol{H}^{(i - 1)}\boldsymbol{X}^{(i)})}{\text{tr}(\boldsymbol{H}^{(i - 1)})} \\
                & = f(\boldsymbol{H}^{(i - 1)}, \boldsymbol{x}^{(i)}), \ i = 1, \ldots, I_1. 
            \end{aligned}
        \end{equation}
        Thus, we have 
        \begin{subequations}\label{eq:ratio-func-monotonic}
            \begin{align}
                g(\boldsymbol{H}^{(i)}) & = \max_{\|\boldsymbol{x}\|_2 \le 1} f(\boldsymbol{H}^{(i)}, \boldsymbol{x}) \ge f(\boldsymbol{H}^{(i)}, \boldsymbol{x}^{(i)}) \\
                & \ge f(\boldsymbol{H}^{(i - 1)}, \boldsymbol{x}^{(i)}) = g(\boldsymbol{H}^{(i - 1)}), \ \forall i. 
            \end{align}
        \end{subequations}
        Therefore, the eigenvalue-ratio function $g(\boldsymbol{H}^{(i)})$ is non-decreasing during the iterations for $i = 0, 1, \ldots, I_1$. 
        Since $g(\boldsymbol{H})$ is upper-bounded by $1$, the convergence of Algorithm~\ref{alg:RX-est-multiary} is guaranteed. 
        Meanwhile, the monotonic increase of the eigenvalue-ratio function indicates that matrix $\boldsymbol{H}^{(i)}$ gradually approaches a rank-one matrix during the iterations. 

        \begin{algorithm}[t]
            \begin{minipage}{0.95\linewidth}
            \centering
            \caption{LRA algorithm for $b\ge 2$. }
            \begin{algorithmic}[1]\label{alg:RX-est-multiary}
                \REQUIRE~$\{\boldsymbol{v}_t$, $t = 1, \ldots, T\}$, $\boldsymbol{p}\in\mathbb{R}^{T\times 1}$, threshold $\epsilon$. 
                \STATE~Initialization: Solve $\boldsymbol{H}^{(0)}$ via trace-minimization relaxation~\cite{ref:PhaseLift}; iteration index $i\gets 1$. 
                \WHILE{$g(\boldsymbol{H}^{(i - 1)}) \le \epsilon$}
                    \STATE~Let $\boldsymbol{x}^{(i)}$ be the normalized eigenvector of $\boldsymbol{H}^{(i -1)}$ corresponding to the largest eigenvalue, and $\boldsymbol{X}^{(i)} \gets \boldsymbol{x}^{(i)}{\boldsymbol{x}^{(i)}}^H$. 
                    \STATE~Solve $\boldsymbol{G}^{(i)}$ and $\gamma^{(i)}$ from problem~\eqref{prob:cov-est-maxmax-frac-progm} with $\boldsymbol{X} = \boldsymbol{X}^{(i)}$, and obtain $\boldsymbol{H}^{(i)} \gets \boldsymbol{G}^{(i)}/\gamma^{(i)}$. 
                    \STATE~$i\gets i + 1$. 
                \ENDWHILE
                \RETURN~The estimated matrix $\hat{\boldsymbol{H}}\gets\boldsymbol{H}^{(i - 1)}$. 
            \end{algorithmic}
            \end{minipage}
        \end{algorithm}


    \vspace{-6pt}
    \subsection{Solution for Problem~\eqref{prob:cov-est-find-rank-two}}\label{subsec:ratio-max-est-binary}
        Next, we consider $b = 1$. 
        To solve matrix $\bar{\boldsymbol{H}}_r$ from problem~\eqref{prob:cov-est-find-rank-two}, a similar method to the case of $b\ge 2$ can be applied. 
        Consider the following rank-minimization problem: 
        \begin{subequations}\label{prob:cov-real-est-rank-min}
            \allowdisplaybreaks
            \begin{align}
                & \mathop{\min_{\boldsymbol{H}_r}} \ \text{rank}(\boldsymbol{H}_r) \tag{\ref{prob:cov-real-est-rank-min}} \\
                & ~ \mathrm{s.t.} \ p_0\text{tr}(\boldsymbol{H}_r\boldsymbol{V}_t) = p_t, \ t = 1, \ldots, T_{p}, \label{prob:cov-real-est-rank-min-power} \\
                & ~~~~~~ {\boldsymbol{H}_r}\in\mathbb{M}_{+}^{N}. \label{prob:cov-real-est-rank-min-semidefinite}
            \end{align}
        \end{subequations}
        Obviously, matrix $\bar{\boldsymbol{H}}_r = \text{Re}(\bar{\boldsymbol{H}})$ is feasible to problem~\eqref{prob:cov-real-est-rank-min}. 
        Any optimal solution for problem~\eqref{prob:cov-real-est-rank-min}, denoted by $\hat{\boldsymbol{H}}_r$, satisfies $\text{rank}(\hat{\boldsymbol{H}}_r)\le\text{rank}(\bar{\boldsymbol{H}}_r)\le 2$, which means that $\hat{\boldsymbol{H}}_r$ is also a solution for problem~\eqref{prob:cov-est-find-rank-two}. 
        Thus, the solution for problem~\eqref{prob:cov-est-find-rank-two} can be obtained by solving problem~\eqref{prob:cov-real-est-rank-min}. 

        Define the generalized eigenvalue-ratio function as
        \begin{equation}\label{def:generalized-lambda-ratio-func}
            g_r(\boldsymbol{H}_r) = \frac{\lambda_1(\boldsymbol{H}_r) + \lambda_2(\boldsymbol{H}_r)}{\text{tr}(\boldsymbol{H}_r)}, \ \boldsymbol{H}_r\in\mathbb{S}_{+}^{N}, 
        \end{equation}
        where $\lambda_1(\boldsymbol{H}_r)$ and $\lambda_2(\boldsymbol{H}_r)$ are the first and second largest eigenvalues of $\boldsymbol{H}_r$, respectively. 
        For nonzero $\boldsymbol{H}_r\in\mathbb{S}_{+}^{N}$, we have $\text{tr}(\boldsymbol{H}_r) > 0$, $0 < g_r(\boldsymbol{H}_r)\le 1$, and $g_r(\boldsymbol{H}_r) = 1$ holds if and only if $\text{rank}(\boldsymbol{H}_r)\le 2$. 
        Therefore, solving problem~\eqref{prob:cov-real-est-rank-min} is equivalent to maximizing $g_r(\boldsymbol{H}_r)$ subject to constraints~\eqref{prob:cov-real-est-rank-min-power} and~\eqref{prob:cov-real-est-rank-min-semidefinite}. 
        Furthermore, since we have
        \begin{subequations}\label{prob:largest-two-eigenvalues}
            \begin{align}
                \lambda_1(\boldsymbol{H}_r) + & \lambda_2(\boldsymbol{H}_r) = \mathop{\max_{\boldsymbol{x}_1, \boldsymbol{x}_2}} \ \boldsymbol{x}_1^T\boldsymbol{H}_r\boldsymbol{x}_1 + \boldsymbol{x}_2^T\boldsymbol{H}_r\boldsymbol{x}_2 \tag{\ref{prob:largest-two-eigenvalues}}, \\
                & ~ \mathrm{s.t.} \ \|\boldsymbol{x}_1\|_2 = 1,~\|\boldsymbol{x}_2\|_2 = 1,~\boldsymbol{x}_1^T\boldsymbol{x}_2 = 0, \label{prob:largest-two-eigenvalues-unitary}
            \end{align}
        \end{subequations}
        problem~\eqref{prob:cov-real-est-rank-min} can be equivalently written as 
        \begin{subequations}\label{prob:cov-est-maxmax-real}
            \begin{align}
                & \mathop{\max_{\boldsymbol{H}_r}\max_{\boldsymbol{x}_1, \boldsymbol{x}_2}} \ f_r(\boldsymbol{H}_r, \boldsymbol{x}_1, \boldsymbol{x}_2) = \frac{\boldsymbol{x}_1^T\boldsymbol{H}_r\boldsymbol{x}_1 + \boldsymbol{x}_2^T\boldsymbol{H}_r\boldsymbol{x}_2}{\text{tr}(\boldsymbol{H}_r)} \tag{\ref{prob:cov-est-maxmax-real}} \\
                & ~ \mathrm{s.t.} \ \eqref{prob:largest-two-eigenvalues-unitary},\eqref{prob:cov-real-est-rank-min-power},\eqref{prob:cov-real-est-rank-min-semidefinite}. \nonumber
            \end{align}
        \end{subequations}
        Alternating optimization can be applied to solve problem~\eqref{prob:cov-est-maxmax-real} sub-optimally. 
        Given $\boldsymbol{H}_r$, the optimal $\boldsymbol{x}_1$ and $\boldsymbol{x}_2$ can be obtained as the normalized eigenvectors corresponding to the first and second largest eigenvalues of $\boldsymbol{H}_r$, respectively. 
        Given $\boldsymbol{x}_1$ and $\boldsymbol{x}_2$, $\boldsymbol{H}_r$ can be optimized via
        \begin{equation}\label{prob:cov-est-maxmax-real-cov}
            \mathop{\max_{\boldsymbol{H}_r}} \ \frac{\text{tr}(\boldsymbol{H}_r\boldsymbol{X})}{\text{tr}(\boldsymbol{H}_r)}, 
            ~~ \mathrm{s.t.} ~\eqref{prob:cov-real-est-rank-min-power},\eqref{prob:cov-real-est-rank-min-semidefinite}, 
        \end{equation}
        with $\boldsymbol{X} = \boldsymbol{x}_1\boldsymbol{x}_1^T + \boldsymbol{x}_2\boldsymbol{x}_2^T$. 
        This problem is also a fractional programming problem and can be solved in the same way as problem~\eqref{prob:cov-est-maxmax-cov}. 
        The proposed LRA algorithm to solve problem~\eqref{prob:cov-est-find-rank-two} for the case of $b = 1$ is summarized in Algorithm~\ref{alg:RX-est-binary}, with a complexity of $\mathcal{O}(N^{4.5}I_2)$ {\color{\highlightcolor}which is also in the same order as the trace-minimization method in~\cite{ref:PhaseLift}}, where $I_2$ is the total number of iterations. 
        The convergence can be guaranteed by the monotonic increase of the generalized eigenvalue-ratio function $g_r(\boldsymbol{H}_r^{(i)})$, similar to Algorithm~\ref{alg:RX-est-multiary}. 

        \begin{algorithm}[t]
            \begin{minipage}{0.95\linewidth}
            \centering
            \caption{LRA algorithm for $b = 1$. }
            \begin{algorithmic}[1]\label{alg:RX-est-binary}
                \REQUIRE~$\{\boldsymbol{v}_t$, $t = 1, \ldots, T\}$, $\boldsymbol{p}\in\mathbb{R}^{T\times 1}$, threshold $\epsilon$.  
                \STATE~Initialization: Solve $\boldsymbol{H}_r^{(0)}$ via trace-minimization relaxation~\cite{ref:PhaseLift}; iteration index $i\gets 1$. 
                \WHILE{$g_r(\boldsymbol{H}_r^{(i - 1)}) \le \epsilon$}
                    \STATE~Let $\boldsymbol{x}_1^{(i)}$ and $\boldsymbol{x}_2^{(i)}$ be the normlized eigenvectors of $\boldsymbol{H}_r^{(i - 1)}$ corresponding to the first and second largest eigenvalues, and $\boldsymbol{X}^{(i)} \gets \boldsymbol{x}_1^{(i)}{\boldsymbol{x}_1^{(i)}}^T + \boldsymbol{x}_2^{(i)}{\boldsymbol{x}_2^{(i)}}^T$. 
                    \STATE~Solve $\boldsymbol{H}_r^{(i)}$ from problem~\eqref{prob:cov-est-maxmax-real-cov} with $\boldsymbol{X} = \boldsymbol{X}^{(i)}$. 
                    \STATE~$i\gets i + 1$. 
                \ENDWHILE
                \RETURN~The estimated matrix $\hat{\boldsymbol{H}_r}\gets\boldsymbol{H}_r^{(i - 1)}$. 
            \end{algorithmic}
            \end{minipage}
        \end{algorithm}

\vspace{-6pt}
\section{Approximate LRA Algorithm}\label{sec:proposed-dist-min}
    
    The proposed LRA algorithms can approach a low-rank matrix solution via alternating maximization of the eigenvalue-ratio functions, but their computational complexity may be high in practice because SDP is involved for each iteration. 
    In this section, proper approximations are implemented in the optimization of $\boldsymbol{H}$ and $\boldsymbol{H}_r$ for $b\ge 2$ and $b = 1$, respectively, and the ALRA algorithm is proposed to reduce the computational complexity. 
    The basic principle of the LRA algorithms is retained, while the eigenvalue-ratio functions are replaced by quadratic functions. 
    By exploiting the vector representations of hermitian matrices, closed-form solutions are derived for the ALRA algorithm during the iterations and thus the computational complexity is significantly reduced. 
    
    \vspace{-6pt}
    \subsection{Distance-Minimization Approximation}\label{subec:dist-min-approx}
        For each iteration of Algorithm~\ref{alg:RX-est-multiary} with $b\ge 2$, a normalized eigenvector $\boldsymbol{x}$ is obtained and then $\boldsymbol{H}$ is optimized via problem~\eqref{prob:cov-est-maxmax-cov}. 
        The objective of problem~\eqref{prob:cov-est-maxmax-cov} can be written as 
        \begin{equation}\label{prob:rx-maxmax-cov}
            \mathop{\max_{\boldsymbol{H}}} \frac{\text{tr}(\boldsymbol{H}\boldsymbol{X})}{\text{tr}(\boldsymbol{H})} = \text{tr}(\boldsymbol{X})\mathop{\max_{\boldsymbol{H}}} \frac{\text{tr}(\boldsymbol{H}\boldsymbol{X})}{\text{tr}(\boldsymbol{H})\text{tr}(\boldsymbol{X})}, 
        \end{equation}
        where $\boldsymbol{X} = \boldsymbol{x}\boldsymbol{x}^H$. 
        Note that $\text{tr}(\boldsymbol{H}\boldsymbol{X}) = \text{vec}(\boldsymbol{H})^H\text{vec}(\boldsymbol{X})$ and for $\boldsymbol{H}\in\mathbb{S}_{+}^{N}$, we have $\text{tr}(\boldsymbol{H})\ge\sqrt{\text{tr}(\boldsymbol{H}^H\boldsymbol{H})} = \|\text{vec}(\boldsymbol{H})\|_2$. 
        Thus, the objective function on the right-hand side of~\eqref{prob:rx-maxmax-cov} is upper-bounded by
        \begin{subequations}\label{eq:rx-maxmax-upper-bound}
            \allowdisplaybreaks
            \begin{align}
                \frac{\text{tr}(\boldsymbol{H}\boldsymbol{X})}{\text{tr}(\boldsymbol{H})\text{tr}(\boldsymbol{X})} & \le \frac{\text{tr}(\boldsymbol{H}\boldsymbol{X})}{\sqrt{\text{tr}(\boldsymbol{H}^H\boldsymbol{H})}\sqrt{\text{tr}(\boldsymbol{X}^H\boldsymbol{X})}} \\
                & = \frac{\text{vec}(\boldsymbol{H})^H\text{vec}(\boldsymbol{X})}{\|\text{vec}(\boldsymbol{H})\|_2\|\text{vec}(\boldsymbol{X})\|_2}, 
            \end{align}
        \end{subequations}
        which is the inner product of two normalized vectors $\text{vec}(\boldsymbol{H})/\|\text{vec}(\boldsymbol{H})\|_2$ and $\text{vec}(\boldsymbol{X})/\|\text{vec}(\boldsymbol{X})\|_2$. 
        As a result, the optimization of $\boldsymbol{H}$ can be approximately seen as aligning the direction of $\text{vec}(\boldsymbol{H})$ to that of $\text{vec}(\boldsymbol{X})$, i.e., minimizing the cosine distance between them, which is defined as 
        \begin{equation}\label{def:cosine-distance}
            d_c(\text{vec}(\boldsymbol{H}), \text{vec}(\boldsymbol{X})) = 1 - \frac{\text{vec}(\boldsymbol{H})^H\text{vec}(\boldsymbol{X})}{\|\text{vec}(\boldsymbol{H})\|_2\|\text{vec}(\boldsymbol{X})\|_2}. 
        \end{equation}
        To reduce the computational complexity, we consider to minimize the Euclidean distance between $\text{vec}(\boldsymbol{H})$ and $\mu\text{vec}(\boldsymbol{X})$, where $\mu$ is a scaling factor to be determined. 
        Specifically, the distance-minimization problem is given by
        \begin{subequations}\label{prob:cov-est-dist-min}
            \allowdisplaybreaks
            \begin{align}
                & \mathop{\min_{\boldsymbol{H}, \mu}}\ d_e(\boldsymbol{H}, \mu\boldsymbol{X}) \tag{\ref{prob:cov-est-dist-min}} \\
                & ~ \mathrm{s.t.} \ p_0\text{tr}(\boldsymbol{H}\boldsymbol{V}_t) = p_t, \ t = 1, \ldots, T_{p}, \label{prob:cov-est-dist-min-power} \\
                & ~~~~~~ \boldsymbol{H}\in\mathbb{S}_{+}^{N}, \label{prob:cov-est-dist-min-semidefinite}
            \end{align}
        \end{subequations}
        where $d_e(\boldsymbol{H}, \mu\boldsymbol{X})$ denotes the Euclidean distance between vectors $\text{vec}(\boldsymbol{H})$ and $\text{vec}(\mu\boldsymbol{X})$ and it is defined as 
        \begin{equation}\label{def:euclidean-distance-func}
            d_e(\boldsymbol{H}, \mu\boldsymbol{X}) = \|\text{vec}(\boldsymbol{H}) - \text{vec}(\mu\boldsymbol{X})\|_2 = \|\boldsymbol{H} - \mu\boldsymbol{X}\|_F. 
        \end{equation}
        To further simplify the optimization of $\boldsymbol{H}$, we relax the semidefinite constraint~\eqref{prob:cov-est-dist-min-semidefinite} to a hermitian constraint and solve the following problem:
        \begin{subequations}\label{prob:cov-est-dist-min-relaxed}
            \allowdisplaybreaks
            \begin{align}
                & \mathop{\min_{\boldsymbol{H}, \mu}}\ \|\boldsymbol{H} - \mu\boldsymbol{X}\|_F^2 \tag{\ref{prob:cov-est-dist-min-relaxed}} \\
                & ~ \mathrm{s.t.} \ p_0\text{tr}(\boldsymbol{H}\boldsymbol{V}_t) = p_t, \ t = 1, \ldots, T_{p}, \label{prob:cov-est-dist-min-relaxed-power} \\
                & ~~~~~~ \boldsymbol{H}^H = \boldsymbol{H}. \label{prob:cov-est-dist-min-relaxed-hermitian}
            \end{align}
        \end{subequations}

        Similarly, for Algorithm~\ref{alg:RX-est-binary} with $b = 1$, the optimization of $\boldsymbol{H}_r$ in problem~\eqref{prob:cov-est-maxmax-real-cov} can be approximated by minimizing the Euclidean distance function. 
        For $b = 1$, $\boldsymbol{X} = \boldsymbol{x}_1\boldsymbol{x}_1^T + \boldsymbol{x}_2\boldsymbol{x}_2^T$ is rank-two. 
        To improve the recovery accuracy, we consider a vector $\boldsymbol{\mu} = [\mu_1, \mu_2]^T\in\mathbb{R}^{2\times 1}$ such that $\boldsymbol{H}_r$ is optimized via
        \begin{subequations}\label{prob:cov-est-dist-min-relaxed-binary}
            \allowdisplaybreaks
            \begin{align}
                & \mathop{\min_{\boldsymbol{H}_r, \boldsymbol{\mu}}}\ \|\boldsymbol{H}_r - \left(\mu_1\boldsymbol{X}_1 + \mu_2\boldsymbol{X}_2\right)\|_F^2 \tag{\ref{prob:cov-est-dist-min-relaxed-binary}} \\
                & ~ \mathrm{s.t.} \ p_0\text{tr}(\boldsymbol{H}_r\boldsymbol{V}_t) = p_t, \ t = 1, \ldots, T_{p}, \label{prob:cov-est-dist-min-relaxed-binary-power} \\
                & ~~~~~~ \boldsymbol{H}_r^T = \boldsymbol{H}_r\in\mathbb{R}^{N\times N}, \label{prob:cov-est-dist-min-relaxed-binary-symmetric}
            \end{align}
        \end{subequations}
        where $\boldsymbol{X}_1 = \boldsymbol{x}_1\boldsymbol{x}_1^T$ and $\boldsymbol{X}_2 = \boldsymbol{x}_2\boldsymbol{x}_2^T$ are obtained in the previous iteration in Algorithm~\ref{alg:RX-est-binary}. 
        Since problems~\eqref{prob:cov-est-dist-min-relaxed} and~\eqref{prob:cov-est-dist-min-relaxed-binary} are quadratic optimization problems with linear constraints, their optimal solutions can be given in closed form, which are derived in the following context. 
    
    \vspace{-8pt}
    \subsection{Solution for Problem~\eqref{prob:cov-est-dist-min-relaxed}}\label{subsec:solution-dist-min-multiary}

        According to~\cite{ref:OLS}, all hermitian matrices of size $N\times N$ form a linear space of dimension $N^2$ with real combination coefficients. 
        In particular, each $N\times N$ hermitian matrix $\boldsymbol{A}$ can be represented by an $N^2$-dimensional real vector $\boldsymbol{w}_a$, which is the coordinate of $\boldsymbol{A}$ in the hermitian matrix space. 
        Following the definition in~\cite{ref:OLS}, denote the bijective mapping from $\boldsymbol{A}$ to $\boldsymbol{w}_a$ as function $\mathcal{M}:\mathbb{C}^{N\times N}\to\mathbb{R}^{N^2\times 1}$, such that $\boldsymbol{w}_a = \mathcal{M}(\boldsymbol{A})$ and $\boldsymbol{A} = \mathcal{M}^{-1}(\boldsymbol{w}_a)$. 
        It has been shown in~\cite{ref:OLS} that for any hermitian matrices $\boldsymbol{A}$ and $\boldsymbol{B}$, we have $\textup{tr}(\boldsymbol{A}\boldsymbol{B}) = \boldsymbol{w}_a^T\boldsymbol{w}_b$, where $\boldsymbol{w}_a = \mathcal{M}(\boldsymbol{A})$ and $\boldsymbol{w}_b = \mathcal{M}(\boldsymbol{B})$. 
        Based on this property, problem~\eqref{prob:cov-est-dist-min-relaxed} can be equivalently written as a quadratic optimization problem with vector variables. 
        Denote $\boldsymbol{w} = \mathcal{M}(\boldsymbol{H})$, $\boldsymbol{w}_{X} = \mathcal{M}(\boldsymbol{X})$, and $\boldsymbol{w}_t = \mathcal{M}(\boldsymbol{V}_t)$, $t = 1, \ldots, T_{p}$. 
        Then, we have $p_t = p_0\text{tr}(\boldsymbol{H}\boldsymbol{V}_t) = p_0\boldsymbol{w}^T\boldsymbol{w}_t$, $\forall t$, and
        \begin{equation}
            \begin{aligned}
                \|\boldsymbol{H} - \mu\boldsymbol{X}\|_F^2 & = \text{tr}\left(
                (\boldsymbol{H} - \mu\boldsymbol{X})^H(\boldsymbol{H} - \mu\boldsymbol{X})
                \right) \\
                & = \|\boldsymbol{w} - \mu\boldsymbol{w}_{X}\|_2^2. 
            \end{aligned}
        \end{equation}
        Define $\boldsymbol{C} = p_0[\boldsymbol{w}_1, \ldots, \boldsymbol{w}_{T_{p}}]\in\mathbb{R}^{N^2\times T_{p}}$, and then the power constraints $p_t = p_0\text{tr}(\boldsymbol{H}\boldsymbol{V}_t)$ can be expressed as $\boldsymbol{C}^T\boldsymbol{w} = \boldsymbol{p}$. 
        Thus, problem~\eqref{prob:cov-est-dist-min-relaxed} can be equivalently written as 
        \begin{equation}\label{prob:cov-est-dist-min-relaxed-decomp}
                \mathop{\min_{\boldsymbol{w}, \mu}}\ \|\boldsymbol{w} - \mu\boldsymbol{w}_{X}\|_2^2, ~~ \mathrm{s.t.} \ \boldsymbol{C}^T\boldsymbol{w} = \boldsymbol{p}, 
        \end{equation}
        where the hermitian constraint~\eqref{prob:cov-est-dist-min-semidefinite} is dropped because it is guaranteed by the transformation from $\boldsymbol{H}$ to $\boldsymbol{w}$. 
        For any given $\mu$, the Lagrange function for variable $\boldsymbol{w}$ is given by
        \begin{equation}\label{def:lagrange-func}
            L(\boldsymbol{w}, \boldsymbol{\lambda}) = \|\boldsymbol{w} - \mu\boldsymbol{w}_{X}\|_2^2 + \boldsymbol{\lambda}^T(\boldsymbol{p} - \boldsymbol{C}^T\boldsymbol{w}), 
        \end{equation}
        where $\boldsymbol{\lambda}\in\mathbb{R}^{T_{p}\times 1}$ is the Lagrange multiplier. 
        Then, $\boldsymbol{w}$ and $\boldsymbol{\lambda}$ can be solved via the Karush-Kuhn-Tucker (KKT) conditions, i.e., 
        \begin{subequations}\label{eq:lagrange-partial}
            \allowdisplaybreaks
            \begin{align}
                \frac{\partial L}{\partial \boldsymbol{\lambda}} & = \boldsymbol{p} - \boldsymbol{C}^T\boldsymbol{w} = \boldsymbol{0}_{T_{p}\times 1}, \\
                \frac{\partial L}{\partial \boldsymbol{w}} & = 2\boldsymbol{w} - 2\mu\boldsymbol{w}_{X} - \boldsymbol{C}\boldsymbol{\lambda} = \boldsymbol{0}_{N^2\times 1}. 
            \end{align}
        \end{subequations}
        By assuming that $T_{p}\ll N^2$ and $\text{rank}(\boldsymbol{C}) = T_{p}$ hold\footnote{In practical systems, the number of power measurement $T_{p}$ is generally much smaller than $N^2$ to avoid excessively large training overhead. {\color{\highlightcolor}Meanwhile, we can properly design the IRS reflection coefficients for training such that $\text{rank}(\boldsymbol{C}) = D_V = T_{p}$ for $T_{p}\ll \mathcal{D}_{N}^{(b)} < N^2$. }. }, these two equations can be solved as
        \begin{subequations}\label{eq:lagrange-solution}
            \allowdisplaybreaks
            \begin{align}
                \boldsymbol{\lambda} & = 2(\boldsymbol{C}^T\boldsymbol{C})^{-1}(\boldsymbol{p} - \mu\boldsymbol{C}^T\boldsymbol{w}_{X}), \label{eq:lagrange-lambda-solution} \\
                \boldsymbol{w} & = \boldsymbol{D}\boldsymbol{p} + \mu(\boldsymbol{I}_{N^2} - \boldsymbol{D}\boldsymbol{C}^T)\boldsymbol{w}_{X}, \label{eq:lagrange-qvec-solution}
            \end{align}
        \end{subequations}
        where $\boldsymbol{D} = \boldsymbol{C}(\boldsymbol{C}^T\boldsymbol{C})^{-1}$. 
        Thus, the optimal $\mu$, denoted as $\mu^{\star}$, can be solved by substituting~\eqref{eq:lagrange-qvec-solution} to the objective function in~\eqref{prob:cov-est-dist-min-relaxed-decomp}, i.e., 
        \begin{equation}\label{eq:optimal-scaler-mu}
                \mu^{\star} = \mathop{\arg\min_{\mu}}{\|\boldsymbol{D}\boldsymbol{p} - \mu\boldsymbol{D}\boldsymbol{C}^T\boldsymbol{w}_{X}\|_2^2}. 
        \end{equation}
        If $\boldsymbol{D}\boldsymbol{C}^T\boldsymbol{w}_{X}\neq\boldsymbol{0}_{N^2\times 1}$, $\mu^{\star}$ is obtained as
        \begin{equation}\label{eq:optimal-mu-quadra}
                \mu^{\star} = \frac{\boldsymbol{w}_{X}^T\boldsymbol{C}\boldsymbol{D}^T\boldsymbol{D}\boldsymbol{p}}{\|\boldsymbol{D}\boldsymbol{C}^T\boldsymbol{w}_{X}\|_2^2} = \frac{\boldsymbol{w}_{X}^T\boldsymbol{D}\boldsymbol{p}}{\boldsymbol{w}_{X}^T\boldsymbol{D}\boldsymbol{C}^T\boldsymbol{w}_{X}}. 
        \end{equation}
        If $\boldsymbol{D}\boldsymbol{C}^T\boldsymbol{w}_{X} = \boldsymbol{0}_{N^2\times 1}$, $\mu^{\star}$ can be any real value. 
        For simplicity, we choose $\mu^{\star} = 0$. 
        Then, the optimal solution $\boldsymbol{w}^{\star}$ can be obtained via~\eqref{eq:lagrange-qvec-solution} with $\mu = \mu^{\star}$, which yields the optimal solution for problem~\eqref{prob:cov-est-dist-min-relaxed} as $\boldsymbol{H}^{\star} = \mathcal{M}^{-1}(\boldsymbol{w}^{\star})$. 

    \vspace{-6pt}
    \subsection{Solution for Problem~\eqref{prob:cov-est-dist-min-relaxed-binary}}\label{subsec:solution-dist-min-binary}
        For $b = 1$, $\boldsymbol{H}_r$ is a real symmetric matrix, which can also be transformed into a real vector, similar to the case of $b\ge 2$. 
        By defining ${\boldsymbol{w}_r} = \mathcal{M}(\boldsymbol{H}_r)$, ${\boldsymbol{w}}_{X1} = \mathcal{M}(\boldsymbol{X_1})$, ${\boldsymbol{w}}_{X2} = \mathcal{M}(\boldsymbol{X}_2)$ and ${\boldsymbol{w}_{rt}} = \mathcal{M}(\boldsymbol{V}_t)$, $t = 1, \ldots, T_{p}$, problem~\eqref{prob:cov-est-dist-min-relaxed-binary} can be equivalently written as 
        \begin{equation}\label{prob:cov-est-dist-min-relaxed-binary-decomp}
                \mathop{\min_{{\boldsymbol{w}_r}, \boldsymbol{\mu}}}\ \|{\boldsymbol{w}_r} - {\boldsymbol{W}}_{X}\boldsymbol{\mu}\|_2^2, ~~ \mathrm{s.t.} \ {\boldsymbol{C}}^T{\boldsymbol{w}_r} = \boldsymbol{p}, 
        \end{equation}
        where ${\boldsymbol{C}} = p_0[{\boldsymbol{w}}_{r1}, \ldots, {\boldsymbol{w}}_{rT_{p}}]\in\mathbb{R}^{N^2\times T_{p}}$ and ${\boldsymbol{W}}_{X} = [{\boldsymbol{w}}_{X1}, {\boldsymbol{w}}_{X2}]\in\mathbb{R}^{N^2\times 2}$. 
        Similar to the case of $b\ge 2$, for any given $\boldsymbol{\mu}$, the optimal vector ${\boldsymbol{w}_r}$ is given by
        \begin{equation}\label{eq:lagrange-rvec-binary}
            {\boldsymbol{w}_r} = {\boldsymbol{D}}\boldsymbol{p} + (\boldsymbol{I}_{N^2} - {\boldsymbol{D}}{\boldsymbol{C}}^T){\boldsymbol{W}}_{X}\boldsymbol{\mu}, 
        \end{equation}
        where ${\boldsymbol{D}} = {\boldsymbol{C}}({\boldsymbol{C}}^T{\boldsymbol{C}})^{-1}$. 
        If ${\boldsymbol{D}}{\boldsymbol{C}}^T{\boldsymbol{W}}_{X}\neq\boldsymbol{0}_{N^2\times 2}$, the optimal $\boldsymbol{\mu}$ for problem~\eqref{prob:cov-est-dist-min-relaxed-binary-decomp} is given by
        \begin{equation}\label{eq:optimal-mu-quadra-bianry}
            \boldsymbol{\mu}^{\star} = \left({\boldsymbol{W}}_{X}^T{\boldsymbol{D}}{\boldsymbol{C}}^T{\boldsymbol{W}}_{X}\right)^{\dagger}{\boldsymbol{W}}_{X}^T{\boldsymbol{D}}\boldsymbol{p}. 
        \end{equation}
        If ${\boldsymbol{D}}{\boldsymbol{C}}^T{\boldsymbol{W}}_{X} = \boldsymbol{0}_{N^2\times 2}$, $\boldsymbol{\mu}^{\star} = \boldsymbol{0}_{2\times 1}$ is chosen. 
        Then, the optimal ${\boldsymbol{w}_r^{\star}}$ for problem~\eqref{prob:cov-est-dist-min-relaxed-binary-decomp} is obtained via~\eqref{eq:lagrange-rvec-binary} with $\boldsymbol{\mu} = \boldsymbol{\mu}^{\star}$, and the optimal solution for problem~\eqref{prob:cov-est-dist-min-relaxed-binary} is given by $\boldsymbol{H}_r^{\star} = \mathcal{M}^{-1}({\boldsymbol{w}_r^{\star}})$. 
    
    \vspace{-6pt}
    \subsection{Complexity and Convergence Analysis}\label{subsec:dist-min-convergence}
        For $b\ge 2$, the ALRA algorithm is implemented by replacing line $4$ in Algorithm~\ref{alg:RX-est-multiary} with $\boldsymbol{H}^{\star}$ given by Section~\ref{subsec:solution-dist-min-multiary}. 
        For $b = 1$, the ALRA algorithm is implemented by replacing line $4$ in Algorithm~\ref{alg:RX-est-binary} with $\boldsymbol{H}_r^{\star}$ given by Section~\ref{subsec:solution-dist-min-binary}. 
        It is worth noting that in the original LRA algorithm, the trace-minimization method used for initialization also requires SDP and may cause high computational complexity. 
        Thus, we also replace the trace-minimization initializations with solutions for problem~\eqref{prob:cov-est-dist-min-relaxed} and~\eqref{prob:cov-est-dist-min-relaxed-binary} as $\boldsymbol{X} = \boldsymbol{0}_{N^2\times N^2}$ and $\boldsymbol{X}_1 = \boldsymbol{X}_2 = \boldsymbol{0}_{N^2\times N^2}$, respectively. 
        The computational complexity for the solutions in~\eqref{eq:lagrange-qvec-solution},\eqref{eq:optimal-mu-quadra} and~\eqref{eq:lagrange-rvec-binary},\eqref{eq:optimal-mu-quadra-bianry} is given by $\mathcal{O}(N^2T_{p})$.
        Meanwhile, the optimization of $\boldsymbol{X}$ can be implemented by the power method for both $b\ge 2$ and $b = 1$ with the complexity of $\mathcal{O}(N^2)$. 
        By denoting $I_1'$ and $I_2'$ as the total number of iterations for $b\ge 2$ and $b = 1$, respectively, the total computational complexity of the ALRA algorithm is given by $\mathcal{O}(N^2T_{p}I_1')$ for $b\ge 2$ and $\mathcal{O}(N^2T_{p}I_2')$ for $b = 1$, which is much lower than those of the corresponding LRA algorithms {\color{\highlightcolor}and the trace-minimization method in~\cite{ref:PhaseLift}}. 

        Moreover, the convergence of the ALRA algorithm is guaranteed, which is analyzed as follows. 
        For $b\ge 2$, $\boldsymbol{x}^{(i)}$, $\mu^{(i)}$, and $\boldsymbol{H}^{(i)}$ are obtained for the $i$-th iteration. 
        Specifically, $\boldsymbol{x}^{(i)}$ is the eigenvector of $\boldsymbol{H}^{(i - 1)}$ corresponding to the largest eigenvalue, denoted as $\lambda_1^{(i - 1)}$. 
        Thus, matrix $\lambda_1^{(i - 1)}\boldsymbol{X}^{(i)} = \lambda_1^{(i - 1)}\boldsymbol{x}^{(i)}{\boldsymbol{x}^{(i)}}^H$ is the rank-one matrix closest to $\boldsymbol{H}^{(i - 1)}$, i.e., 
        \begin{equation}\label{eq:closest-rank-one-matrix}
            \lambda_1^{(i - 1)}\boldsymbol{X}^{(i)} = \mathop{\arg\min_{\text{rank}(\boldsymbol{A}) \le 1}}{\|\boldsymbol{H}^{(i - 1)} - \boldsymbol{A}\|_F^2}. 
        \end{equation}
        On the other hand, $\boldsymbol{H}^{(i)}$ and $\mu^{(i)}$ are jointly optimized to minimize $\|\boldsymbol{H} - \mu\boldsymbol{X}^{(i)}\|_F^2$ under constraints~\eqref{prob:cov-est-dist-min-relaxed-power} and~\eqref{prob:cov-est-dist-min-relaxed-hermitian}. 
        Therefore, for $i\ge 1$, we have
        \begin{subequations}\label{eq:dist-min-conv}
            \allowdisplaybreaks
            \begin{align}
                \left\|\boldsymbol{H}^{(i)} - \mu^{(i)}\boldsymbol{X}^{(i)}\right\|_F^2 & \le \left\|\boldsymbol{H}^{(i - 1)} - \lambda_1^{(i - 1)}\boldsymbol{X}^{(i)}\right\|_F^2 \\
                & \le \left\|\boldsymbol{H}^{(i - 1)} - \mu^{(i - 1)}\boldsymbol{X}^{(i - 1)}\right\|_F^2, 
            \end{align}
        \end{subequations}
        which guarantees the convergence of the ALRA algorithm for $b\ge 2$. 
        Similarly, the ALRA algorithm for $b = 1$ converges with a non-increasing sequence $\|\boldsymbol{H}_r^{(i)} - (\mu_1^{(i)}\boldsymbol{X}_1^{(i)} + \mu_2^{(i)}\boldsymbol{X}_2^{(i)})\|_F$, indicating that the estimated channel autocorrelation matrix approaches a rank-two matrix iteratively.

\section{Robust Estimation with Practical Power measurement}\label{sec:robust-est}
    In the above sections, the channel autocorrelation matrix was estimated based on the received signal power. 
    However, the exact value of the received signal power cannot be obtained in practical communication systems. 
    This is because the received signals are corrupted by noise. 
    Besides, the power measurement at the user can only be fed back to the central processing unit approximately after quantization, which thus incurs quantization error. 
    In this section, the robust extensions of the proposed LRA and ALRA algorithms are designed based on practical power measurement, where the effects of receiver noise and quantization error are considered. 

    \vspace{-6pt}
    \subsection{Power Measurement}\label{subsec:rsrp-protocol}
        According to~\eqref{def:signal-model}, the noisy signal power received at the user is given by 
        \begin{equation}\label{def:received-power-noisy}
            p_r = \left|\left(
                    \boldsymbol{v}^H\bar{\boldsymbol{h}}
                \right)x + z\right|^2
            = p_0\text{tr}(\bar{\boldsymbol{H}}\boldsymbol{V}) + Z, 
        \end{equation}
        where $Z = |z|^2 + 2\text{Re}((\boldsymbol{v}^{H}\bar{\boldsymbol{h}})xz^*)$ is the noise power satisfying $\mathbb{E}[Z] = \sigma^2$. 
        By averaging $N_{0}$ noisy power values, each practical power measurement is given by 
        \begin{equation}\label{def:RSRP-power}
            q = \frac{1}{N_{0}}\sum_{i = 1}^{N_{0}}{\left|\left(\boldsymbol{v}^{H}\bar{\boldsymbol{h}}\right)x + z_i\right|^2} = p_0\text{tr}(\bar{\boldsymbol{H}}\boldsymbol{V}) + \frac{1}{N_{0}}\sum_{i = 1}^{N_{0}}{Z_i}, 
        \end{equation}
        where $z_i\sim\mathcal{CN}(0, \sigma^2)$, $i = 1, \ldots, N_{0}$, are i.i.d. random noise variables and $Z_i = |z_i|^2 + 2\text{Re}((\boldsymbol{v}^{H}\bar{\boldsymbol{h}})xz_i^*)$. 
        Define $\bar{Z} = \frac{1}{N_{0}}\sum_{i = 1}^{N_{0}}{Z_i}$ as the average noise power in each power measurement. 
        If $N_{0}$ is sufficiently large, $\bar{Z}$ converges to a Gaussian random variable with expectation $\mathbb{E}[\bar{Z}] = \sigma^2$, according to the central limit theorem. 

        Despite the above model of practical power measurement, its value cannot be exactly acquired because of the quantization. 
        For example, in the current protocol~\cite{ref:3gpp:36.133}, the RSRP value is mapped into equally quantized levels within the range from $-156$ dBm to $-44$ dBm. 
        Thus, to investigate the effect of power quantization mapping on the performance of channel autocorrelation matrix estimation, we split the range $[-156, -44]$ dBm into $M$ equal levels and define $D = 112/M$ dB as the interval for each level. 
        As such, the set of quantization power levels is given by $\Omega_{D} = \{-44 - (l + 1/2)D~$\text{dBm}$, l = 0, \ldots, M - 1\}$. 
        For the $l$-th level, the index $l$ is reported to the central processing unit if the power value $q$ lies between $-44 - (l + 1)D$ dBm and $-44 - lD$ dBm. 


    \vspace{-6pt}
    \subsection{Reformulation of Channel Estimation Problem}\label{subsec:est-prob-formulation}
        Denote the actual power value for the $t$-th measurement as $q_t = p_0\text{tr}(\bar{\boldsymbol{H}}\boldsymbol{V}_t) + \bar{Z}_t$, where $\bar{Z}_t$ is the noise power. 
        Given the quantized power level index $l_t$, we have $\zeta_t\le q_t\le \xi_t$, where $\zeta_t = 10^{-7.4 - (l_t + 1)D/10}$ and $\xi_t = 10^{-7.4 - l_tD/10}$ are the lower and upper bounds of the $l_t$-th power level. 
        For $b\ge 2$, the channel autocorrelation matrix estimation problem is reformulated as
        \begin{subequations}\label{prob:cov-est-robust-origin}
            \begin{align}
                & \mathop{\min_{\boldsymbol{H}, \boldsymbol{\delta}\ge\boldsymbol{0}_{T_{p}\times 1}}} \ {\|\boldsymbol{\delta}\|_2^2} \tag{\ref{prob:cov-est-robust-origin}} \\
                & ~ \mathrm{s.t.} \ \zeta_t - \delta_t\le p_0\text{tr}(\boldsymbol{H}\boldsymbol{V}_t) + \sigma^2 \le \xi_t + \delta_t, \forall t, \label{prob:cov-est-robust-origin-power} \\
                & ~~~~~~ \text{rank}(\boldsymbol{H}) = 1, ~\boldsymbol{H}\in\mathbb{S}_{+}^{N}, \label{prob:cov-est-robust-origin-rank-semidefinite}
            \end{align}
        \end{subequations}
        where $\boldsymbol{\delta} = [\delta_1, \ldots, \delta_{T_{p}\times 1}]\in\mathbb{R}^{T_{p}\times 1}$ is a slack vector which is introduced to tackle the uncertainty of noise power $\bar{Z}_t$. 
        For $b = 1$, the problem is given by 
        \begin{subequations}\label{prob:cov-est-robust-origin-binary}
            \begin{align}
                & \mathop{\min_{\boldsymbol{H}_r, \boldsymbol{\delta}\ge\boldsymbol{0}_{T_{p}\times 1}}} \ {\|\boldsymbol{\delta}\|_2^2} \tag{\ref{prob:cov-est-robust-origin-binary}} \\
                & ~ \mathrm{s.t.} \ \zeta_t - \delta_t\le p_0\text{tr}(\boldsymbol{H}_r\boldsymbol{V}_t) + \sigma^2 \le \xi_t + \delta_t, \forall t, \label{prob:cov-est-robust-origin-binary-power} \\
                & ~~~~~~ \text{rank}(\boldsymbol{H}_r) \le 2, ~\boldsymbol{H}_r\in\mathbb{M}_{+}^{N}. \label{prob:cov-est-robust-origin-rank-binary-semidefinite}
            \end{align}
        \end{subequations}
        It should be noticed that problems~\eqref{prob:cov-est-robust-origin} and~\eqref{prob:cov-est-robust-origin-binary} are an non-convex optimization problems because of the rank constraints, whose optimal solutions cannot be solved by existing optimization tools in polynomial time. 
        In the following, we generalize the proposed LRA and ALRA algorithms to solve problems~\eqref{prob:cov-est-robust-origin} and~\eqref{prob:cov-est-robust-origin-binary} for achieving robust estimation of the channel autocorrelation matrix.

        \begin{algorithm}[t]
            \begin{minipage}{0.95\linewidth}
            \centering
            \caption{Robust LRA algorithm for $b\ge 2$. }
            \begin{algorithmic}[1]\label{alg:RX-est-multiary-robust}
                \REQUIRE~$\{\boldsymbol{v}_t$, $t = 1, \ldots, T\}$, $\boldsymbol{l}\in\mathbb{R}^{T\times 1}$, noise variance $\sigma^2$, penalty parameter $\rho$, and threshold $\epsilon$. 
                \STATE~Initialization: Compute $\boldsymbol{\zeta}$ and $\boldsymbol{\xi}$; solve $\boldsymbol{H}^{(0)}$ and $\boldsymbol{\delta}^{(0)}$ from problem~\eqref{prob:cov-est-tracemin-robust-multiary}; iteration index $i\gets 1$. 
                \WHILE{$g(\boldsymbol{H}^{(i - 1)}) \le \epsilon$}
                    \STATE~Let $\boldsymbol{x}^{(i)}$ be the normalized eigenvector of $\boldsymbol{H}^{(i -1)}$ corresponding to the largest eigenvalue, and $\boldsymbol{X}^{(i)} \gets \boldsymbol{x}^{(i)}{\boldsymbol{x}^{(i)}}^H$. 
                    \STATE~Given $\boldsymbol{X} = \boldsymbol{X}^{(i)}$, solve $\boldsymbol{G}^{(i)}$ and $\gamma_1$ from problem~\eqref{prob:cov-est-maxmax-cov-robust-multiary-fp} with $\boldsymbol{\delta} = \boldsymbol{\delta}^{(i - 1)}$ fixed, and then solve $\boldsymbol{\delta}^{(i)}$ and $\gamma^{(i)}$ with $\boldsymbol{G} = \boldsymbol{G}^{(i)}$ fixed. Obtain $\boldsymbol{H}^{(i)} \gets \boldsymbol{G}^{(i)}/\gamma^{(i)}$. 
                    \STATE~$i\gets i + 1$. 
                \ENDWHILE
                \RETURN~The estimated matrix $\hat{\boldsymbol{H}}\gets\boldsymbol{H}^{(i - 1)}$. 
            \end{algorithmic}
            \end{minipage}
        \end{algorithm}

    \vspace{-6pt}
    \subsection{Robust LRA Algorithm}\label{subsec:ratio-max-est-robust}
    \vspace{-6pt}
        In this subsection, the robust extension of the LRA algorithm is proposed. 
        The eigenvalue-ratio function defined in~\eqref{def:lambda-ratio-func} and the generalized eigenvalue-ratio function defined in~\eqref{def:generalized-lambda-ratio-func} are added to the objective functions of problems~\eqref{prob:cov-est-robust-origin} and~\eqref{prob:cov-est-robust-origin-binary}, respectively, as relaxations for the rank constraint. 
        Specifically, for $b\ge 2$, problem~\eqref{prob:cov-est-robust-origin} is relaxed to
        \begin{subequations}\label{prob:cov-est-maxmax-robust-multiary}
            \allowdisplaybreaks
            \begin{align}
                & \mathop{\max_{\boldsymbol{H}, \boldsymbol{\delta}\ge\boldsymbol{0}_{T_{p}\times 1}}\max_{\|\boldsymbol{x}\|_2\le 1}} \ f(\boldsymbol{H}, \boldsymbol{x}) - \rho\|\boldsymbol{\delta}\|_2^2 \tag{\ref{prob:cov-est-maxmax-robust-multiary}} \\
                & ~ \mathrm{s.t.} \ \zeta_t - \delta_t\le p_0\text{tr}(\boldsymbol{H}\boldsymbol{V}_t) + \sigma^2\le \xi_t + \delta_t, \ \forall t, \label{prob:cov-est-maxmax-robust-multiary-power} \\
                & ~~~~~~ {\boldsymbol{H}}\in\mathbb{S}_{+}^{N}, \label{prob:cov-est-maxmax-robust-multiary-semidefinite} 
            \end{align}
        \end{subequations}
        where $\rho > 0$ is a predefined parameter. 
        Compared to problem~\eqref{prob:cov-est-maxmax}, problem~\eqref{prob:cov-est-maxmax-robust-multiary} approaches a low-rank solution and minimizes the penalty term $\|\boldsymbol{\delta}\|_2^2$ simultaneously, thus leading to a robust estimation of the channel autocorrelation matrix. 
        The alternating optimization method can be applied similar to that for solving problem~\eqref{prob:cov-est-maxmax}, while the optimization for $\boldsymbol{H}$ should be modified to decrease $\|\boldsymbol{\delta}\|_2^2$. 
        Given matrix $\boldsymbol{X} = \boldsymbol{x}\boldsymbol{x}^H$, problem~\eqref{prob:cov-est-maxmax-robust-multiary} is simplified as 
        \begin{equation}\label{prob:cov-est-maxmax-cov-robust-multiary}
            \mathop{\max_{\boldsymbol{H}, \boldsymbol{\delta}\ge\boldsymbol{0}_{T_{p}\times 1}}} \ \frac{\text{tr}(\boldsymbol{H}\boldsymbol{X})}{\text{tr}(\boldsymbol{H})} - \rho\|\boldsymbol{\delta}\|_2^2, ~~ \mathrm{s.t.} \ \eqref{prob:cov-est-maxmax-robust-multiary-power},\eqref{prob:cov-est-maxmax-robust-multiary-semidefinite}. 
        \end{equation}
        With $\boldsymbol{G} = \boldsymbol{H}/\text{tr}(\boldsymbol{H})$ and $\gamma = 1/\text{tr}(\boldsymbol{H})$, problem~\eqref{prob:cov-est-maxmax-cov-robust-multiary} can be transformed into
        \begin{subequations}\label{prob:cov-est-maxmax-cov-robust-multiary-fp}
            \allowdisplaybreaks
            \begin{align}
                & \mathop{\max_{\boldsymbol{G}, \boldsymbol{\delta}\ge\boldsymbol{0}_{T_{p}\times 1}, \gamma > 0}} \ \text{tr}(\boldsymbol{G}\boldsymbol{X}) - \rho\|\boldsymbol{\delta}\|_2^2 \tag{\ref{prob:cov-est-maxmax-cov-robust-multiary-fp}} \\
                & ~ \mathrm{s.t.} \ \zeta_t - \delta_t\le p_0\text{tr}(\boldsymbol{G}\boldsymbol{V}_t)/\gamma + \sigma^2 \le \xi_t + \delta_t, \ \forall t, \label{prob:cov-est-maxmax-cov-robust-multiary-fp-power} \\
                & ~~~~~~ {\boldsymbol{G}}\in\mathbb{S}_{+}^{N}, \ \text{tr}(\boldsymbol{G}) = 1. \label{prob:cov-est-maxmax-cov-robust-multiary-fp-semidefinite}
            \end{align}
        \end{subequations}
        Although problem~\eqref{prob:cov-est-maxmax-cov-robust-multiary-fp} is nonconvex, it can be transformed into a convex one with respect to any two of the variables $\boldsymbol{G}, \boldsymbol{\delta}$, and $\gamma$, with the third variable given and fixed. 
        To obtain an approximate solution, we first solve $\boldsymbol{G}$ and $\gamma$ in problem~\eqref{prob:cov-est-maxmax-cov-robust-multiary-fp} with $\boldsymbol{\delta}$ being fixed to maximize the eigenvalue-ratio function, and then solve $\gamma$ and $\boldsymbol{\delta}$ with $\boldsymbol{G}$ being fixed to minimize the penalty term $\|\boldsymbol{\delta}\|_2^2$. 
        For initialization, the robust trace-minimization problem is employed as
        \begin{equation}\label{prob:cov-est-tracemin-robust-multiary}
                \mathop{\min_{\boldsymbol{H}, \boldsymbol{\delta}\ge\boldsymbol{0}_{T_{p}\times 1}}} \ \text{tr}(\boldsymbol{H}) + \rho\|\boldsymbol{\delta}\|_2^2, ~~ \mathrm{s.t.} \ \eqref{prob:cov-est-maxmax-robust-multiary-power},\eqref{prob:cov-est-maxmax-robust-multiary-semidefinite}. 
        \end{equation}
        The robust LRA algorithm for $b\ge 2$ is summarized in Algorithm~\ref{alg:RX-est-multiary-robust}. 
        It is easy to verify that its computational complexity is $\mathcal{O}(N^{4.5}I_3)$, which is the same as that of Algorithm~\ref{alg:RX-est-multiary}, and the convergence of the algorithm is guaranteed, where $I_3$ is the total number of iterations for Algorithm~\ref{alg:RX-est-multiary-robust}. 
        For $i\ge 1$, $(\boldsymbol{G}^{(i)}, \gamma_1)$ is solved from problem~\eqref{prob:cov-est-maxmax-cov-robust-multiary-fp} with $\boldsymbol{\delta} = \boldsymbol{\delta}^{(i - 1)}$ being fixed, i.e., $\text{tr}(\boldsymbol{G}\boldsymbol{X}^{(i)})$ is maximized subject to constraints~\eqref{prob:cov-est-maxmax-cov-robust-multiary-fp-power} and~\eqref{prob:cov-est-maxmax-cov-robust-multiary-fp-semidefinite} with $\boldsymbol{\delta} = \boldsymbol{\delta}^{(i - 1)}$. 
        Meanwhile, since $(\boldsymbol{G}^{(i - 1)}, \gamma^{(i - 1)}, \boldsymbol{\delta}^{(i - 1)})$ is feasible to problem~\eqref{prob:cov-est-maxmax-cov-robust-multiary-fp}, we have $\text{tr}(\boldsymbol{G}^{(i)}\boldsymbol{X}^{(i)}) \ge \text{tr}(\boldsymbol{G}^{(i - 1)}\boldsymbol{X}^{(i)})$. 
        Additionally, it is easy to verify that $f(\boldsymbol{H}^{(i)}, \boldsymbol{x}) = \text{tr}(\boldsymbol{G}^{(i)}\boldsymbol{X})$ for any $\boldsymbol{x}$ and $\boldsymbol{X} = \boldsymbol{x}\boldsymbol{x}^H$. 
        Therefore, we have
        \begin{subequations}\label{prob:cov-est-robust-conv-multiary}
            \allowdisplaybreaks
            \begin{align}
                g(\boldsymbol{H}^{(i)}) & \ge f(\boldsymbol{H}^{(i)}, \boldsymbol{x}^{(i)}) = \text{tr}(\boldsymbol{G}^{(i)}\boldsymbol{X}^{(i)}) \\
                & \ge \text{tr}(\boldsymbol{G}^{(i - 1)}\boldsymbol{X}^{(i)}) = f(\boldsymbol{H}^{(i - 1)}, \boldsymbol{x}^{(i)})  \\
                & = g(\boldsymbol{H}^{(i - 1)}), ~ i = 1, \ldots, I_3, 
            \end{align}
        \end{subequations}
        which guarantees the convergence of Algorithm~\ref{alg:RX-est-multiary-robust}.

        Similarly, for $b = 1$, problem~\eqref{prob:cov-est-robust-origin-binary} is relaxed to 
        \begin{subequations}\label{prob:cov-est-maxmax-robust-binary}
            \allowdisplaybreaks
            \begin{align}
                & \mathop{\max_{\boldsymbol{H}_r, \boldsymbol{\delta}\ge\boldsymbol{0}_{T_{p}\times 1}}\max_{\boldsymbol{x}_1, \boldsymbol{x}_2}} \ f_r(\boldsymbol{H}_r, \boldsymbol{x}_1, \boldsymbol{x}_2) - \rho\|\boldsymbol{\delta}\|_2^2 \tag{\ref{prob:cov-est-maxmax-robust-binary}} \\
                & ~ \mathrm{s.t.} \ \|\boldsymbol{x}_1\|_2 = 1,~\|\boldsymbol{x}_2\|_2 = 1,~\boldsymbol{x}_1^T\boldsymbol{x}_2 = 0, \label{prob:eigenvec-ortho} \\
                & ~~~~~~ \zeta_t - \delta_t\le p_0\text{tr}(\boldsymbol{H}_r\boldsymbol{V}_t) + \sigma^2\le \xi_t + \delta_t, \ \forall t, \label{prob:cov-est-maxmax-robust-binary-power} \\
                & ~~~~~~ {\boldsymbol{H}_r}\in\mathbb{M}_{+}^{N}. \label{prob:cov-est-maxmax-robust-binary-semidefinite} 
            \end{align}
        \end{subequations}
        By applying alternating optimization, the optimization of $\boldsymbol{H}_r$ given matrix $\boldsymbol{X} = \boldsymbol{x}_1\boldsymbol{x}_1^T + \boldsymbol{x}_2\boldsymbol{x}_2^T$ is expressed as
        \begin{equation}\label{prob:cov-est-maxmax-cov-robust-binary}
                \mathop{\max_{\boldsymbol{H}_r, \boldsymbol{\delta}\ge\boldsymbol{0}_{T_{p}\times 1}}} \ \frac{\text{tr}(\boldsymbol{H}_r\boldsymbol{X})}{\text{tr}(\boldsymbol{H}_r)} - \rho\|\boldsymbol{\delta}\|_2^2, ~~ \mathrm{s.t.} \ \eqref{prob:cov-est-maxmax-robust-binary-power},~\eqref{prob:cov-est-maxmax-robust-binary-semidefinite}. 
        \end{equation}
        Problem~\eqref{prob:cov-est-maxmax-cov-robust-binary} has the same form as problem~\eqref{prob:cov-est-maxmax-cov-robust-multiary} and thus it can be solved similarly. 
        For initialization, the following problem is considered: 
        \begin{equation}\label{prob:cov-est-tracemin-robust-binary}
            \mathop{\max_{\boldsymbol{H}_r, \boldsymbol{\delta}\ge\boldsymbol{0}_{T_{p}\times 1}}} \ \text{tr}(\boldsymbol{H}_r) + \rho\|\boldsymbol{\delta}\|_2^2, ~~ \mathrm{s.t.} \ \eqref{prob:cov-est-maxmax-robust-binary-power},~\eqref{prob:cov-est-maxmax-robust-binary-semidefinite}. 
        \end{equation}
        The robust LRA algorithm for $b = 1$ is summarized in Algorithm~\ref{alg:RX-est-binary-robust}. 
        By denoting the total number of iterations as $I_4$, its computational complexity can be obtained as $\mathcal{O}(N^{4.5}I_4)$, and the convergence of Algorithm~\ref{alg:RX-est-binary-robust} is ensured by the non-decreasing property of the sequence $g_r(\boldsymbol{H}_r^{(i)})$, $i = 1, \ldots, I_4$. 
        
        \begin{algorithm}[t]
            \begin{minipage}{0.95\linewidth}
            \centering
            \caption{Robust LRA algorithm for $b = 1$. }
            \begin{algorithmic}[1]\label{alg:RX-est-binary-robust}
                \REQUIRE~$\{\boldsymbol{v}_t$, $t = 1, \ldots, T\}$, $\boldsymbol{l}\in\mathbb{R}^{T\times 1}$, noise variance $\sigma^2$, penalty parameter $\rho$, and threshold $\epsilon$. 
                \STATE~Initialization: Compute $\boldsymbol{\zeta}$ and $\boldsymbol{\xi}$; solve $\boldsymbol{H}_r^{(0)}$ and $\boldsymbol{\delta}^{(0)}$ from problem~\eqref{prob:cov-est-tracemin-robust-binary}; iteration index $i\gets 1$. 
                \WHILE{$g(\boldsymbol{H}^{(i - 1)}) \le \epsilon$}
                    \STATE~Let $\boldsymbol{x}_1^{(i)}$ and $\boldsymbol{x}_2^{(i)}$ be the normlized eigenvectors of $\boldsymbol{H}_r^{(i - 1)}$ corresponding to the first and second largest eigenvalues, and $\boldsymbol{X}^{(i)} \gets \boldsymbol{x}_1^{(i)}{\boldsymbol{x}_1^{(i)}}^T + \boldsymbol{x}_2^{(i)}{\boldsymbol{x}_2^{(i)}}^T$. 
                    \STATE~Given $\boldsymbol{X} = \boldsymbol{X}^{(i)}$, solve $\boldsymbol{H}_r^{(i)}$ from problem~\eqref{prob:cov-est-maxmax-cov-robust-binary}. 
                    \STATE~$i\gets i + 1$. 
                \ENDWHILE
                \RETURN~The estimated matrix $\hat{\boldsymbol{H}}_r\gets\boldsymbol{H}_r^{(i - 1)}$. 
            \end{algorithmic}
            \end{minipage}
        \end{algorithm}

    \vspace{-6pt}
    \subsection{Robust ALRA Algorithm}\label{subsec:dist-min-est-robust}
        In this subsection, the robust ALRA algorithm is proposed with a lower computational complexity. 
        Similar to that in Section~\ref{sec:proposed-dist-min}, the Euclidean distance functions are employed as approximations of the eigenvalue-ratio functions of the robust LRA algorithm. 
        Specifically, given $\boldsymbol{X}$ for $b\ge 2$, the optimization of $\boldsymbol{H}$ (i.e., problem~\eqref{prob:cov-est-maxmax-cov-robust-multiary}) is approximated by 
        \begin{subequations}\label{prob:cov-est-dist-min-robust-multiary}
            \begin{align}
                & \mathop{\min_{\boldsymbol{H}, \boldsymbol{\delta}\ge\boldsymbol{0}_{T_{p}\times 1}, \mu}}\ \|\boldsymbol{H} - \mu\boldsymbol{X}\|_F^2 + \rho\|\boldsymbol{\delta}\|_2^2 \tag{\ref{prob:cov-est-dist-min-robust-multiary}} \\
                & ~ \mathrm{s.t.} \ \zeta_t - \delta_t \le p_0\text{tr}(\boldsymbol{H}\boldsymbol{V}_t) + \sigma^2 \le \xi_t + \delta_t, \forall t, \label{prob:cov-est-dist-min-robust-multiary-power} \\
                & ~~~~~~ \boldsymbol{H}^H = \boldsymbol{H}. \label{prob:cov-est-dist-min-robust-multiary-hermitian}
            \end{align}
        \end{subequations}
        Note that problem~\eqref{prob:cov-est-dist-min-robust-multiary} is a quadratic minimization problem with inequality constraints, which is convex but the optimal solution cannot be obtained in closed form. 
        To derive an explicit solution, we define $\hat{q}_t = (\zeta_t + \xi_t) / 2$, and approximate constraint~\eqref{prob:cov-est-dist-min-robust-multiary-power} with $\hat{q}_t - \delta_t \le p_0\text{tr}(\boldsymbol{H}\boldsymbol{V}_t) + \sigma^2 \le \hat{q}_t + \delta_t$, $\forall t$. 
        Then, problem~\eqref{prob:cov-est-dist-min-robust-multiary} can be further approximated by 
        \begin{subequations}\label{prob:cov-est-dist-min-approx-robust-multiary}
            \begin{align}
                & \mathop{\min_{\boldsymbol{H}, \mu}}\ \|\boldsymbol{H} - \mu\boldsymbol{X}\|_F^2 + \rho\sum_{t = 1}^{T_{p}}{\left|p_0\text{tr}(\boldsymbol{H}\boldsymbol{V}_t) + \sigma^2 - \hat{q}_t\right|^2} \tag{\ref{prob:cov-est-dist-min-approx-robust-multiary}} \\
                & ~ \mathrm{s.t.} \ \boldsymbol{H}^H = \boldsymbol{H}, \label{prob:cov-est-dist-min-approx-robust-multiary-hermitian}
            \end{align}
        \end{subequations}
        which is a quadratic minimization problem with linear constraints. 
        By leveraging the transformation $\mathcal{M}$, problem~\eqref{prob:cov-est-dist-min-approx-robust-multiary} can be equivalently written as 
        \begin{equation}\label{prob:dist-min-approx-robust-multiary-decomp}
                \mathop{\min_{\boldsymbol{w}, \mu}}\ \|\boldsymbol{w} - \mu\boldsymbol{w}_{X}\|_2^2 + \rho\|\boldsymbol{C}^T\boldsymbol{w} - (\hat{\boldsymbol{q}} - \sigma^2\boldsymbol{1}_{T_{p}})\|_2^2, 
        \end{equation}
        with $\hat{\boldsymbol{q}} = [\hat{q}_1, \ldots, \hat{q}_{T_{p}}]^T$. 
        By taking the first derivative of the objective function, the optimal solutions can be obtained in closed form as 
        \begin{subequations}\label{eq:dist-min-approx-robust-multiary-solution}
            \allowdisplaybreaks
            \begin{align}
                \mu^{\star} & = \frac{\rho}{\|\boldsymbol{\psi}\|_2^2}\boldsymbol{\psi}^T\boldsymbol{\Upsilon}\boldsymbol{C}(\hat{\boldsymbol{q}} - \sigma^2\boldsymbol{1}_{T_{p}}), \\
                \boldsymbol{w}^{\star} & = \rho\boldsymbol{\Upsilon}\boldsymbol{C}(\hat{\boldsymbol{q}} - \sigma^2\boldsymbol{1}_{T_{p}}) + \mu^{\star}\boldsymbol{\Upsilon}\boldsymbol{w}_{X}, \\
                \boldsymbol{\psi} & = \left(\boldsymbol{I}_{N^2} - \boldsymbol{\Upsilon}\right)\boldsymbol{w}_{X}, \\
                \boldsymbol{\Upsilon} & = \left(\boldsymbol{I}_{N^2} + \rho\boldsymbol{C}\boldsymbol{C}^T\right)^{-1}, 
            \end{align}
        \end{subequations}
        and thus the optimal solution for problem~\eqref{prob:cov-est-dist-min-approx-robust-multiary} is given by $\boldsymbol{H}^{\star} = \mathcal{M}^{-1}(\boldsymbol{w}^{\star})$. 
        Therefore, the robust ALRA algorithm for $b\ge 2$ can be implemented by replacing line $4$ in Algorithm~\ref{alg:RX-est-multiary-robust} with the optimal solution for problem~\eqref{prob:cov-est-dist-min-approx-robust-multiary} given $\boldsymbol{X}$ from the previous iteration, i.e., $\boldsymbol{H}^{\star} = \mathcal{M}^{-1}(\boldsymbol{w}^{\star})$, where $\boldsymbol{w}^{\star}$ is given by~\eqref{eq:dist-min-approx-robust-multiary-solution}. 
        Besides, the optimal solution $\boldsymbol{H}^{\star}$ for problem~\eqref{prob:cov-est-dist-min-approx-robust-multiary} given $\boldsymbol{X} = \boldsymbol{0}_{N^2\times N^2}$ is employed as the initialization. 
        
        Similarly, for $b = 1$, the robust estimation of $\boldsymbol{H}_r$ can be realized by converting problem~\eqref{prob:cov-est-dist-min-relaxed-binary} to the following problem:
        \begin{subequations}\label{prob:cov-est-dist-min-approx-robust-binary}
            \begin{align}
                & \begin{aligned}
                    \mathop{\min_{\boldsymbol{H}_r, \boldsymbol{\mu}}}\ & \|\boldsymbol{H}_r - (\mu_1\boldsymbol{X}_1 + \mu_2\boldsymbol{X}_2)\|_F^2 \\
                    & ~~~~~~~~~~ + \rho\sum_{t = 1}^{T_{p}}{\left|p_0\text{tr}(\boldsymbol{H}_r\boldsymbol{V}_t) + \sigma^2 - \hat{q}_t\right|^2}
                \end{aligned}
                \tag{\ref{prob:cov-est-dist-min-approx-robust-binary}} \\
                & ~ \mathrm{s.t.} \ \boldsymbol{H}_r^T = \boldsymbol{H}_r. \label{prob:cov-est-dist-min-approx-robust-binary-hermitian}
            \end{align}
        \end{subequations}
        According to the derivations for solving problem~\eqref{prob:cov-est-dist-min-approx-robust-multiary}, the optimal solution for problem~\eqref{prob:cov-est-dist-min-approx-robust-binary} can be obtained as $\boldsymbol{H}_r^{\star} = \mathcal{M}^{-1}(\boldsymbol{w}_r^{\star})$, where $\boldsymbol{w}_r^{\star}$ is solved from 
        \begin{equation}\label{prob:dist-min-approx-robust-binary-decomp}
            \mathop{\min_{{\boldsymbol{w}_r}, \boldsymbol{\mu}}}\ \|{\boldsymbol{w}_r} - {\boldsymbol{W}}_{X}\boldsymbol{\mu}\|_2^2 + \rho\|{\boldsymbol{C}}^T{\boldsymbol{w}_r} - (\hat{\boldsymbol{q}} - \sigma^2\boldsymbol{1}_{T_{p}})\|_2^2. 
        \end{equation}
        Specifically, the optimal solutions ${\boldsymbol{w}}_r^{\star}$ and $\boldsymbol{\mu}^{\star}$ are given by
        \begin{subequations}\label{eq:dist-min-approx-robust-binary-solution}
            \allowdisplaybreaks
            \begin{align}
                \boldsymbol{\mu}^{\star} & = \rho\left({\boldsymbol{\Psi}}^T{\boldsymbol{\Psi}}\right)^{\dagger}{\boldsymbol{\Psi}}^T{\boldsymbol{\Upsilon}}{\boldsymbol{C}}(\hat{\boldsymbol{q}} - \sigma^2\boldsymbol{1}_{T_{p}}), \\
                {\boldsymbol{w}}_r^{\star} & = \rho{\boldsymbol{\Upsilon}}{\boldsymbol{C}}(\hat{\boldsymbol{q}} - \sigma^2\boldsymbol{1}_{T_{p}}) + {\boldsymbol{\Upsilon}}{\boldsymbol{W}}_{X}\boldsymbol{\mu}^{\star}, \\
                {\boldsymbol{\Psi}} & = \left(\boldsymbol{I}_{N^2} - {\boldsymbol{\Upsilon}}\right){\boldsymbol{W}}_{X}, \\
                {\boldsymbol{\Upsilon}} & = \left(\boldsymbol{I}_{N^2} + \rho{\boldsymbol{C}}{\boldsymbol{C}}^T\right)^{-1}. 
            \end{align}
        \end{subequations}
        Thus, the robust ALRA algorithm for $b = 1$ is implemented by replacing line $4$ in Algorithm~\ref{alg:RX-est-binary-robust} with the optimal solution for problem~\eqref{prob:cov-est-dist-min-approx-robust-binary} given $\boldsymbol{X}_1$ and $\boldsymbol{X}_2$ from the previous iteration, i.e., $\boldsymbol{H}_r^{\star} = \mathcal{M}^{-1}(\boldsymbol{w}_r^{\star})$, where $\boldsymbol{w}_r^{\star}$ is given by~\eqref{eq:dist-min-approx-robust-binary-solution}. 
        The initialization is obtained as $\boldsymbol{H}_r^{\star}$ with $\boldsymbol{X}_1 = \boldsymbol{X}_2 = \boldsymbol{0}_{N^2\times N^2}$. 

        Note that $\boldsymbol{\Upsilon}\in\mathbb{R}^{N^2\times N^2}$ can be calculated offline since $\rho$ and $\boldsymbol{C}$ are known. 
        Thus, the online complexity of the robust estimation for the channel autocorrelation matrix is dominated by that of the multiplication by $\boldsymbol{\Upsilon}$ to vectors, which is $\mathcal{O}(N^4)$. 
        To further reduce the complexity, we apply the Woodbury matrix identity~\cite{ref:woodbury-identity} to $\boldsymbol{\Upsilon}$ such that we have
        \begin{subequations}
            \allowdisplaybreaks
            \begin{align}
                \boldsymbol{\Upsilon} & = \left(\boldsymbol{I}_{N^2} + \rho\boldsymbol{C}\boldsymbol{C}^T\right)^{-1} \\
                & = \boldsymbol{I}_{N^2} - \boldsymbol{C}\left({\rho}^{-1}\boldsymbol{I}_{T_{p}^2} + \boldsymbol{C}^T\boldsymbol{C}\right)^{-1}\boldsymbol{C}^T. 
            \end{align}
        \end{subequations}
        Given $\boldsymbol{C}\in\mathbb{R}^{N^2\times T_{p}}$, matrix $\boldsymbol{S} = ({\rho}^{-1}\boldsymbol{I}_{T_{p}^2} + \boldsymbol{C}^T\boldsymbol{C})^{-1}\in\mathbb{R}^{T_{p}\times T_{p}}$ can be calculated offline. 
        Then, the multiplication to a vector by $\boldsymbol{\Upsilon} = \boldsymbol{I}_{N^2} - \boldsymbol{C}\boldsymbol{S}\boldsymbol{C}^T$ can be implemented with complexity $\mathcal{O}(N^2T_{p})$. 
        Such implementation can be applied to both cases of $b \ge 2$ and $b = 1$. 
        Therefore, the computational complexity of the robust ALRA algorithm is $\mathcal{O}(N^2T_{p}I_3')$ for $b\ge 2$ and $\mathcal{O}(N^2T_{p}I_4')$ for $b = 1$, where $I_3'$ and $I_4'$ are the total number of iterations for the robust ALRA algorithm for $b\ge 2$ and $b = 1$, respectively. 
        {\color{\highlightcolor}The complexity is thus significantly reduced compared to the robust LRA algorithms and trace-minimization method~\cite{ref:PhaseLift}. }

        Next, the convergence of the proposed robust ALRA algorithm is analyzed. 
        For $b\ge 2$, it can be shown that the objective function of problem~\eqref{prob:cov-est-dist-min-approx-robust-multiary}, which is denoted as $\varphi(\boldsymbol{H}, \mu, \boldsymbol{X})$ and can be seen as a penaltized version of the squared distance function $d_e^2(\boldsymbol{H}, \mu\boldsymbol{X})$, is non-increasing during the iterations. 
        In the $i$-th iteration for $b\ge 2$, $\boldsymbol{x}^{(i)}$ is the eigenvector of $\boldsymbol{H}^{(i - 1)}$ corresponding to the largest eigenvalue $\lambda_1^{(i - 1)}$, and thus~\eqref{eq:closest-rank-one-matrix} still holds, i.e., $\lambda_1^{(i - 1)}\boldsymbol{X}^{(i)}$ is the rank-one matrix closest to $\boldsymbol{H}^{(i - 1)}$. 
        Moreover, $\boldsymbol{H}^{(i)}$ and $\mu^{(i)}$ are solved optimally from problem~\eqref{prob:cov-est-dist-min-approx-robust-multiary} given $\boldsymbol{X}^{(i)}$. 
        Thus, it is easy to verify that 
        \begin{subequations}\label{eq:converg-dm-robust-multiary}
            \begin{align}
                \varphi(\boldsymbol{H}^{(i)}, \mu^{(i)}, \boldsymbol{X}^{(i)}) & \le \varphi(\boldsymbol{H}^{(i - 1)}, \lambda_1^{(i - 1)}, \boldsymbol{X}^{(i)}) \\
                & \le \varphi(\boldsymbol{H}^{(i - 1)}, \mu^{(i - 1)}, \boldsymbol{X}^{(i - 1)})
            \end{align}
        \end{subequations}
        always holds for $i = 1, \cdots, I_3'$. 
        Since $\varphi(\boldsymbol{H}, \mu, \boldsymbol{X}) \ge 0$, the robust ALRA algorithm converges for $b\ge 2$. 
        The similar analysis applies to the case of $b = 1$, where the objective function of problem~\eqref{prob:cov-est-dist-min-approx-robust-binary}, denoted as $\varphi_r(\boldsymbol{H}_r, \boldsymbol{\mu}, \boldsymbol{X}_1, \boldsymbol{X}_2)$, is also non-increasing during the iterations.


\section{Performance Evaluation}\label{sec:performance-evaluation}
    \subsection{Simulation Setup}\label{subsec:setup}
        In the simulation, the BS and the IRS are located at $(50, -150, 20)$ and $(-2, -1, 0)$ in meters (m) in a three-dimensional coordinate system, respectively. 
        The location of the user is randomly generated as $(x_{u}, y_{u}, 0)$ with $0\le x_u, y_u\le 10$. 
        The size of IRS is set as $N_x\times N_z = 8\times 8 = 64$, and thus we have $N = 64 + 1 = 65$. 
        The path loss model for all channels is given by $\eta = C_{0}d^{-\alpha}$, where $d$ is the signal propagation distance, while $C_{0}$ and $\alpha$ are the channel gain at the reference distance of $1$ m and the path-loss exponent, which are denoted for the BS-user, BS-IRS, and IRS-user links as $C_{0, BU}$ and $\alpha_{BU}$, $C_{0, BI}$ and $\alpha_{BI}$, and $C_{0, IU}$ and $\alpha_{IU}$, respectively. 
        The corresponding path loss coefficients are denoted as $\eta_{BU}$, $\eta_{BI}$ and $\eta_{IU}$, respectively. 
        For all channels, Rician fading is assumed with the Rician factor denoted as $\beta_{BU}$, $\beta_{BI}$ and $\beta_{IU}$ for the BS-user, BS-IRS and IRS-user links, respectively. 
        Specifically, the expression of the BS-IRS channel vector $\boldsymbol{g}$ is given below as an example:
        \begin{equation}\label{def:rician-channel}
            \boldsymbol{g} = \sqrt{\frac{\beta_{BI}}{1 + \beta_{BI}}}\boldsymbol{g}^{\text{LoS}} + \sqrt{\frac{1}{1 + \beta_{BI}}}\boldsymbol{g}^{\text{NLoS}}, 
        \end{equation}
        where the vector $\boldsymbol{g}^{\text{NLoS}} = \sqrt{\eta_{BI}}\tilde{\boldsymbol{g}}$ is the Gaussian non-line-of-sight (NLoS) component with $\tilde{\boldsymbol{g}}\sim\mathcal{CN}(\boldsymbol{0}_{N_{irs}\times 1}, \boldsymbol{I}_{N_{irs}})$ and $\boldsymbol{g}^{\text{LoS}}$ is the deterministic line-of-sight (LoS) component given by 
        \begin{equation}\label{def:mmW-channel}
                \boldsymbol{g}^{\text{LoS}} = \sqrt{\eta_{BI}}\boldsymbol{b}_{N}(\omega, \psi). 
        \end{equation}
        Vector $\boldsymbol{b}_{N}(\omega, \psi)\in\mathbb{C}^{N\times 1}$ is the steering vector of the LoS path from the BS to the IRS, 
        where $\omega\in[0, \pi)$ and $\psi\in[0, \pi)$ are the physical azimuth and elevation angles of arrival (AoAs) at the IRS, respectively. 
        Specifically, define $\boldsymbol{a}_{N}(\phi) = [e^{j\pi 0\phi}, \ldots, e^{j\pi (N - 1)\phi}]^T\in\mathbb{C}^{N\times 1}$ as the $N$-dimensional steering vector. Then, $\boldsymbol{b}_{N}(\omega, \psi)$ is defined as $\boldsymbol{b}_{N}(\omega, \psi) = \boldsymbol{a}_{N_x}(\cos{(\omega)}\sin{(\psi)})\otimes\boldsymbol{a}_{N_z}(\cos{(\psi)})$, where $\otimes$ is the Kronecker product. 
        In addition, we set $C_{0, BU} = -33$ dB, $C_{0, BI} = C_{0, IU} = -30$ dB, $\alpha_{BU} = 3.7$, $\alpha_{BI} = \alpha_{IU} = 2$, $\beta_{BU} = 0$, $\beta_{BI} = 10$, $\beta_{IU} = 1$ and $p_0 = 30$ dBm, $\sigma^2 = -90$ dBm, $\epsilon = 0.95$ and $\rho = 10$, unless specified otherwise. 
        {\color{\highlightcolor}The IRS reflection vectors $\{\boldsymbol{v}_t, t = 1, \ldots, T_{p}\}$ are randomly generated subject to the discrete phase shift constraint. }
        Each point in simulation results is averaged over $1000$ random user locations and channel realizations. 


        \begin{figure}[t!]
            \centering
            {
                \begin{subfigure}[t]{0.22\textwidth}
                    \centering
                    \includegraphics[scale = 0.52]{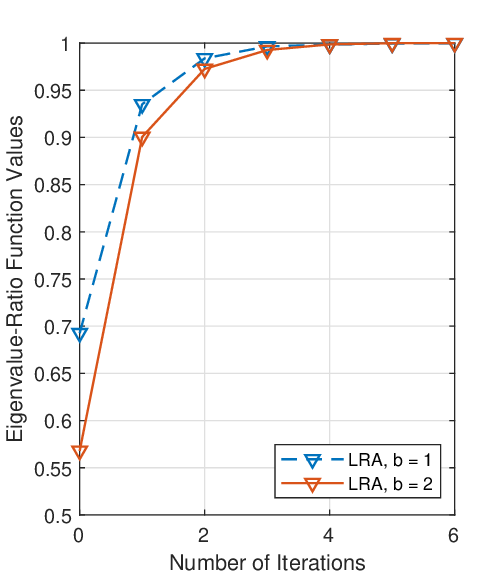}
                    \caption{LRA convergence. }
                    \label{subfig:AC-conv-RX-N64b12}
                \end{subfigure}
                \hspace{-1pt}
                \begin{subfigure}[t]{0.22\textwidth}
                    \centering
                    \includegraphics[scale = 0.52]{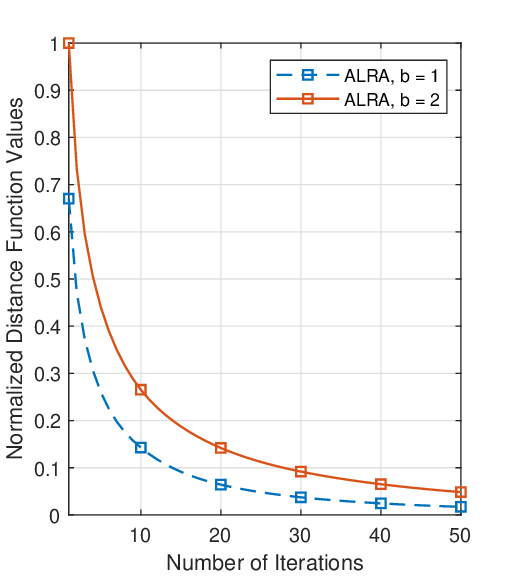}
                    \caption{ALRA convergence. }
                    \label{subfig:AC-conv-DM-N64b12}
                \end{subfigure}
            }
            \vspace{-2pt}
            \caption{Convergence of the LRA and ALRA algorithms assuming perfect power measurement ($T_{p} = 65$). }
            \label{fig:AC-conv-N64b12}
        \end{figure}

    \vspace{-6pt}
    \subsection{Performance with Perfect Power Measurement}\label{subsec:perf-accurate-measure}
        In this subsection, the performance of both LRA and ALRA algorithms is evaluated assuming perfect received signal power measurement as considered in Sections~\ref{sec:proposed-ratio-max} and~\ref{sec:proposed-dist-min}. 

        \subsubsection{Algorithm Convergence}\label{subsubsec:convergence-accurate-measure}
            As discussed in Sections~\ref{sec:proposed-ratio-max} and~\ref{sec:proposed-dist-min}, the convergence of the proposed algorithms is guaranteed with the alternating optimization process. 
            In particular, for the LRA algorithm, the eigenvalue-ratio function and the generalized eigenvalue-ratio function, i.e., $g(\boldsymbol{H}^{(i)})$ in~\eqref{def:lambda-ratio-func} for $b\ge 2$ and $g_r(\boldsymbol{H}_r^{(i)})$ in~\eqref{def:generalized-lambda-ratio-func} for $b = 1$, are non-decreasing with $i$ and upper-bounded by $1$ for the LRA algorithm. 
            On the other hand, the Euclidean distance functions for the ALRA algorithm, i.e., $d_e(\boldsymbol{H}^{(i)}, \mu^{(i)}\boldsymbol{X}^{(i)})$ for $b\ge 2$ and $d_e(\boldsymbol{H}_r^{(i)}, \mu_1^{(i)}\boldsymbol{X}_1^{(i)} + \mu_2^{(i)}\boldsymbol{X}_2^{(i)})$ for $b = 1$ as defined in~\eqref{def:euclidean-distance-func}, are positive and non-increasing with $i$. 
            In Figs.~\ref{subfig:AC-conv-RX-N64b12} and~\ref{subfig:AC-conv-DM-N64b12}, the values of eigenvalue-ratio functions and the normalized distance functions over the iterations for both $b = 1$ and $b = 2$ are shown for LRA and ALRA algorithms, respectively. 
            The total number of power measurement is set as $T_{p} = 65$. 
            It is observed from Fig.~\ref{subfig:AC-conv-RX-N64b12} that the eigenvalue-ratio function values start from small values with the initializations, and converge to $1$ within $6$ iterations, while the ALRA algorithm converges after $50$ iterations, as is shown in Fig.~\ref{subfig:AC-conv-DM-N64b12}. 
            This verifies that the estimated channel autocorrelation matrices obtained by the proposed LRA and ALRA algorithms iteratively approach a rank-one matrix for $b = 2$ and a rank-two matrix for $b = 1$, respectively. 
            
            \begin{figure}[t!]
                \centering
                \includegraphics[scale = 0.46]{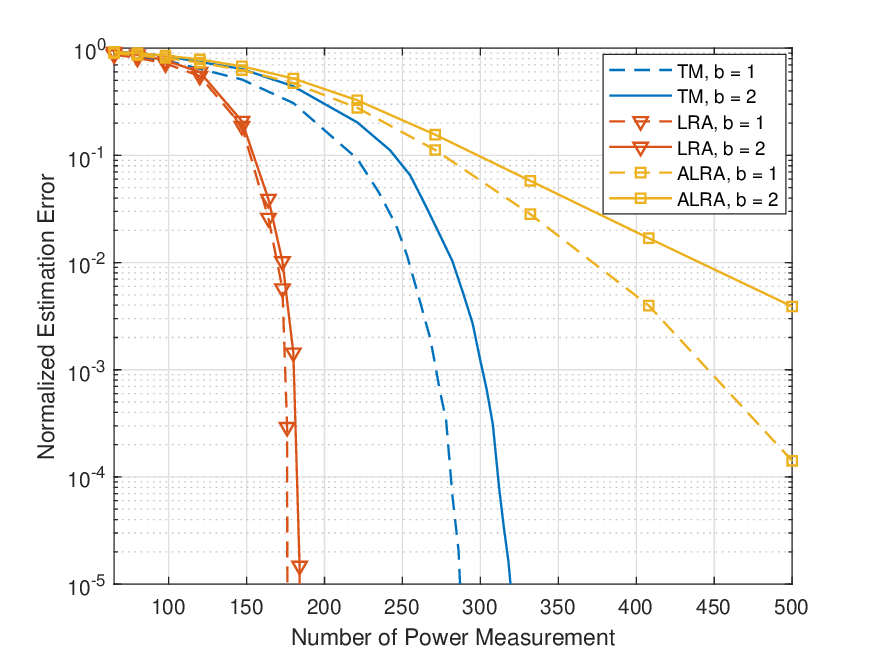}
                \vspace{-2pt}
                \caption{Normalized error by different channel estimation schemes assuming perfect power measurement. }
                \label{fig:AC-err-N64b12}
                \vspace{-13pt}
            \end{figure}

            \begin{figure}[t!]
                \centering
                {
                    \begin{subfigure}[t]{0.22\textwidth}
                        \centering
                        \includegraphics[scale = 0.52]{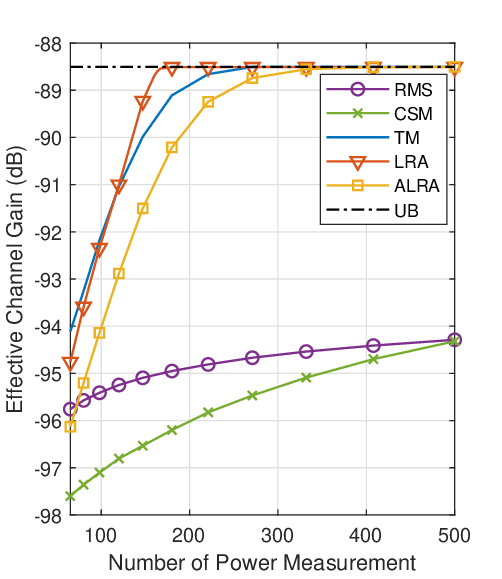}
                        \caption{$b = 1$. }
                        \label{subfig:AC-gain-N64b1}
                    \end{subfigure}
                    \hspace{-1pt}
                    \begin{subfigure}[t]{0.22\textwidth}
                        \centering
                        \includegraphics[scale = 0.52]{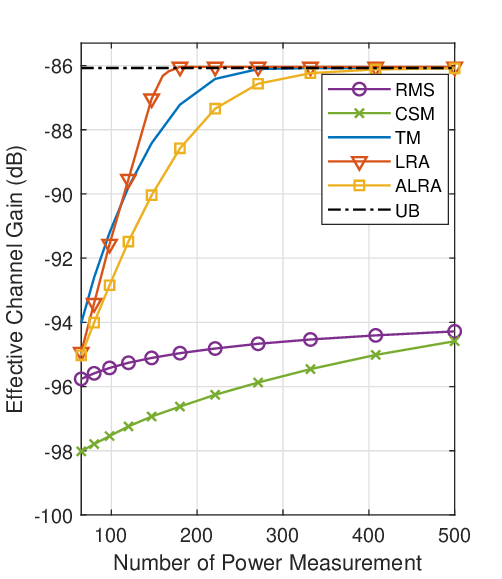}
                        \caption{$b = 2$. }
                        \label{subfig:AC-gain-N64b2}
                    \end{subfigure}
                }
                \caption{Effective channel gain with different channel estimation schemes assuming perfect power measurement. }
                \label{fig:AC-gain-N64b12}
            \end{figure}

        \subsubsection{Channel Autocorrelation Matrix Estimation Error}\label{subsubsec:est-err-accurate-measure}
            For the case of $b\ge 2$, the channel autocorrelation matrix $\bar{\boldsymbol{H}}$ is estimated as $\hat{\boldsymbol{H}}$, and the normalized estimation error is defined as $\mathcal{E}_b = \|\hat{\boldsymbol{H}} - \bar{\boldsymbol{H}}\|_F^2 / \|\bar{\boldsymbol{H}}\|_F^2$. 
            For the case of $b = 1$, the matrix $\bar{\boldsymbol{H}}_r$ is estimated as $\hat{\boldsymbol{H}}_r$, and thus the normalized estimation error is defined as $\mathcal{E}_b = \|\hat{\boldsymbol{H}}_r - \bar{\boldsymbol{H}}_r\|_F^2 / \|\bar{\boldsymbol{H}}_r\|_F^2$. 

            Fig.~\ref{fig:AC-err-N64b12} shows the normalized estimation error under different numbers of power measurement, i.e., $T_{p}$, for both LRA and ALRA algorithms, where ``TM'' represents the error of the estimated channel autocorrelation matrix by employing the trace-minimization relaxation method given in~\cite{ref:PhaseLift}. 
            It is observed that the estimation error of the LRA algorithms decreases rapidly with $T_{p}$ and is much smaller than that of the TM benchmark. 
            Moreover, the normalized estimation error vanishes for LRA and TM algorithms when $T_{p}\ge 185$ and $320$, respectively, which means that the unique solution for the channel autocorrelation matrix is recovered successfully and the proposed LRA algorithm outperforms TM significantly in terms of estimation accuracy. 
            On the other hand, the normalized estimation error of the ALRA algorithm decreases slowly with $T_{p}$, which is caused by the distance-minimization approximation, while its performance is good for sufficiently large $T_{p}$. 
            Additionally, it is shown that the estimation error for $b = 1$ is smaller than that for $b = 2$ for all algorithms. 
            This is because only the real part of $\bar{\boldsymbol{H}}$ needs to be estimated for $b = 1$. 


        \subsubsection{IRS Reflection Design with Estimated Channel}\label{subsubsec:eff-gain-accurate-measure}
            After the channel autocorrelation matrix is estimated, the IRS reflection vector $\boldsymbol{v}$ can be optimized to maximize the effective channel gain between the BS and user, denoted as $\bar{\gamma} = \text{tr}(\bar{\boldsymbol{H}}\boldsymbol{V})$, for data transmission. 
            For $b\ge 2$, we apply the eigenvalue decomposition to the estimated matrix $\hat{\boldsymbol{H}}$ and define $\hat{\lambda}_1$ as the largest eigenvalue of $\hat{\boldsymbol{H}}$ and $\hat{\boldsymbol{x}}_1$ as the corresponding normalized eigenvector. 
            Since $\hat{\boldsymbol{H}}\approx\bar{\boldsymbol{H}}$ is nearly rank-one, the effective channel gain can be approximated as $\bar{\gamma} \approx \text{tr}(\hat{\boldsymbol{H}}\boldsymbol{V}) / p_0 \approx \hat{\lambda}_1|\hat{\boldsymbol{x}}_1^H\boldsymbol{v}|^2 / p_0$. 
            Then, the IRS beamforming vector $\boldsymbol{v}$ is optimized to maximize $|\hat{\boldsymbol{x}}_1^H\boldsymbol{v}|^2$ subject to the discrete phase shift constraint $\boldsymbol{v}\in\Phi_b^N$, which can be solved optimally by using the method proposed in~\cite{ref:optimal-discrete-IRS-vector}. 
            For $b = 1$, similarly, eigenvalue decomposition is applied to the estimated matrix $\hat{\boldsymbol{H}}_r$, where $\hat{\lambda}_1$ and $\hat{\lambda}_2$ denote the first and second largest eigenvalues of $\hat{\boldsymbol{H}}_r$, with $\hat{\boldsymbol{x}}_1$ and $\hat{\boldsymbol{x}}_2$ denoting the corresponding eigenvectors, respectively. 
            Since the IRS beamforming vector $\boldsymbol{v}$ is always a real vector for $b = 1$, the effective channel gain can be approximated by 
            \begin{equation}
                \begin{aligned}
                    \bar{\gamma} \approx \frac{1}{p_0}\text{tr}(\hat{\boldsymbol{H}}_r\boldsymbol{V}) & \approx \frac{1}{p_0}\left|
                        \hat{\lambda}_1^{\frac{1}{2}}\hat{\boldsymbol{x}}_1^T\boldsymbol{v}
                    \right|^2 + \frac{1}{p_0}\left|
                        \hat{\lambda}_2^{\frac{1}{2}}\hat{\boldsymbol{x}}_2^T\boldsymbol{v}
                    \right|^2 \\
                    & = \frac{1}{p_0}\left|\left(
                        \hat{\lambda}_1^{\frac{1}{2}}\hat{\boldsymbol{x}}_1 + j\hat{\lambda}_2^{\frac{1}{2}}\hat{\boldsymbol{x}}_2
                    \right)^H\boldsymbol{v}\right|^2, 
                \end{aligned}
            \end{equation}
            and the optimal vector $\boldsymbol{v}$ can be obtained according to the method proposed in~\cite{ref:optimal-discrete-IRS-vector}. 

            \begin{figure}[t!]
                \centering
                {
                    \begin{subfigure}[t]{0.22\textwidth}
                        \centering
                        \includegraphics[scale = 0.52]{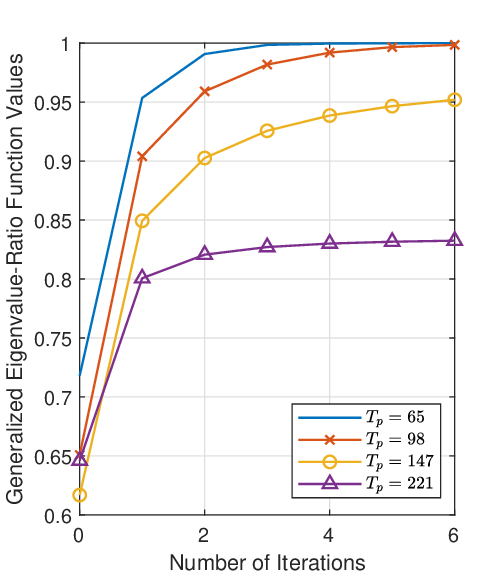}
                        \caption{R-LRA convergence. }
                        \label{subfig:NE13-conv-RX-N64b1}
                    \end{subfigure}
                    \hspace{-1pt}
                    \begin{subfigure}[t]{0.22\textwidth}
                        \centering
                        \includegraphics[scale = 0.52]{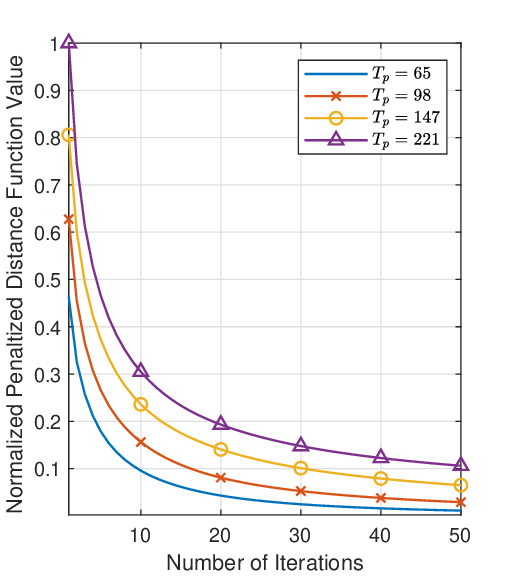}
                        \caption{R-ALRA convergence. }
                        \label{subfig:NE13-conv-DM-N64b1}
                    \end{subfigure}
                }
                \vspace{-2pt}
                \caption{Convergence of robust LRA and robust ALRA algorithms based on noisy power measurement for $b = 1$ ($\sigma^2 = -85$ dBm and $N_{0} = 1$). }
                \label{fig:NE13-conv-N64b1}
            \end{figure}

            \begin{figure*}[t!]
                \centering
                {
                    \begin{subfigure}[t]{1\textwidth}
                        \centering
                        \includegraphics[scale = 0.5]{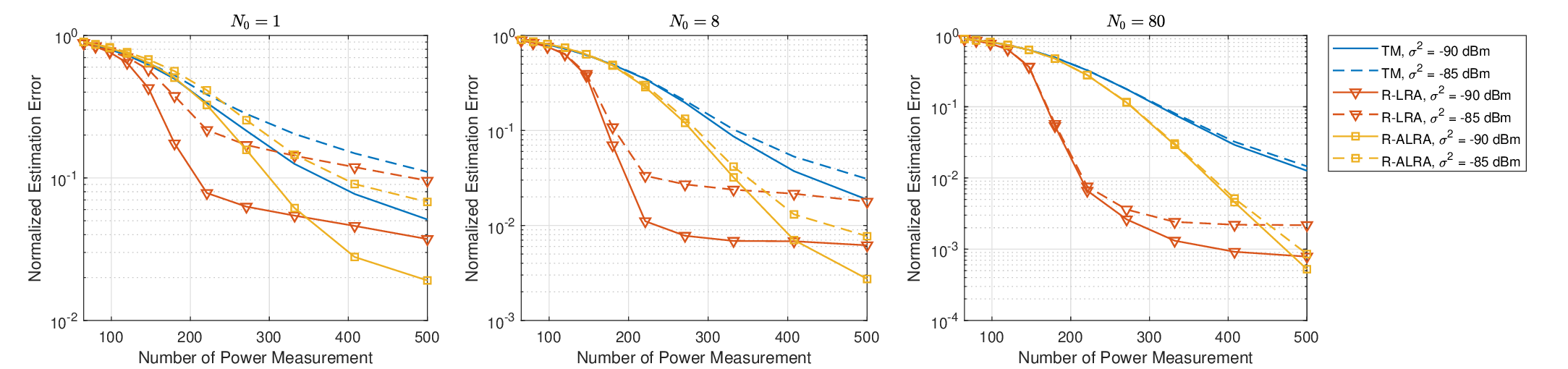}
                        \caption{Normalized estimation error versus $T_{p}$ for different values of $N_{0}$ and $\sigma^2$ ($b = 1$). }
                        \label{subfig:NE-err-N64b1}
                    \end{subfigure}
                    \begin{subfigure}[t]{1\textwidth}
                        \centering
                        \includegraphics[scale = 0.5]{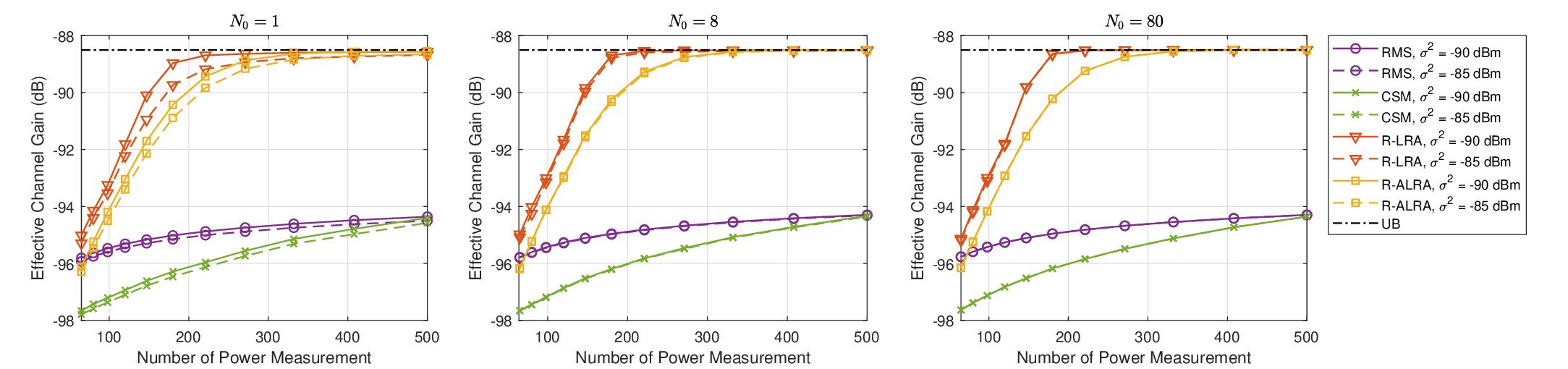}
                        \caption{Effective channel gain versus $T_{p}$ for different values of $N_{0}$ and $\sigma^2$ ($b = 1$). }
                        \label{subfig:NE-gain-N64b1}
                    \end{subfigure}
                }
                \vspace{-2pt}
                \caption{Performance of different schemes versus $T_{p}$ with various noise levels for $b = 1$. }
                \label{fig:NE-perf-N64b1}
                \vspace{-15pt}
            \end{figure*}

            For comparison, three benchmark schemes for IRS reflection design based on power measurement are listed as follows: 
            {\color{\highlightcolor}i) \textbf{TM} (trace-minimization): The IRS reflection vector is optimized as illustrated above based on the channel autocorrelation matrix estimated via the trace-minimization method in~\cite{ref:PhaseLift}. }
            ii) \textbf{UB} (upper bound): The upper bound on the effective channel gain is obtained by using the optimal IRS reflection vector based on the perfect CSI $\boldsymbol{h}$; 
            iii) \textbf{RMS} (random-max sampling): A large number of random IRS reflection vectors are applied with $u_n$ uniformly distributed in $\Phi_b$, $\forall n$, and the one achieving the largest received signal power is used; 
            iv) \textbf{CSM} (conditional sample mean): This is the method proposed in~\cite{ref:CSM}, where a large number of random IRS reflection vectors are applied, and the empirical expectation of the received signal power is calculated conditioned on $u_n$ fixed at every possible value, $\forall n$. 
            For each element, CSM selects the phase shift that maximizes the empirical expectation conditioned on its value.

            The effective channel gain obtained based on the estimated channels by the proposed algorithms as well as other benchmark schemes for $b = 1$ and $b = 2$ are shown in Fig.~\ref{subfig:AC-gain-N64b1} and~\ref{subfig:AC-gain-N64b2}, respectively. 
            Due to the discrete phase shift for IRS, the upper bound for the effective channel gain for $b = 2$ is higher than that for $b = 1$. 
            As can be observed, the effective channel gain achieved by the proposed LRA scheme increases rapidly with the number of power measurement, and reaches the upper bound when $T\ge 180$ for both $b = 1$ and $b\ge 2$ cases. 
            In comparison, the ALRA scheme achieve lower effective channel gain and approach the upper bound when $T_{p}\ge 350$. 
            {\color{\highlightcolor}Similar to the case of estimation error, the effective channel gain achieved by TM also falls in between those of LRA and ALRA, while reaching the upper bound when $T_{p}\ge 270$.}
            Moreover, both LRA and ALRA schemes outperform RMS and CSM schemes significantly for almost all values of $T_{p}$. 
            These results validate the effectiveness of the proposed channel estimation algorithms in improving the effective channel gain between the BS and user with optimized IRS reflections. 

    \vspace{-6pt}
    \subsection{Performance with Noisy Power Measurement}\label{subsec:perf-noise-effect}
        In this subsection, the robust LRA and robust ALRA algorithms proposed in Section~\ref{sec:robust-est} are evaluated based on noisy power measurement, which are labeled as R-LRA and R-ALRA, respectively. 
        The impact of noise received at the user on the performance is analyzed while the quantization effect is ignored for the time being. 
        {\color{\highlightcolor}Due to space limitation, we only consider the case of $b = 1$ in this subsection and the effective channel gain for TM is omitted for brevity.}



        \subsubsection{Algorithm Convergence}\label{subsubsec:convergence-noise-effect}
            To verify the convergence of the R-LRA and R-ALRA algorithms based on noisy power measurement, Fig.~\ref{subfig:NE13-conv-RX-N64b1} shows the values of the generalized eigenvalue-ratio function $g_r(\boldsymbol{H}_r^{(i)})$ defined in~\eqref{def:generalized-lambda-ratio-func}, while Fig.~\ref{subfig:NE13-conv-DM-N64b1} shows the normalized values of the penaltized distance function $\varphi_r(\boldsymbol{H}_r^{(i)}, \boldsymbol{\mu}^{(i)}, \boldsymbol{X}_1^{(i)}, \boldsymbol{X}_2^{(i)})$ defined in Section~\ref{subsec:dist-min-est-robust} during the iterations. 
            We set $\sigma^2 = -85$ dBm and $N_{0} = 1$. 
            It is observed that the convergence of the proposed robust algorithms is still guaranteed. 
            However, the generalized eigenvalue-ratio function of R-LRA may not converge to $1$ due to the impact of random noise. 
            Especially for large $T_p$, it becomes increasingly difficult to find a low-rank matrix $\boldsymbol{H}_r$ that can keep the penalty term small enough. 
            The penaltized distance function for R-ALRA also increases with $T_p$. 
            This indicates that both the proposed R-LRA and R-ALRA algorithms approach low-rank solutions and reduce the penalty terms at the same time, yielding robust estimations for channel autocorrelation matrices.

        \subsubsection{Estimation Error and Effective Channel Gain}\label{subsubsec:error-gain-noise-effect}
            
            The normalized estimation error and effective channel gain achieved by the proposed and benchmark schemes for $b = 1$ are shown versus the number of power measurement in Figs.~\ref{subfig:NE-err-N64b1} and~\ref{subfig:NE-gain-N64b1}, respectively. 
            For three subfigures in Fig.~\ref{subfig:NE-err-N64b1}, the number of reference signals for each power measurement, $N_{0}$, is set to be $1$, $8$ and $80$, respectively, and the estimation errors with $\sigma^2 = -90$ dBm and $\sigma^2 = -85$ dBm are shown in each subfigure for comparison. 
            Their corresponding effective channel gains are shown in the three subfigures in Fig.~\ref{subfig:NE-gain-N64b1}. 

            In Fig.~\ref{subfig:NE-err-N64b1}, it is observed that the normalized estimation error decreases with $T_{p}$ and $N_{0}$ but increases with $\sigma^2$ for all schemes, and both the R-LRA and R-ALRA algorithms always achieve a lower error than TM. 
            Notably, the estimation error of the R-ALRA algorithm is even lower than R-LRA when $T_{p}$ is large. 
            This is because, with larger noise, the effectiveness of the trace-minimization relaxation is undermined for TM because of the greater uncertainty of the measured power values. 
            In comparison, the R-LRA algorithm may be more likely to converge to a locally optimal solution since the estimation problem is non-convex, especially when $T_{p}$ is large. 
            For R-ALRA, however, the optimizations of $\boldsymbol{H}$ and $\boldsymbol{H}_r$ for $b\ge 2$ and $b = 1$, are both convex and always have only one globally optimal solution, leading to better estimation robustness with higher noise power and larger $T_{p}$. 
            In contrast, when $T_p$ is small, e.g., $T_{p} \le 300$, the R-LRA algorithm achieves much lower estimation error than R-ALRA. 

            Despite that the estimation error varies greatly for different values of $\sigma^2$ and $N_{0}$, the effective channel gain achieved by IRS reflection design based on estimated channel autocorrelation matrices does not change significantly, as shown in Fig.~\ref{subfig:NE-gain-N64b1}. 
            For $N_{0} = 1$, the effective channel gains achieved by all schemes for $\sigma^2 = -90$ dBm and $-85$ dBm are close to each other. 
            Besides, the performance difference between $N_{0} = 8$ and $N_{0} = 80$ is small. 
            For the proposed schemes based on the R-LRA and R-ALRA algorithms, the optimized IRS reflection vector $\boldsymbol{v}$ based on the estimated channel autocorrelation matrix $\hat{\boldsymbol{H}}$ or $\hat{\boldsymbol{H}}_r$ is still near-optimal even when the estimation error is large. 
            By comparing Fig.~\ref{subfig:NE-err-N64b1} to Fig.~\ref{subfig:NE-gain-N64b1}, it can be observed that as long as the normalized estimation error is no larger than $10^{-1}$, the optimized IRS reflection vector based on the estimated channel autocorrelation matrix is nearly optimal, and thus the effective channel gain almost achieves the upper bound. 
            As such, significant improvements can still be obtained by the R-LRA and R-ALRA schemes compared to other benchmark schemes. 

    \vspace{-8pt}
    \subsection{Performance with Quantized Power Measurement}\label{subsec:perf-quantize-effect}
        Finally, the impact of quantization errors on the performance of R-LRA and R-ALRA algorithms is evaluated. 
        We fix $T_{p} = 200$, $\sigma^2 = -90$ dBm and $N_{0} = 8$, while the estimation error and effective channel gain are evaluated with different values of the quantization gap $D$. 
        

        \begin{figure}[t!]
            \centering
            {
                \begin{subfigure}[t]{0.22\textwidth}
                    \centering
                    \includegraphics[scale = 0.52]{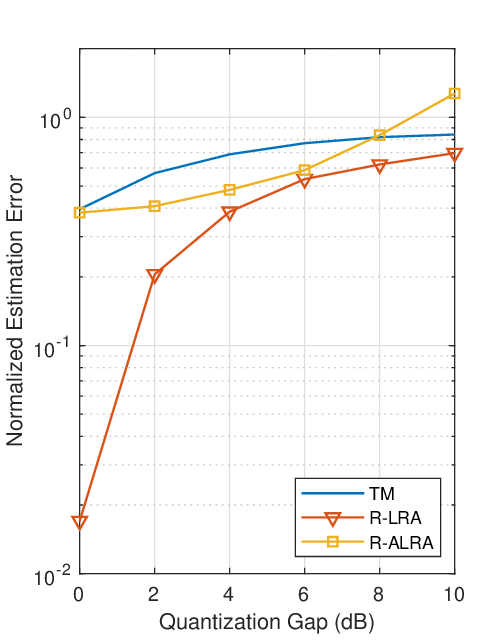}
                    \caption{Estimation error. }
                    \label{subfig:QGE-error-NR08-N64b1}
                \end{subfigure}
                \hspace{-1pt}
                \begin{subfigure}[t]{0.22\textwidth}
                    \centering
                    \includegraphics[scale = 0.52]{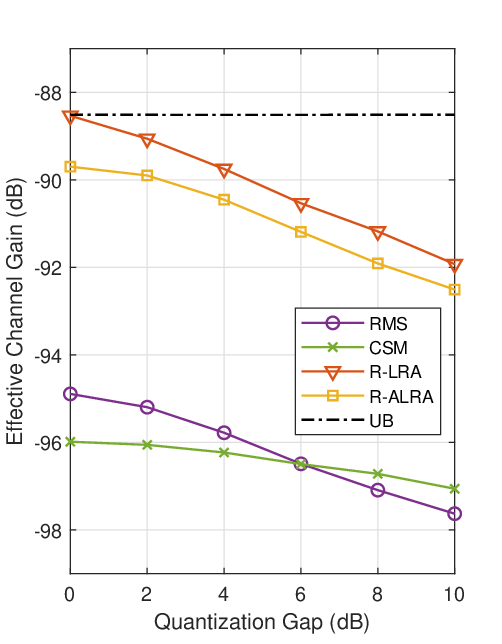}
                    \caption{Effective channel gain. }
                    \label{subfig:QGE-gain-NR08-N64b1}
                \end{subfigure}
            }
            \caption{Performance with quantized power measurement for $b = 1$ ($\sigma^2 = -90$ dBm, $N_{0} = 8$, and $T_{p} = 200$). }
            \label{fig:QGE-NR08-N64b1}
        \end{figure}

        In Fig.~\ref{subfig:QGE-error-NR08-N64b1}, the normalized estimation error of different schems versus $D$ is shown for $b = 1$. 
        Correspondingly, the effective channel gain under the same condition is given in Fig.~\ref{subfig:QGE-gain-NR08-N64b1}. 
        It can be observed that the estimation error increases and the effective channel gain decreases with the quantization gap $D$. 
        The estimation error of R-LRA is always lower than that of TM or R-ALRA, while the error of R-ALRA becomes higher than TM for $T_{p}\ge 8$ dB. 
        This is caused by introducing the variable $\hat{\boldsymbol{q}}$ and approximating problem~\eqref{prob:cov-est-dist-min-robust-multiary} with problem~\eqref{prob:cov-est-dist-min-approx-robust-multiary} to obtain closed-form solutions. 
        With larger $D$, the accuracy of the approximation decreases and thus the estimation error increases. 
        Nevertheless, the effective channel gains of both R-LRA and R-ALRA are close to each other and they are higher compared to all benchmark schemes for all values of $D$, which verifies the robustness of the proposed schemes against quantization error. 
        It is worth mentioning that in the current protocol~\cite{ref:3gpp:36.133}, $D = 1$ dB is employed, for which the performance loss caused by quantization for both R-LRA and R-ALRA is quite small based on Fig.~\ref{subfig:QGE-gain-NR08-N64b1}. 

\vspace{-6pt}
\section{Conclusion}\label{Conclusion}
    This paper studied the IRS-cascaded channel estimation problem based on received signal power measurement at the user. 
    Since the signal phase information is not available in power measurement, the channel autocorrelation matrix was estimated by solving a rank-minimization problem. 
    Based on the analysis of existence and uniqueness of the solution, two algorithms of different computational complexity were proposed to solve the channel autocorrelation estimation problem. 
    By relaxing the non-convex low-rank constrained problem to a fractional programming problem, the proposed LRA algorithm applied the alternating optimization method to iteratively approach a low-rank matrix solution. 
    To reduce computational complexity, the ALRA algroithm was also developed to obtain an approximate solution in closed-form during each iteration. 
    Moreover, robust extensions of the LRA and ALRA algorithms were proposed based on practical power measurement, where the effects of receiver noise and quantization error in power measurement were considered. 
    The convergence and efficiency of the proposed algorithms were validated via simulations, demonstrating that small estimation errors of the channel autocorrelation matrix can be achieved. 
    It was also shown that by applying the IRS reflection designs based on the channel autocorrelation matrices estimated by the proposed LRA and ALRA algorithms, significant improvement in the effective channel gain between BS and user can be achieved compared to other benchmark schemes. 
    Furthermore, the performance of the robust LRA and robust ALRA algorithms was evaluated using practical power measurement with noise and quantization errors, which demonstrated their robustness in improving the effective channel gain for IRS-assisted communication systems, even with imperfect power measurement. 


\appendices

    \vspace{-6pt}
    \section{Proof of Lemma~\ref{lemma:dim-defficiency}}\label{appendix:lemma-proof-dim-deff}
        Note that the set of $(N\times N)$-dimensional hermitian matrices forms an $N^2$-dimensional linear space over $\mathbb{R}$~\cite{ref:OLS}. 
        Define $\mathcal{S}_V^{b}(N) = \{\boldsymbol{V}\in\mathbb{C}^{N\times N} | \boldsymbol{V} = \boldsymbol{v}\boldsymbol{v}^H, \boldsymbol{v}\in\Phi_b^N\}$ as the set of all possible autocorrelation matrices of the IRS reflection vector and $\mathcal{S}_X^{b}(N)$ as the linear space spanned by elements from $\mathcal{S}_V^{b}(N)$ with real coefficients. 
        Due to the unit-amplitude of the entries of $\boldsymbol{v}$, the diagonal entries in matrix $\boldsymbol{V}$ are always equal to one, which results in a dimensional defficiency for $\mathcal{S}_X^{b}(N)$ in the $N^2$-dimensional hermitian matrix space. 
        Specifically, it can be proved that $\text{dim}(\mathcal{S}_X^{b}(N)) = \mathcal{D}_N^{(b)} = N^2 - N + 1$ for $b\ge 2$ and $(N^2 - N) / 2 + 1$ for $b = 1$ (see Appendix~\ref{complm-sec:lemma1-proof}). 
        Since we have $D_V \le \text{dim}(\mathcal{S}_X^{b}(N)) = \mathcal{D}_N^{(b)}$, Lemma~\ref{lemma:dim-defficiency} is proved. 

    \vspace{-6pt}
    \section{Proof of Proposition~\ref{prop:cov-est-existence-uniqueness}}\label{appendix:prop-proof-uniqueness}
        The existence of $\bar{\boldsymbol{H}}$ is straightforward because problem~\eqref{prob:cov-est-find-origin} follows~\eqref{def:received-power}, where the actual channel autocorrelation matrix is always a feasible solution. 

        Next, the uniqueness is analyzed. 
        For $b\ge 2$, $\{\boldsymbol{v}_t\in\Phi_b^{N}$, $t = 1, \ldots, T_{p}\}$ are complex vectors. 
        When $T_{p}\ge D_V = N^2 - N + 1$, matrices $\{\boldsymbol{V}_t, t = 1, \ldots, T_{p}\}$ span the space $\mathcal{S}_X^{b}(N)$ and thus the projection of $\boldsymbol{H}$ in $\mathcal{S}_X^{b}(N)$ is determined. 
        By analyzing the orthonormal basis of $\mathcal{S}_X^{b}(N)$, it can be shown that the values of $\text{tr}(\boldsymbol{H})$ and all non-diagonal entries in $\boldsymbol{H}$ can be uniquely determined by~\eqref{prob:cov-est-find-power}, i.e., $\text{tr}(\boldsymbol{H}) = \text{tr}(\bar{\boldsymbol{H}})$ and $H_{nm} = \bar{H}_{nm}$, $\forall n\neq m$. 
        Under such conditions, it can be proved that $\boldsymbol{H} = \bar{\boldsymbol{H}}$ is the only rank-one solution for problem~\eqref{prob:cov-est-find-origin} when $N\ge 3$ (see Appendix~\ref{complm-sec:prop1-proof}). 

        For $b = 1$, $\boldsymbol{v}_t\in\Phi_1^{N} = \{\pm 1\}^N$ is always a real vector for any $t$, indicating that $\boldsymbol{V}_t = \boldsymbol{V}_t^T$, $\forall t$. 
        Then, we have 
        \begin{equation}
            \text{tr}(\boldsymbol{H}\boldsymbol{V}_t) = \text{tr}((\boldsymbol{H}\boldsymbol{V}_t)^T) = \text{tr}(\boldsymbol{V}_t^T\boldsymbol{H}^*) = \text{tr}(\boldsymbol{H}^*\boldsymbol{V}_t), 
        \end{equation}
        which means that the imaginary part of $\boldsymbol{H}$ does not influence the received signal power at the user for $b = 1$. 
        Thus, for any solution $\hat{\boldsymbol{H}}$ for problem~\eqref{prob:cov-est-find-origin}, its conjugate $\hat{\boldsymbol{H}}^*$ is also a solution. 
        Similar to the case of $b\ge 2$, by analyzing the orthonormal basis of $\mathcal{S}_X^{b}(N)$ for $b = 1$, it can be shown that the values of $\text{tr}(\boldsymbol{H})$ and real parts of all non-diagonal entries in $\boldsymbol{H}$ can be uniquely determined by~\eqref{prob:cov-est-find-power} when $T_{p}\ge D_V = \mathcal{D}_{N}^{(1)}$, i.e., $\text{tr}(\boldsymbol{H}) = \text{tr}(\bar{\boldsymbol{H}})$ and $\text{real}(H_{nm}) = \text{real}(\bar{H}_{nm})$, $\forall n\neq m$. 
        Based on that, it can be proved that $\bar{\boldsymbol{H}}$ and $\bar{\boldsymbol{H}}^*$ are the only two rank-one solutions for problem~\eqref{prob:cov-est-find-origin} when $N\ge 6$ (see Appendix~\ref{complm-sec:prop1-proof}).

    \vspace{-6pt}
    \section{Proof of Proposition~\ref{prop:solution-set-cov-real}}\label{appendix:prop-proof-bit1-uniqueness}
        It is easy to verify that $\bar{\boldsymbol{H}}_r$ is a solution for problem~\eqref{prob:cov-est-find-rank-two}, which guarantees the existence of the solution. 
        For the uniqueness, any solution for problem~\eqref{prob:cov-est-find-rank-two}, denoted by $\hat{\boldsymbol{H}}_r$, is symmetric and semidefinite with $\text{rank}(\hat{\boldsymbol{H}}_r) \le 2$. 
        With eigenvalue decomposition, we have $\hat{\boldsymbol{H}}_r = \boldsymbol{U}_r\text{diag}([\alpha_1, \alpha_2])\boldsymbol{U}_r^T$, where $\boldsymbol{U}_r = [\boldsymbol{q}_{r1}, \boldsymbol{q}_{r2}]\in\mathbb{R}^{N\times 2}$ and $\alpha_1, \alpha_2 \ge 0$. 
        Denote $\hat{\boldsymbol{h}}_1 = \sqrt{\alpha_1}\boldsymbol{q}_{r1}, \hat{\boldsymbol{h}}_2 = \sqrt{\alpha_2}\boldsymbol{q}_{r2}\in\mathbb{R}^{N\times 1}$, then $\hat{\boldsymbol{H}}_r = \hat{\boldsymbol{h}}_1\hat{\boldsymbol{h}}_1^T + \hat{\boldsymbol{h}}_2\hat{\boldsymbol{h}}_2^T$. 
        Consider $\hat{\boldsymbol{h}} = \hat{\boldsymbol{h}}_1 + j\hat{\boldsymbol{h}}_2$ and $\hat{\boldsymbol{H}} = \hat{\boldsymbol{h}}\hat{\boldsymbol{h}}^H$. 
        Obviously, we have $\hat{\boldsymbol{H}}_r = \text{Re}(\hat{\boldsymbol{H}})$, $\text{rank}(\hat{\boldsymbol{H}}) = 1$ and $p_0\text{tr}(\hat{\boldsymbol{H}}\boldsymbol{V}_t) = p_0\text{tr}(\hat{\boldsymbol{H}}_r\boldsymbol{V}_t) = p_t$, $t = 1, \ldots, T_{p}$. 
        Thus, $\hat{\boldsymbol{H}}$ is a solution for problem~\eqref{prob:cov-est-find-origin}, which means $\hat{\boldsymbol{H}} = \bar{\boldsymbol{H}}$ or $\bar{\boldsymbol{H}}^{*}$ according to Proposition~\ref{prop:cov-est-existence-uniqueness}. 
        Then, we have $\hat{\boldsymbol{H}}_r = \text{Re}(\hat{\boldsymbol{H}}) = \bar{\boldsymbol{H}}_r$, which is uniquely determined. 


\vspace{-6pt}
\section{Complementary Proof of Lemma~\ref{lemma:dim-defficiency} in Appendix~\ref{appendix:lemma-proof-dim-deff}}\label{complm-sec:lemma1-proof}

    In this section, the detailed proof for the dimension of space $\mathcal{S}_X^{b}(N)$ defined in Appendix~\ref{appendix:lemma-proof-dim-deff} is given to complete the proof of Lemma~\ref{lemma:dim-defficiency}, where it is claimed that $\text{dim}(\mathcal{S}_X^{b}(N)) = \mathcal{D}_{N}^{(b)}$. 
    The function $\mathcal{D}_{N}^{(b)}$ is defined in Lemma~\ref{lemma:dim-defficiency} as
    \begin{equation}\label{def:spanned-spaced-dim}
        \mathcal{D}_{N}^{(b)} = \left\{
            \begin{aligned}
                & \frac{N^2 - N}{2} + 1, ~~~ b = 1 \\
                & N^2 - N + 1, ~~~~ b \ge 2
            \end{aligned}
        \right.
        .
    \end{equation}
    To prove $\text{dim}(\mathcal{S}_X^{b}(N)) = \mathcal{D}_{N}^{(b)}$, we first denote the $N^2$-dimensional hermitian matrices space as $\mathcal{S}_H(N)$ and explicitly define an orthonormal basis for it. 
    Then, we show that $\mathcal{S}_X^{b}(N)$ can be spanned by $\mathcal{D}_{N}^{(b)}$ matrices in the basis. 

    \vspace{-5pt}
    \subsection{Orthonormal Basis of Space $\mathcal{S}_H(N)$}\label{complm-subsec:herm-ortho-basis}
        Consider vector $\boldsymbol{b}_1 = \boldsymbol{1}_{N} / \sqrt{N} \in\mathbb{R}^{N\times 1}$, and denote $\boldsymbol{b}_2, \ldots, \boldsymbol{b}_N$ as the vectors that form an orthonormal basis along with $\boldsymbol{b}_1$ for the $N$-dimensional real vector space $\mathbb{R}^{N\times 1}$. 
        Then, an orthonormal basis for space $\mathcal{S}_H(N)$ defined in~\cite{ref:OLS} can be given by the following lemma. 
        \begin{lemma}
            Matrices $\{\boldsymbol{B}_{1}, \ldots, \boldsymbol{B}_{N}; \boldsymbol{E}_{nl}^{(R)}, \boldsymbol{E}_{nl}^{(I)}, 1\le n < l\le N\}$ form an orthonormal basis for $\mathcal{S}_H(N)$ defined as  
            \begin{subequations}\label{complm-def:herm-space-basis}
                \begin{align}
                    \boldsymbol{B}_{n} & = \textup{diag}(\boldsymbol{b}_n), ~1\le n\le N, \label{complm-def:herm-space-diag-basis} \\
                    \boldsymbol{E}_{nl}^{(R)} & = \frac{1}{\sqrt{2}}\left(
                        \boldsymbol{e}_n\boldsymbol{e}_l^T + \boldsymbol{e}_l\boldsymbol{e}_n^T
                    \right), ~1\le n < l \le N, \label{complm-def:herm-space-real-basis} \\
                    \boldsymbol{E}_{nl}^{(I)} & = \frac{j}{\sqrt{2}}\left(
                        \boldsymbol{e}_n\boldsymbol{e}_l^T - \boldsymbol{e}_l\boldsymbol{e}_n^T
                    \right), ~1\le n < l \le N, \label{complm-def:herm-space-imag-basis}
                \end{align}
            \end{subequations}
            where $\boldsymbol{e}_n\in\mathbb{R}^{N\times 1}$ has a $1$ for the $n$-th element and all other elements equal to zero. 
        \end{lemma}
        \begin{proof}[Proof\textup{:}\nopunct]
            It is easy to verify that these matrices are hermitian and orthonormal to each other under the Frobenius product $\text{tr}(\boldsymbol{A}^H\boldsymbol{B})$ for arbitrary matrices $\boldsymbol{A}$ and $\boldsymbol{B}$. 
            Meanwhile, on one hand, for any $N\times N$ hermitian matrix $\boldsymbol{A}$, its main diagonal entries form an $N$-dimensional real vector $\boldsymbol{d}_{\boldsymbol{A}} = \text{diag}(\boldsymbol{A})$. 
            Thus, there exists a real vector $\boldsymbol{c}_{\boldsymbol{A}} = [c_{\boldsymbol{A}, 1}, \ldots, c_{\boldsymbol{A}, N}]^T\in\mathbb{R}^{N\times 1}$ such that $\boldsymbol{d}_{\boldsymbol{A}} = \sum_{n = 1}^{N}{c_{\boldsymbol{A}, n}\boldsymbol{b}_n}$, i.e., 
            \begin{equation}\label{complm-eq:diag-vec-basis}
                \text{diag}(\boldsymbol{d}_{\boldsymbol{A}}) = \sum_{n = 1}^{N}{c_{\boldsymbol{A}, n}\boldsymbol{B}_n}. 
            \end{equation}
            On the other hand, entries of matrices $\boldsymbol{E}_{nl}^{(R)}$ and $\boldsymbol{E}_{nl}^{(I)}$ are zero except for two on the $n$-th row and $l$-th column, and the $l$-th row and $n$-th column, $\forall n < l$. 
            Thus, it can be verified that $\boldsymbol{A}$ can be written as 
            \begin{equation}\label{complm-eq:herm-mat-basis-expansion}
                    \boldsymbol{A} = \sqrt{2}\sum_{n < l}{\text{Re}(A_{nl})\boldsymbol{E}_{nl}^{(R)} + \text{Im}(A_{nl})\boldsymbol{E}_{nl}^{(I)}} + \sum_{n = 1}^{N}{c_{\boldsymbol{A}, n}\boldsymbol{B}_{n}}, 
            \end{equation}
            which is the linear combination of $N(N - 1) + N = N^2$ matrices with $N^2$ real coefficients. 
            Therefore, the $N^2$ matrices $\{\boldsymbol{B}_1, \ldots, \boldsymbol{B}_{N}; \boldsymbol{E}_{nl}^{(R)}, \boldsymbol{E}_{nl}^{(I)}$, $1\le n < l \le N\}$ form an orthonormal basis for the space of hermitian matrices. 
        \end{proof}

    \subsection{Dimension of Space $\mathcal{S}_X^{b}(N)$}\label{complm-subsec:spanned-space-dim}
        To show $\text{dim}(\mathcal{S}_X^{b}(N)) = \mathcal{D}_{N}^{(b)}$, we provide a two-fold proof, where we first prove $\text{dim}(\mathcal{S}_X^{b}(N)) \le \mathcal{D}_{N}^{(b)}$ and then $\text{dim}(\mathcal{S}_X^{b}(N)) \ge \mathcal{D}_{N}^{(b)}$. 

        For the first part, we show $\text{dim}(\mathcal{S}_X^{b}(N)) \le \mathcal{D}_{N}^{(b)}$. 
        Since $\mathcal{S}_X^{b}(N)$ is spanned by $\mathcal{S}_V^{b}(N)$ with real coefficients, we have $\mathcal{S}_X^{b}(N)\in\mathcal{S}_H(N)$ and thus $\text{dim}(\mathcal{S}_X^{b}(N))\le N^2$. 
        Due to the unit-amplitude of the IRS reflecting coefficients, however, the diagonal entries in any matrix $\boldsymbol{V}\in\mathcal{S}_V^{b}(N)$ are always equal to 1, i.e., $\text{diag}(\boldsymbol{V}) = \boldsymbol{1}_N = \sqrt{N}\boldsymbol{b}_1$. 
        Then, it is easy to verify that $\boldsymbol{V}$ is orthogonal to $\boldsymbol{B}_{2}, \ldots, \boldsymbol{B}_{N}$ for any values of $b$, i.e., 
        \begin{equation}\label{complm-eq:cov-null-space}
            \text{tr}(\boldsymbol{V}^H\boldsymbol{B}_{n}) = \sqrt{N}\boldsymbol{b}_{1}^T\boldsymbol{b}_{n} = 0, ~2\le n\le N. 
        \end{equation}
        Moreover, for $b = 1$, $\boldsymbol{V}$ is always a real matrix and thus we have $\text{tr}(\boldsymbol{V}^H\boldsymbol{E}_{nl}^{(I)}) = \text{Im}(V_{nl}) = 0$, $\forall n\neq l$. 
        Thus, for $b\ge 2$, $\mathcal{S}_X^{b}(N)$ should be a subspace of the space spanned by $\{\boldsymbol{B}_1; \boldsymbol{E}_{nl}^{(R)}, \boldsymbol{E}_{nl}^{(I)}, 1\le n < l \le N\}$, leading to $\text{dim}(\mathcal{S}_X^{b}(N))\le N^2 - N + 1 = \mathcal{D}_{N}^{(b)}$, while for $b = 1$, $\mathcal{S}_X^{1}(N)$ should be a subspace of the space spanned by $\{\boldsymbol{B}_1; \boldsymbol{E}_{nl}^{(R)}, 1\le n < l \le N\}$, leading to $\text{dim}(\mathcal{S}_X^{1}(N))\le N(N - 1)/2 + 1 = \mathcal{D}_{N}^{(1)}$. 
        Therefore, $\text{dim}(\mathcal{S}_X^{b}(N))\le\mathcal{D}_{N}^{(b)}$ holds for any $b$.

        For the second part of the proof, we prove by induction. 
        To prove $\text{dim}(\mathcal{S}_X^{b}(N)) \ge \mathcal{D}_{N}^{(b)}$, it suffices to show that there exist at least $\mathcal{D}_{N}^{(b)}$ matrices in $\mathcal{S}_V^{b}(N)$ such that they are linearly independent with real coefficients. 
        For the following, we first prove the case of $b = 1$, and then consider the case of $b\ge 2$. 

        For $b = 1$, we show that there always exists $\mathcal{D}_{N}^{(1)} = (N^2 - N) / 2 + 1$ matrices in $\mathcal{S}_V^{1}(N)$ that are linearly independent with real coefficients. 
        It is easy to verify that this is true when $N \le 3$. 
        Assume that this is true for $N = L\ge 3$, and let $K_1 = \mathcal{D}_{L}^{(1)}$ for simplicity, then we have to show that this is true for $N = L + 1$. 
        Suppose $\boldsymbol{V}_k = \boldsymbol{v}_k\boldsymbol{v}_k^H$, $k = 1, \ldots, {K_1}$ are linearly independent matrices in $\mathcal{S}_V^{1}(L)$, where $\boldsymbol{v}_k\in{\{\pm 1\}}^L$, $k = 1, \ldots, K_1$. 
        Define matrix $\boldsymbol{W} = [\boldsymbol{v}_1, \ldots, \boldsymbol{v}_{K_1}]\in\mathbb{R}^{L\times {K_1}}$, which has full row rank $L$ according to the following lemma. 
        \begin{lemma}\label{complm-lemma:covariance-basis-full-rank}
            The matrix $\boldsymbol{W}$ has full row rank $L$ for $L\ge 3$. 
        \end{lemma}
        \begin{proof}[Proof\textup{:}\nopunct]
            As $K_1 > L$ for $L\ge 3$, we have $\text{rank}(\boldsymbol{W})\le L < {K_1}$. 
            If $\text{rank}(\boldsymbol{W}) < L$, then the space spanned by the columns of $\boldsymbol{W}$, denoted as $\text{span}(\boldsymbol{W})$, has a dimension lower than $L$. 
            Since all vectors in $\{\pm 1\}^L$ span the $L$-dimensional space $\mathbb{R}^L$, there exists a vector $\boldsymbol{x}\in\{\pm 1\}^{L}$ such that $\boldsymbol{x}\notin\text{span}(\boldsymbol{W})$. 
            Let $\boldsymbol{x} = \boldsymbol{v}_{x} + \boldsymbol{y}$, where $\boldsymbol{v}_{x}\in\text{span}(\boldsymbol{W})$ while $\boldsymbol{y}\perp\text{span}(\boldsymbol{W})$. 
            Note that $\boldsymbol{X} = \boldsymbol{x}\boldsymbol{x}^H\in \mathcal{S}_V^{1}(N)$ and thus matrices $\{\boldsymbol{V}_1, \ldots, \boldsymbol{V}_{K_1}, \boldsymbol{X}\}$ must be linearly dependent with real coefficients according to the first part of the proof, because there are ${K_1} + 1 > \mathcal{D}_{L}^{(1)}$ hermitian matrices from $\mathcal{S}_V^{1}(N)$ in total. 
            Hence, there exists $\boldsymbol{\mu}\in\mathbb{R}^{{K_1}\times 1}$ and $\mu_x\in\mathbb{R}$, such that they are not all zeros and 
            \begin{equation}\label{complm-eq:mux-linear-dependent}
                \mu_x\boldsymbol{X} + \sum_{k = 1}^{{K_1}}{\mu_k\boldsymbol{V}_k} = \boldsymbol{0}_{L\times L}. 
            \end{equation}
            Consider the matrix $\boldsymbol{Y} = \boldsymbol{y}\boldsymbol{y}^H$ and we have 
            \begin{subequations}\label{complm-eq:trace-applied}
                \allowdisplaybreaks
                \begin{align}
                    0 & = \text{tr}\left(
                        \boldsymbol{Y}^H\left(\mu_x\boldsymbol{X} + \sum_{k = 1}^{{K_1}}{\mu_k\boldsymbol{V}_k}\right)
                    \right) \\
                    & = \mu_x\text{tr}(\boldsymbol{Y}^H\boldsymbol{X}) + \sum_{k = 1}^{{K_1}}{\mu_k\text{tr}(\boldsymbol{Y}^H\boldsymbol{V}_k)}. 
                \end{align}
            \end{subequations}
            Note that $\text{tr}(\boldsymbol{Y}^H\boldsymbol{X}) = |\boldsymbol{y}^H(\boldsymbol{v}_x + \boldsymbol{y})|^2 = \|\boldsymbol{y}\|_2^4$ and $\text{tr}(\boldsymbol{Y}^H\boldsymbol{V}_k) = |\boldsymbol{y}^H\boldsymbol{v}_k|^2 = 0$, $\forall k$. 
            Thus, equation~\eqref{complm-eq:trace-applied} is equivalent to $\mu_x\|\boldsymbol{y}\|_2^4 = 0$, which leads to $\mu_x = 0$ as $\boldsymbol{y}$ is nonzero. 
            Then, equation~\eqref{complm-eq:mux-linear-dependent} becomes $\sum_{k = 1}^{{K_1}}{\mu_k\boldsymbol{V}_k} = \boldsymbol{0}$ for some $\boldsymbol{\mu}\neq\boldsymbol{0}$. 
            This is contradictory to the assumption that $\{\boldsymbol{V}_k, k = 1, \ldots, {K_1}\}$ are linearly independent with real coefficients. 
            Therefore, $\boldsymbol{W}$ should have full row rank $L$. 
        \end{proof}
        Given Lemma~\ref{complm-lemma:covariance-basis-full-rank}, we implicitly construct $\mathcal{D}_{L + 1}^{(1)}$ linearly independent matrices in $S_V^1(L + 1)$ for the case of $N = L + 1$. 
        Define $(L + 1)$-dimensional vectors 
        \begin{equation}\label{complm-def:induction-construction-b1}
            \boldsymbol{u}_{k} = \left[
                \begin{array}{c}
                    \boldsymbol{v}_k \\
                    1
                \end{array}
            \right], 
            ~\boldsymbol{y}_{k} = \left[
                \begin{array}{c}
                    \boldsymbol{v}_k \\
                    -1
                \end{array}
            \right], 
            ~k = 1, \ldots, {K_1}, 
        \end{equation}
        and matrices $\boldsymbol{U}_k = \boldsymbol{u}_k\boldsymbol{u}_k^H$ and $\boldsymbol{Y}_k = \boldsymbol{y}_k\boldsymbol{y}_k^H$, $\forall k$. 
        Obviously, $\boldsymbol{U}_k, \boldsymbol{Y}_k\in S_V^1(L + 1)$, $\forall k$, and there are $2{K_1} = 2\mathcal{D}_{L}^{(1)}$ matrices in total. 
        Note that for $L \ge 3$, we have $2\mathcal{D}_{L}^{(1)} > \mathcal{D}_{L + 1}^{(1)}$, so matrices $\{\boldsymbol{U}_k, \boldsymbol{Y}_k, k = 1, \ldots, {K_1}\}$ are linearly dependent with real coefficients according to the first part of the proof, i.e., there exists real vectors $\boldsymbol{\alpha} = [\alpha_1, \ldots, \alpha_{K_1}]^T, \boldsymbol{\beta} = [\beta_1, \ldots, \beta_{K_1}]^T\in\mathbb{R}^{{K_1}\times 1}$ such that 
        \begin{subequations}\label{complm-eq:induction-linearly-dependent-b1}
            \begin{align}
                \boldsymbol{0}_{(L + 1)\times(L + 1)} & = \sum_{k = 1}^{{K_1}}{\alpha_k\boldsymbol{U}_k + \beta_k\boldsymbol{Y}_k} \\
                & = \sum_{k = 1}^{{K_1}}{
                    \left[
                        \begin{array}{cc}
                            (\alpha_k + \beta_k)\boldsymbol{V}_k & (\alpha_k - \beta_k)\boldsymbol{v}_k \\
                            (\alpha_k - \beta_k)\boldsymbol{v}_k^H & \alpha_k + \beta_k
                        \end{array}
                    \right]
                }. 
            \end{align}
        \end{subequations}
        Thus, we have $\sum_{k = 1}^{{K_1}}{(\alpha_k + \beta_k)\boldsymbol{V}_k} = \boldsymbol{0}$. 
        However, $\{\boldsymbol{V}_k, k = 1, \ldots, {K_1}\}$ are linearly independent with real coefficients by assumption, which leads to $\boldsymbol{\alpha} + \boldsymbol{\beta} = \boldsymbol{0}$. 
        Then, $\sum_{k = 1}^{{K_1}}(\alpha_k - \beta_k)\boldsymbol{v}_k = 2\sum_{k = 1}^{{K_1}}{\alpha_k\boldsymbol{v}_k} = 2\boldsymbol{W}\boldsymbol{\alpha} = \boldsymbol{0}$. 
        Reversely, it is easy to verify that~\eqref{complm-eq:induction-linearly-dependent-b1} holds if $\boldsymbol{\alpha} + \boldsymbol{\beta} = \boldsymbol{0}$ and $\boldsymbol{W}\boldsymbol{\alpha} = \boldsymbol{0}$. 
        Thus,~\eqref{complm-eq:induction-linearly-dependent-b1} holds if and only if 
        \begin{equation}\label{complm-eq:linear-dept-equiv-b1}
            \boldsymbol{F}\left[
                \begin{array}{c}
                    \boldsymbol{\alpha} \\
                    \boldsymbol{\beta}
                \end{array}
            \right] = \boldsymbol{0}_{(K_1 + L)\times 1}, 
            ~\boldsymbol{F} = \left[
                \begin{array}{cc}
                    \boldsymbol{I}_{K_1} & \boldsymbol{I}_{K_1} \\
                    \boldsymbol{W} & \boldsymbol{0}_{L\times K_1}
                \end{array}
            \right]. 
        \end{equation}
        Therefore, the maximum number of linearly independent matrices among $\{\boldsymbol{U}_k, \boldsymbol{Y}_k, k = 1, \ldots, {K_1}\}$ with real coefficients, denoted as $M_1$, should be equal to the number of linearly independent columns of matrix $\boldsymbol{F}\in\mathbb{R}^{(K_{1} + L)\times (2K_{1})}$, i.e., $M_{1} = \text{rank}(\boldsymbol{F})$, which can be easily verified to be $K_{1} + L$. 
        Since ${K_1} = \mathcal{D}_{L}^{(1)} = (L^2 - L) / 2 + 1$, we have $M_1 = (L^2 + L) / 2 + 1 = \mathcal{D}_{L + 1}^{(1)}$. 
        Thus, there exists $\mathcal{D}_{L + 1}^{(1)}$ linearly independent matrices with real coefficients in the set $S_V^1(L + 1)$, which completes the proof of $\text{dim}(\mathcal{S}_{X}^{1}(N)) \ge \mathcal{D}_{N}^{(1)}$.

        For $b \ge 2$, we have $\Phi_2 = \{\pm 1, \pm j\}\subseteq\Phi_b$ and $\mathcal{D}_{N}^{(b)} = \mathcal{D}_{N}^{(2)} = N^2 - N + 1$. 
        Therefore, it suffices to check the case of $b = 2$ only, i.e., showing that there always exists $\mathcal{D}_{N}^{(2)}$ matrices in $S_V^2(N)$ that are linearly independent with real coefficients. 
        Then, for any $b\ge 2$, the same matrices can be found in $\mathcal{S}_V^{b}(N)$ such that they are linearly independent with real coefficients. 

        For $b = 2$, it is easy to verify the cases where $N \le 3$. 
        Assume that for $N = L\ge 3$, there exists $K_2 = \mathcal{D}_{L}^{(2)} = L^2 - L + 1$ matrices in $S_V^2(L)$ that are linearly independent with real coefficients, which are denoted as $\boldsymbol{V}_k = \boldsymbol{v}_k\boldsymbol{v}_k^H$, $k = 1, \ldots, K_2$, with $\boldsymbol{v}_k\in{\{\pm 1, \pm j\}}^L$, $\forall k$. 
        Similar to the case of $b = 1$, we implicitly construct $\mathcal{D}_{L + 1}^{(2)}$ linearly independent matrices in $S_V^2(L + 1)$ to prove the case of $N = L + 1$. 
        Specifically, define $\boldsymbol{d}_k = (1 - j)\boldsymbol{v}_k\in\{1+j, 1-j, -1+j, -1-j\}^L$, $k = 1, \ldots, K$, then we have $\boldsymbol{V}_k = \boldsymbol{d}_k\boldsymbol{d}_k^H / 2, \forall k$. 
        Note that the real and imaginary parts of vector $\boldsymbol{d}_k$ can be considered as two binary variables independent of each other. 
        Thus, it can be decomposed into two parts, i.e., $\boldsymbol{d}_k = \boldsymbol{d}_{k, r} + j\boldsymbol{d}_{k, i}$, where $\boldsymbol{d}_{k, r}, \boldsymbol{d}_{k, i}\in\{\pm 1\}^L = \Phi_1^L$. 
        Then, define vector $\boldsymbol{q}_k$ by concatenating $\boldsymbol{d}_{k, r}$ and $\boldsymbol{d}_{k, i}$ together, i.e., 
        \begin{equation}\label{complm-def:qvec-b2}
            \boldsymbol{q}_k = \left[
                \begin{array}{c}
                    \boldsymbol{d}_{k, r} \\
                    \boldsymbol{d}_{k, i}
                \end{array}
            \right]\in\{\pm 1\}^{2L}, ~k = 1, \ldots, K_2, 
        \end{equation}
        and matrix $\boldsymbol{Q} = [\boldsymbol{q}_1, \ldots, \boldsymbol{q}_{K_2}]\in\mathbb{R}^{2L\times {K_2}}$. 
        Let $R = 2L - \text{rank}(\boldsymbol{Q}) \ge 0$. 
        If $R > 0$, i.e., $\text{rank}(\boldsymbol{Q}) < 2L$, there exists vectors $\boldsymbol{q}_{{K_2} + 1}, \ldots, \boldsymbol{q}_{{K_2} + R}\in\{\pm 1\}^{2L}$ such that $\tilde{\boldsymbol{Q}} = [\boldsymbol{Q}, \boldsymbol{q}_{{K_2} + 1}, \ldots, \boldsymbol{q}_{{K_2} + R}]\in\mathbb{R}^{2L\times({K_2} + R)}$ has full row rank $2L$. 
        Correspondingly, for $k = {K_2} + 1, \ldots, {K_2} + R$, we decompose $\boldsymbol{q}_{k}$ as $\boldsymbol{q}_{k} = [\boldsymbol{d}_{k, r}^T, \boldsymbol{d}_{k, i}^T]^T$ with $\boldsymbol{d}_{k, r}, \boldsymbol{d}_{k, i}\in\mathbb{R}^{L\times 1}$ and define $\boldsymbol{d}_k = \boldsymbol{d}_{k, r} + j\boldsymbol{d}_{k, i}$ and $\boldsymbol{v}_k = \boldsymbol{d}_k / (1 - j)\in\Phi_2^{L}$. 
        If $R = 0$, however, let $\tilde{\boldsymbol{Q}} = \boldsymbol{Q}$. 
        
        Next, for $N = L + 1$, consider $(L + 1)$-dimensional vectors
        \begin{equation}\label{complm-def:induction-construction-b2}
            \boldsymbol{u}_{k} = \left[
                \begin{array}{c}
                    \boldsymbol{v}_k \\
                    1
                \end{array}
            \right], 
            ~\boldsymbol{y}_{k} = \left[
                \begin{array}{c}
                    \boldsymbol{v}_k \\
                    -1
                \end{array}
            \right], 
            ~k = 1, \ldots, {K_2} + R, 
        \end{equation}
        and matrices $\boldsymbol{U}_k = \boldsymbol{u}_k\boldsymbol{u}_k^H$ and $\boldsymbol{Y}_k = \boldsymbol{y}_k\boldsymbol{y}_k^H$, $\forall k$. 
        Matrices $\boldsymbol{U}_k$ and $\boldsymbol{Y}_k$ are included in set $S_V^2(L + 1)$ and there are $2({K_2} + R)$ matrices in total. 
        Since $2({K_2} + R) \ge 2{K_2} > \mathcal{D}_{L + 1}^{(2)}$ for $L\ge 3$, they must be linearly dependent with real coefficients according to the first part of the proof. 
        Thus, there exists real vectors $\boldsymbol{\alpha} = [\alpha_1, \ldots, \alpha_{{K_2} + R}]^T, \boldsymbol{\beta} = [\beta_1, \ldots, \beta_{{K_2} + R}]^T\in\mathbb{R}^{({K_2} + R)\times 1}$ such that 
        \begin{subequations}\label{eq:induction-linearly-dependent-b2}
            \begin{align}
                \boldsymbol{0}_{{(L + 1)}\times {(L + 1)}} & = \sum_{k = 1}^{{K_2} + R}{\alpha_k\boldsymbol{U}_k + \beta_k\boldsymbol{Y}_k} \\
                & = \sum_{k = 1}^{{K_2} + R}{\left[
                    \begin{array}{cc}
                        (\alpha_k + \beta_k)\boldsymbol{V}_k & (\alpha_k - \beta_k)\boldsymbol{v}_k \\
                        (\alpha_k - \beta_k)\boldsymbol{v}_k^H & \alpha_k + \beta_k
                    \end{array}
                \right]}. 
            \end{align}
        \end{subequations}
        Define $\mathcal{G}_{s}$ as the set of all vectors $\boldsymbol{s}\in\mathbb{R}^{(K_2 + R)\times 1}$ such that $\sum_{k = 1}^{{K_2} + R}{s_k\boldsymbol{V}_k} = \boldsymbol{0}$ holds, then $\boldsymbol{\alpha} + \boldsymbol{\beta}\in\mathcal{G}_{s}$. 
        Obviously, $\mathcal{G}_s$ is a linear space whose dimension is bounded by $\text{dim}(\mathcal{G}_s)\le {K_2} + R - {K_2} = R$ because $\boldsymbol{V}_1, \ldots, \boldsymbol{V}_{K_2}$ are linearly independent with real coefficients. 
        Meanwhile, by defining $\boldsymbol{z} = \boldsymbol{\alpha} - \boldsymbol{\beta}$, we have $\sum_{k = 1}^{{K_2} + R}{z_k\boldsymbol{v}_k} = \boldsymbol{0}_{L\times 1}$ and thus $\sum_{k = 1}^{{K_2} + R}{z_k\boldsymbol{d}_k} = (1 - j)\sum_{{K_2} = 1}^{K_2 + R}{z_k\boldsymbol{v}_k} = \boldsymbol{0}_{L\times 1}$. 
        As $\boldsymbol{z}$ is a real vector, we have 
        \begin{equation}\label{complm-eq:linear-dept-equiv-b2}
            \tilde{\boldsymbol{Q}}\boldsymbol{z} = \sum_{k = 1}^{{K_2} + R}{z_k\left[
                \begin{array}{c}
                    \boldsymbol{d}_{k, r} \\
                    \boldsymbol{d}_{k, i}
                \end{array}
            \right]} = \sum_{k = 1}^{{K_2} + R}\left[
                \begin{array}{c}
                    {z_k\text{Re}(\boldsymbol{d}_{k})} \\
                    {z_k\text{Im}(\boldsymbol{d}_{k})}
                \end{array}
            \right] = \boldsymbol{0}_{2L\times 1}, 
        \end{equation}
        which means $\boldsymbol{z}$ is in the null space of $\tilde{\boldsymbol{Q}}$, i.e., $\mathcal{N}(\tilde{\boldsymbol{Q}})$. 
        The dimension of $\mathcal{N}(\tilde{\boldsymbol{Q}})$ is ${K_2} + R - 2L$, because $\tilde{\boldsymbol{Q}}\in\mathbb{R}^{2L\times ({K_2} + R)}$ has full row rank $2L$ by definition. 
        Reversely, one can easily verify that equation~\eqref{eq:induction-linearly-dependent-b2} is true as long as $\boldsymbol{\alpha} + \boldsymbol{\beta}\in{\mathcal{G}_s}$ and $\boldsymbol{\alpha} - \boldsymbol{\beta}\in\mathcal{N}(\tilde{\boldsymbol{Q}})$. 
        Therefore, $\sum_{k = 1}^{{K_2} + R}{\alpha_k\boldsymbol{U}_k + \beta_k\boldsymbol{Y}_k} = \boldsymbol{0}_{{K_2}\times{K_2}}$ if and only if $\boldsymbol{\alpha} + \boldsymbol{\beta}\in\mathcal{G}_s$ and $\boldsymbol{\alpha} - \boldsymbol{\beta}\in\mathcal{N}(\tilde{\boldsymbol{Q}})$, i.e., 
        \begin{equation}\label{complm-eq:linear-dept-coeff-space}
            \begin{aligned}
                \left[
                    \begin{array}{c}
                        \boldsymbol{\alpha} + \boldsymbol{\beta} \\
                        \boldsymbol{\alpha} - \boldsymbol{\beta}
                    \end{array}
                \right] \in\mathcal{G} & = \bigg\{
                    \boldsymbol{x}\in\mathbb{R}^{2({K_2} + R)\times 1}\bigg|
                        \boldsymbol{x} = \left[
                            \begin{array}{c}
                                \boldsymbol{s} \\
                                \boldsymbol{z}
                            \end{array}
                        \right], \\
                        & ~~~~~~~~~~~~ \boldsymbol{s}\in\mathcal{G}_s, \ \boldsymbol{z}\in\mathcal{N}(\tilde{\boldsymbol{Q}})
                \bigg\}. 
            \end{aligned}
        \end{equation}
        where $\mathcal{G}$ is a linear space with dimension $\text{dim}(\mathcal{G}) = \text{dim}(\mathcal{G}_s) + \text{dim}(\mathcal{N}(\tilde{\boldsymbol{Q}})) \le K_2 + 2R - 2L$. 
        Note that 
        \begin{equation}\label{complm-eq:linear-dept-coeff-transform-b2}
            \left[
                \begin{array}{c}
                    \boldsymbol{\alpha} \\
                    \boldsymbol{\beta}
                \end{array}
            \right] = \boldsymbol{F}_s\left[
                \begin{array}{c}
                    \boldsymbol{\alpha} + \boldsymbol{\beta} \\
                    \boldsymbol{\alpha} - \boldsymbol{\beta}
                \end{array}
            \right], 
            ~\boldsymbol{F}_s = \frac{1}{2}\left[
                \begin{array}{cc}
                    \boldsymbol{I} & \boldsymbol{I} \\
                    \boldsymbol{I} & -\boldsymbol{I}
                \end{array}
            \right], 
        \end{equation}
        where the transformation matrix $\boldsymbol{F}_s$ is invertible, so the dimension of the solution space of $[\boldsymbol{\alpha}^T, \boldsymbol{\beta}^T]^T$ for~\eqref{eq:induction-linearly-dependent-b2} equals to the dimension of $\mathcal{G}$. 
        Thus, the maximum number $M_2$ of linearly independent matrices among $\{\boldsymbol{U}_k, \boldsymbol{Y}_k, k = 1, \ldots, {K_2} + R\}$ with real coefficients can be obtained as
        \begin{subequations}\label{complm-eq:dim-construct-multiary}
            \begin{align}
                M_2 & = 2({K_2} + R) - \text{dim}(\mathcal{G}) \\
                & \ge 2({K_2} + R) - ({K_2} + 2R - 2L) \\
                & = {K_2} + 2L. 
            \end{align}
        \end{subequations}
        As ${K_2} = \mathcal{D}_{L}^{(2)} = L^2 - L + 1$, we have $M_2 \ge {K_2} + 2L = L^2 + L + 1 = \mathcal{D}_{L + 1}^{(2)}$. 
        Therefore, there exists at least $\mathcal{D}_{L + 1}^{(2)}$ linearly independent matrices in the set $S_V^2(L + 1)$, which completes the proof of $\text{dim}(\mathcal{S}_X^{b}(N))\ge\mathcal{D}_{N}^{(b)}$ for the case of $b = 2$, and thus the case of $b\ge 2$.  

        Combining two parts of the proof, i.e., $\text{dim}(\mathcal{S}_X^{b}(N))\le\mathcal{D}_{N}^{(b)}$ and $\text{dim}(\mathcal{S}_X^{b}(N))\ge\mathcal{D}_{N}^{(b)}$, we conclude that $\text{dim}(\mathcal{S}_X^{b}(N)) = \mathcal{D}_{N}^{(b)}$, as stated in Appendix~\ref{appendix:lemma-proof-dim-deff}. 
        Additionally, for $b\ge 2$, it has been shown that $\mathcal{S}_X^{b}(N)$ is a subspace of the space spanned by $\{\boldsymbol{B}_1; \boldsymbol{E}_{nl}^{(R)}, \boldsymbol{E}_{nl}^{(I)}, 1\le n < l \le N\}$ with real coefficients, whose dimension is $N^2 - N + 1$, which is the same as the dimension of $\mathcal{S}_X^{b}(N)$ given above. 
        Hence, $\mathcal{S}_X^{b}(N)$ equals to this spanned space and $\{\boldsymbol{B}_1; \boldsymbol{E}_{nl}^{(R)}, \boldsymbol{E}_{nl}^{(I)}, 1\le n < l \le N\}$ form an orthonormal basis of $\mathcal{S}_X^{b}(N)$ for $b\ge 2$. 
        Similarly, an orthonormal basis of $\mathcal{S}_X^{b}(N)$ for $b = 1$ is given by $\{\boldsymbol{B}_1; \boldsymbol{E}_{nl}^{(R)}, 1\le n < l \le N\}$.

\section{Complementary Proof of Proposition~\ref{prop:cov-est-existence-uniqueness} in Appendix~\ref{appendix:prop-proof-uniqueness}}\label{complm-sec:prop1-proof}
        In this section, the uniqueness of the solution for problem~\eqref{prob:cov-est-find-origin} demonstrated in Proposition~\ref{prop:cov-est-existence-uniqueness} is proved in details. 
        First we consider the case of $b\ge 2$ and then $b = 1$. 

        For $b\ge 2$, suppose there exists a matrix $\hat{\boldsymbol{H}}\neq\bar{\boldsymbol{H}}$ that is also a solution to Problem~\eqref{prob:cov-est-find-origin}. 
        Then, $\hat{\boldsymbol{H}}$ is a rank-one hermitian matrix and we have $\text{tr}(\hat{\boldsymbol{H}}\boldsymbol{V}_t) = p_t = \text{tr}(\bar{\boldsymbol{H}}\boldsymbol{V}_t)$, $t = 1, \ldots, T_{p}$. 
        By defining $\hat{\boldsymbol{H}} = \hat{\boldsymbol{h}}\hat{\boldsymbol{h}}^H$ and $\mathcal{\boldsymbol{E}} = \bar{\boldsymbol{H}} - \hat{\boldsymbol{H}}\in\mathbb{C}^{N\times N}$, we have $\text{rank}(\mathcal{\boldsymbol{E}})\le 2$ and $\text{tr}(\mathcal{\boldsymbol{E}}\boldsymbol{V}_t) = 0$ for $\forall t$. 
        If $D_V = \mathcal{D}_{N}^{(b)} = \text{dim}(\mathcal{S}_X^{b}(N))$, the matrix $\mathcal{\boldsymbol{E}}$ is perpendicular to the space $\mathcal{S}_X^{b}(N)$, which results in $\text{tr}(\mathcal{\boldsymbol{E}}\boldsymbol{E}_{nl}^{(R)}) = \text{tr}(\mathcal{\boldsymbol{E}}\boldsymbol{E}_{nl}^{(I)}) = 0$, $1\le n < l\le N$, and also $\text{tr}(\mathcal{\boldsymbol{E}}\boldsymbol{B}_{1}) = \text{tr}(\mathcal{\boldsymbol{E}})/\sqrt{N} = 0$. 
        According to the definitions of matrices $\boldsymbol{E}_{nl}^{(R)}, \boldsymbol{E}_{nl}^{(I)}, 1\le n < l\le N$, all non-diagonal entries of $\mathcal{\boldsymbol{E}}$ are zero and thus we have 
        \begin{equation}\label{eq:sol-err-multiary}
            \begin{aligned}
                \mathcal{\boldsymbol{E}} & = 
                    \left[
                        \begin{array}{cccc}
                            \mathcal{E}_{11} & 0 & \ldots & 0 \\
                            0 & \mathcal{E}_{22} & \ldots & 0 \\
                            \vdots & \vdots & \ddots & 0 \\
                            0 & 0 & 0 & \mathcal{E}_{NN}
                        \end{array}
                    \right], 
            \end{aligned}
        \end{equation}
        where $\mathcal{E}_{nn} = |\bar{h}_n|^2 - |\hat{h}_n|^2$, $\forall n$. 
        Due to $\text{rank}(\mathcal{\boldsymbol{E}})\le 2$, we assume $\mathcal{E}_{nn} = 0$ for $n\ge 3$ without loss of generality, and thus $|\bar{h}_n| = |\hat{h}_n|$ for $n\ge 3$. 
        By assuming $N\ge 3$, we have $\mathcal{E}_{n1} = \bar{h}_n\bar{h}_1^* - \hat{h}_n\hat{h}_1^* = 0$, $n = 3, \ldots, N$, leading to $|\bar{h}_n\bar{h}_1| = |\hat{h}_n\hat{h}_1|$ for $\forall n\ge 3$. 
        Since $\bar{h}_n \neq 0$ holds with probability $1$ for any $n$, we find that $|\bar{h}_1| = |\hat{h}_1|$. 
        Similarly, $|\bar{h}_2| = |\hat{h}_2|$ can be obtained. 
        Then, $\mathcal{E}_{11} = \mathcal{E}_{22} = 0$ and thus $\mathcal{\boldsymbol{E}} = \boldsymbol{0}$, which indicates $\bar{\boldsymbol{H}} = \hat{\boldsymbol{H}}$. 
        Hence, $\bar{\boldsymbol{H}}$ is the unique solution to Problem~\eqref{prob:cov-est-find-origin} with probability $1$ for $b\ge 2$ when $N\ge 3$ and $D_V = \mathcal{D}_{N}^{(b)}$. 

        For $b = 1$, the IRS reflection vectors are all real vectors and it is easy to verify that $\bar{\boldsymbol{H}}$ and $\bar{\boldsymbol{H}}^*$ are both solutions for problem~\eqref{prob:cov-est-find-origin}. 
        Suppose there exists a matrix $\hat{\boldsymbol{H}}$ that does not equal to $\bar{\boldsymbol{H}}$ or $\bar{\boldsymbol{H}}^*$ but is also a solution for problem~\eqref{prob:cov-est-find-origin}, and its real part is denoted as $\hat{\boldsymbol{H}}_r = \text{Re}(\hat{\boldsymbol{H}})$. 
        Note that the ranks of both $\bar{\boldsymbol{H}}_r = \text{Re}(\bar{\boldsymbol{H}})$ and $\hat{\boldsymbol{H}}_r$ are no more than two and $\text{tr}(\bar{\boldsymbol{H}}_r\boldsymbol{V}_t) = \text{tr}(\bar{\boldsymbol{H}}\boldsymbol{V}_t) = p_t = \text{tr}(\hat{\boldsymbol{H}}\boldsymbol{V}_t) = \text{tr}(\hat{\boldsymbol{H}}_r\boldsymbol{V}_t)$. 
        Similar to the case where $b\ge 2$, we define $\mathcal{\boldsymbol{E}}_r = \bar{\boldsymbol{H}}_r - \hat{\boldsymbol{H}}_r\in\mathbb{R}^{N\times N}$, which is a real symmetric matrix satisfying $\text{rank}(\mathcal{\boldsymbol{E}}_r)\le 4$ and $\text{tr}(\mathcal{\boldsymbol{E}}_r\boldsymbol{V}_t) = 0$, $t = 1, \ldots, T_{p}$. 
        When $D_V = \mathcal{D}_N^{(1)} = \text{dim}(\mathcal{S}_X^{1}(N))$, the matrix $\mathcal{\boldsymbol{E}}_r$ is perpendicular to the space $\mathcal{S}_X^{b}(N)$ and thus $\text{tr}(\mathcal{\boldsymbol{E}}_r\boldsymbol{E}_{nl}^{(R)}) = 0$, $1\le n < l\le N$. 
        As a result, $\mathcal{\boldsymbol{E}}_r$ can be written as 
        \begin{equation}\label{eq:sol-err-binary}
            \begin{aligned}
                \mathcal{\boldsymbol{E}}_r & = 
                    \left[
                        \begin{array}{cccc}
                            \mathcal{E}_{r, 11} & 0 & \ldots & 0 \\
                            0 & \mathcal{E}_{r, 22} & \ldots & 0 \\
                            \vdots & \vdots & \ddots & 0 \\
                            0 & 0 & 0 & \mathcal{E}_{r, NN}
                        \end{array}
                    \right], 
            \end{aligned}
        \end{equation}
        where $\mathcal{E}_{r, nn}$ is the $n$-th diagonal element of $\mathcal{\boldsymbol{E}}_r$. 
        Without loss of generality, we have $\mathcal{E}_{r, nn} = 0$ for $n\ge 5$ due to $\text{rank}(\mathcal{\boldsymbol{E}}_r)\le 4$. 
        Since $\bar{\boldsymbol{H}} = \bar{\boldsymbol{h}}\bar{\boldsymbol{h}}^H$, it is easy to verify that $\bar{\boldsymbol{H}}_r = \bar{\boldsymbol{h}}_{r}\bar{\boldsymbol{h}}_{r}^T + \bar{\boldsymbol{h}}_{m}\bar{\boldsymbol{h}}_{m}^T$, where $\bar{\boldsymbol{h}}_{r} = \text{Re}(\bar{\boldsymbol{h}})$ and $\bar{\boldsymbol{h}}_{m} = \text{Im}(\bar{\boldsymbol{h}})$. 
        Similarly, by defining $\hat{\boldsymbol{H}} = \hat{\boldsymbol{h}}\hat{\boldsymbol{h}}^H$, we have $\hat{\boldsymbol{H}}_r = \hat{\boldsymbol{h}}_{r}\hat{\boldsymbol{h}}_{r}^T + \hat{\boldsymbol{h}}_{m}\hat{\boldsymbol{h}}_{m}^T$, where $\hat{\boldsymbol{h}}_{r} = \text{Re}(\hat{\boldsymbol{h}})$ and $\hat{\boldsymbol{h}}_{m} = \text{Im}(\hat{\boldsymbol{h}})$. 
        Consider real matrices
        \begin{subequations}\label{def:csi-seg-binary}
            \begin{align}
                \bar{\boldsymbol{A}} & = 
                    \left[
                        \begin{array}{cc}
                            \bar{h}_{r, 1} & \bar{h}_{m, 1} \\
                            \vdots    & \vdots    \\
                            \bar{h}_{r, 4} & \bar{h}_{m, 4}
                        \end{array}
                    \right], 
                & \hat{\boldsymbol{A}} & = 
                \left[
                    \begin{array}{cc}
                        \hat{h}_{r, 1} & \hat{h}_{m, 1} \\
                        \vdots    & \vdots    \\
                        \hat{h}_{r, 4} & \hat{h}_{m, 4}
                    \end{array}
                \right], \\
                \bar{\boldsymbol{B}} & = 
                    \left[
                        \begin{array}{cc}
                            \bar{h}_{r, 5} & \bar{h}_{m, 5} \\
                            \vdots    & \vdots    \\
                            \bar{h}_{r, N} & \bar{h}_{m, N}
                        \end{array}
                    \right], 
                & \hat{\boldsymbol{B}} & = 
                \left[
                    \begin{array}{cc}
                        \hat{h}_{r, 5} & \hat{h}_{m, 5} \\
                        \vdots    & \vdots    \\
                        \hat{h}_{r, N} & \hat{h}_{m, N}
                    \end{array}
                \right], 
            \end{align}
        \end{subequations}
        where $\bar{h}_{r, n}$, $\bar{h}_{m, n}$, $\hat{h}_{r, n}$ and $\hat{h}_{m, n}$ are the $n$-th elements of vectors $\bar{\boldsymbol{h}}_{r}$, $\bar{\boldsymbol{h}}_{m}$, $\hat{\boldsymbol{h}}_{r}$ and $\hat{\boldsymbol{h}}_{m}$, respectively. 
        Then, the matrix $\mathcal{\boldsymbol{E}}_r$ can be equivalently written as 
        \begin{subequations}\label{eq:err-mat-block-binary}
            \allowdisplaybreaks
            \begin{align}
                \mathcal{\boldsymbol{E}}_r & = \bar{\boldsymbol{H}}_r - \hat{\boldsymbol{H}}_r \\
                & = 
                \left[
                    \begin{array}{c}
                        \bar{\boldsymbol{A}} \\
                        \bar{\boldsymbol{B}}
                    \end{array}
                \right]\left[
                    \bar{\boldsymbol{A}}^T, \bar{\boldsymbol{B}}^T
                \right] - \left[
                    \begin{array}{c}
                        \hat{\boldsymbol{A}} \\
                        \hat{\boldsymbol{B}}
                    \end{array}
                \right]\left[
                    \hat{\boldsymbol{A}}^T, \hat{\boldsymbol{B}}^T
                \right] \\
                & = \left[
                    \begin{array}{cc}
                        \bar{\boldsymbol{A}}\bar{\boldsymbol{A}}^T - \hat{\boldsymbol{A}}\hat{\boldsymbol{A}}^T & \bar{\boldsymbol{A}}\bar{\boldsymbol{B}}^T - \hat{\boldsymbol{A}}\hat{\boldsymbol{B}}^T \\
                        \bar{\boldsymbol{B}}\bar{\boldsymbol{A}}^T - \hat{\boldsymbol{B}}\hat{\boldsymbol{A}}^T & \bar{\boldsymbol{B}}\bar{\boldsymbol{B}}^T - \hat{\boldsymbol{B}}\hat{\boldsymbol{B}}^T
                    \end{array}
                \right], \label{eq:err-block-binary}
            \end{align}
        \end{subequations}
        where $\bar{\boldsymbol{A}}\bar{\boldsymbol{A}}^T - \hat{\boldsymbol{A}}\hat{\boldsymbol{A}}^T$ is a $4\times 4$ diagonal matrix and all other three blocks in~\eqref{eq:err-block-binary} are zeros, according to equation~\eqref{eq:sol-err-binary} and the analysis above. 
        Thus, we have $\bar{\boldsymbol{A}}\bar{\boldsymbol{B}}^T = \hat{\boldsymbol{A}}\hat{\boldsymbol{B}}^T$ and $\bar{\boldsymbol{B}}\bar{\boldsymbol{B}}^T = \hat{\boldsymbol{B}}\hat{\boldsymbol{B}}^T$. 
        Obviously, the space spanned by the columns of $\hat{\boldsymbol{B}}$ is the same as that of $\bar{\boldsymbol{B}}$, which means the columns of $\hat{\boldsymbol{B}}$ can be represented by $\bar{\boldsymbol{B}}$ with a coefficient matrix $\boldsymbol{Y}\in\mathbb{R}^{2\times 2}$, i.e., $\hat{\boldsymbol{B}} = \bar{\boldsymbol{B}}\boldsymbol{Y}$. 
        Thus, $\bar{\boldsymbol{B}}\bar{\boldsymbol{B}}^T = \hat{\boldsymbol{B}}\hat{\boldsymbol{B}}^T = \bar{\boldsymbol{B}}\boldsymbol{Y}\boldsymbol{Y}^T\bar{\boldsymbol{B}}^T$. 
        Since $\bar{\boldsymbol{h}}$ represents the equivalent channel from the BS to user, the matrix $\bar{\boldsymbol{B}}$ has full column rank $2$ with probability $1$ for $N\ge 6$, and thus $\bar{\boldsymbol{B}}^T\bar{\boldsymbol{B}}\in\mathbb{R}^{2\times 2}$ is invertible. 
        Hence, $\boldsymbol{Y}\boldsymbol{Y}^T$ can be solved as
        \begin{equation}
            \boldsymbol{Y}\boldsymbol{Y}^T = \left(\bar{\boldsymbol{B}}^T\bar{\boldsymbol{B}}\right)^{-1}\bar{\boldsymbol{B}}^T\left(\bar{\boldsymbol{B}}\bar{\boldsymbol{B}}^T\right)\bar{\boldsymbol{B}}\left(\bar{\boldsymbol{B}}^T\bar{\boldsymbol{B}}\right)^{-1} = \boldsymbol{I}, 
        \end{equation}
        which means that $\boldsymbol{Y}\in\mathbb{R}^{2\times 2}$ is an orthogonal matrix. 
        On the other hand, we have $\bar{\boldsymbol{A}}\bar{\boldsymbol{B}}^T = \hat{\boldsymbol{A}}\hat{\boldsymbol{B}}^T = \hat{\boldsymbol{A}}\boldsymbol{Y}^T\bar{\boldsymbol{B}}^T$. 
        By applying $\bar{\boldsymbol{B}}$ on the left to both sides, it can be simplified as $\bar{\boldsymbol{A}} = \hat{\boldsymbol{A}}\boldsymbol{Y}^T$, or equivalently, $\hat{\boldsymbol{A}} = \bar{\boldsymbol{A}}\boldsymbol{Y}$. 
        Then, $\bar{\boldsymbol{A}}\bar{\boldsymbol{A}}^T = \hat{\boldsymbol{A}}\boldsymbol{Y}^T\boldsymbol{Y}\hat{\boldsymbol{A}}^T = \hat{\boldsymbol{A}}\hat{\boldsymbol{A}}^T$ holds and $\mathcal{\boldsymbol{E}}_r = \boldsymbol{0}$ follows, which leads to $\bar{\boldsymbol{H}}_r = \hat{\boldsymbol{H}}_r$. 
        
        Meanwhile, it is worth noting that
        \begin{equation}
            \bar{\boldsymbol{h}} = \left[
                \bar{\boldsymbol{h}}_{r}, \bar{\boldsymbol{h}}_{m}
            \right]\left[
                \begin{array}{c}
                    1 \\
                    j
                \end{array}
            \right] = \left[
                \begin{array}{c}
                    \bar{\boldsymbol{A}} \\
                    \bar{\boldsymbol{B}}
                \end{array}
            \right]\left[
                \begin{array}{c}
                    1 \\
                    j
                \end{array}
            \right]. 
        \end{equation}
        Therefore, the matrix $\bar{\boldsymbol{H}}$ can be written as 
        \begin{equation}
            \bar{\boldsymbol{H}} = \bar{\boldsymbol{h}}\bar{\boldsymbol{h}}^H = \left[
                \begin{array}{c}
                    \bar{\boldsymbol{A}} \\
                    \bar{\boldsymbol{B}}
                \end{array}
            \right]\left[
                \begin{array}{cc}
                    1 & -j \\
                    j &  1
                \end{array}
            \right]\left[
                \bar{\boldsymbol{A}}^T, \bar{\boldsymbol{B}}^T
            \right]. 
        \end{equation}
        Similarly, $\hat{\boldsymbol{H}}$ can be equivalently written as
        \begin{subequations}
            \begin{align}
                \hat{\boldsymbol{H}} & = \left[
                    \begin{array}{c}
                        \hat{\boldsymbol{A}} \\
                        \hat{\boldsymbol{B}}
                    \end{array}
                \right]\left[
                    \begin{array}{cc}
                        1 & -j \\
                        j &  1
                    \end{array}
                \right]\left[
                        \hat{\boldsymbol{A}}^T, \hat{\boldsymbol{B}}^T
                \right] \\
                & = \left[
                    \begin{array}{c}
                        \bar{\boldsymbol{A}} \\
                        \bar{\boldsymbol{B}}
                    \end{array}
                \right]\boldsymbol{Y}\left[
                    \begin{array}{cc}
                        1 & -j \\
                        j &  1
                    \end{array}
                \right]\boldsymbol{Y}^T\left[
                    \bar{\boldsymbol{A}}^T, \bar{\boldsymbol{B}}^T
                \right]. 
            \end{align}
        \end{subequations}
        Define matrix $\boldsymbol{J}\in\mathbb{R}^{2\times 2}$ as 
        \begin{equation}
            \boldsymbol{J} = \boldsymbol{Y}\left[
                    \begin{array}{cc}
                        1 & -j \\
                        j &  1
                    \end{array}
                \right]\boldsymbol{Y}^T. 
        \end{equation}
        Let $\boldsymbol{Y} = [\boldsymbol{y}_1, \boldsymbol{y}_2]$ with $\boldsymbol{y}_1 = [y_{11}, y_{12}]^T$ and $\boldsymbol{y}_2 = [y_{21}, y_{22}]^T$. 
        Then, we have $\boldsymbol{y}_1\boldsymbol{y}_1^T + \boldsymbol{y}_2\boldsymbol{y}_2^T = \boldsymbol{Y}\boldsymbol{Y}^T = \boldsymbol{I}$ and
        \begin{subequations}
            \begin{align}
                \boldsymbol{J} & = \left(\boldsymbol{y}_1 + j\boldsymbol{y}_2\right)\left(\boldsymbol{y}_1^T - j\boldsymbol{y}_2^T\right) \\
                & = \left(\boldsymbol{y}_1\boldsymbol{y}_1^T + \boldsymbol{y}_2\boldsymbol{y}_2^T\right) + j\left(\boldsymbol{y}_2\boldsymbol{y}_1^T - \boldsymbol{y}_1\boldsymbol{y}_2^T\right) \\
                & = \boldsymbol{I} + j\left[
                    \begin{array}{cc}
                        0 & -\eta \\
                        \eta &  0
                    \end{array}
                \right], 
            \end{align}
        \end{subequations}
        where $\eta = y_{11}y_{22} - y_{12}y_{21} = \text{det}(\boldsymbol{Y}) = \pm 1$ because $\boldsymbol{Y}$ is an orthogonal matrix. 
        If $\text{det}(\boldsymbol{Y}) = 1$, we have 
        \begin{equation}
            \boldsymbol{J} = \left[
                \begin{array}{cc}
                    1 & -j \\
                    j &  1
                \end{array}
            \right], \ \hat{\boldsymbol{H}} = \left[
                    \begin{array}{c}
                        \bar{\boldsymbol{A}} \\
                        \bar{\boldsymbol{B}}
                    \end{array}
                \right]\boldsymbol{J}\left[
                    \bar{\boldsymbol{A}}^T, \bar{\boldsymbol{B}}^T
                \right] = \bar{\boldsymbol{H}}. 
        \end{equation}
        If $\text{det}(\boldsymbol{Y}) = -1$, however, $\boldsymbol{J}$ and $\hat{\boldsymbol{H}}$ are given as 
        \begin{equation}
            \boldsymbol{J} = \left[
                \begin{array}{cc}
                    1 & j \\
                    -j &  1
                \end{array}
            \right], \ \hat{\boldsymbol{H}} = \left[
                    \begin{array}{c}
                        \bar{\boldsymbol{A}} \\
                        \bar{\boldsymbol{B}}
                    \end{array}
                \right]\boldsymbol{J}\left[
                    \bar{\boldsymbol{A}}^T, \bar{\boldsymbol{B}}^T
                \right] = \bar{\boldsymbol{H}}^*. 
        \end{equation}
        Therefore, $\bar{\boldsymbol{H}}$ and $\bar{\boldsymbol{H}}^*$ are the only two solutions to Problem~\eqref{prob:cov-est-find-origin} with probability $1$ for $b = 1$ when $N\ge 6$ and $D_V = \mathcal{D}_{N}^{(1)}$. 
        Hereto, the proof for the existence and uniqueness of the solutions to Problem~\eqref{prob:cov-est-find-origin} claimed in Proposition~\ref{prop:cov-est-existence-uniqueness} have been completed.


\bibliographystyle{IEEEtran} 
\bibliography{IEEEabrv, reference}

\end{document}